\documentclass[sigconf]{aamas} 

\usepackage{balance} 
\usepackage{bm}

\usepackage[normalem]{ulem}
\usepackage{xcolor}
\usepackage{bbm}
\usepackage[linesnumbered,ruled,vlined]{algorithm2e}

\usepackage{pgfplots}
\DeclareUnicodeCharacter{2212}{−}
\usepgfplotslibrary{groupplots, dateplot}
\usetikzlibrary{patterns,shapes.arrows}
\pgfplotsset{compat=newest}

\usepackage{tikz}

\usepackage[outline]{contour}

\usepackage{caption}
\usepackage{subcaption}

\makeatletter
\def\blfootnote{\xdef\@thefnmark{}\@footnotetext}
\makeatother

\setcopyright{ifaamas}
\acmConference[AAMAS '22]{Proc.\@ of the 21st International Conference
on Autonomous Agents and Multiagent Systems (AAMAS 2022)}{May 9--13, 2022}
{Online}{P.~Faliszewski, V.~Mascardi, C.~Pelachaud,
M.E.~Taylor (eds.)}
\copyrightyear{2022}
\acmYear{2022}
\acmDOI{}
\acmPrice{}
\acmISBN{}

\acmSubmissionID{138}

\title{Planning Not to Talk: Multiagent Systems\\that are Robust to Communication Loss}

\author{Mustafa O. Karabag*}
\affiliation{
  \institution{The University of Texas at Austin}
    \state{Texas}
  \country{USA}
}
\email{karabag@utexas.edu}

\author{Cyrus Neary*}
\affiliation{
  \institution{The University of Texas at Austin}
    \state{Texas}
  \country{USA}
}
\email{cneary@utexas.edu}

\author{Ufuk Topcu}
\affiliation{
  \institution{The University of Texas at Austin}
    \state{Texas}
  \country{USA}
}
\email{utopcu@utexas.edu}


\newcommand{\BibTeX}{\rm B\kern-.05em{\sc i\kern-.025em b}\kern-.08em\TeX}

\begin{abstract}

In a cooperative multiagent system, a collection of agents executes a joint policy in order to achieve some common objective.
The successful deployment of such systems hinges on the availability of reliable inter-agent communication.
However, many sources of potential disruption to communication exist in practice, such as radio interference, hardware failure, and adversarial attacks. 
In this work, we develop joint policies for cooperative multiagent systems that are robust to potential losses in communication.
More specifically, we develop joint policies for cooperative Markov games with reach-avoid objectives.
First, we propose an algorithm for the decentralized execution of joint policies during periods of communication loss. 
Next, we use the total correlation of the state-action process induced by a joint policy as a measure of the intrinsic dependencies between the agents.
We then use this measure to lower-bound the performance of a joint policy when communication is lost.
Finally, we present an algorithm that maximizes a proxy to this lower bound in order to synthesize minimum-dependency joint policies that are robust to communication loss.
Numerical experiments show that the proposed minimum-dependency policies require minimal coordination between the agents while incurring little to no loss in performance; 
the total correlation value of the synthesized policy is one fifth of the total correlation value of the baseline policy which does not take potential communication losses into account.
As a result, the performance of the minimum-dependency policies remains consistently high regardless of whether or not communication is available. 
By contrast, the performance of the baseline policy decreases by twenty percent when communication is lost.
\end{abstract}

\begin{document}
\pagestyle{fancy}
\fancyhead{}
\maketitle 

\newcommand{\expectation}{\mathbb{E}}
\newcommand{\kl}{KL}
\newcommand{\entropy}{H}
\newcommand{\distribution}{\Delta}
\newcommand{\probabilityMeasure}{\mu}
\newcommand{\genericRandomVar}{Y}
\newcommand{\genericRandomVarSupport}{\mathcal{\genericRandomVar}}
\newcommand{\genericDistribution}{Q}
\newcommand{\genericDistributionSupport}{\mathcal{\genericDistribution}}
\newcommand{\genericFunction}{f}
\newcommand{\genericFunctionAlt}{g}
\newcommand{\emptyString}{\varepsilon}
\newcommand{\epsilonTransition}{\epsilon}
\newcommand{\genericString}{w}
\newcommand{\constantNumber}{K}
\newcommand{\genericSet}{V}

\newcommand{\mdp}{\mathcal{M}}
\newcommand{\mdpState}{s}
\newcommand{\mdpStateAlt}{y}
\newcommand{\mdpInitialState}{\mdpState_{I}}
\newcommand{\mdpStateSpace}{\mathcal{S}}
\newcommand{\mdpAction}{a}
\newcommand{\mdpActionAlt}{b}
\newcommand{\mdpActionSpace}{\mathcal{A}}
\newcommand{\mdpReward}{\mathcal{R}}
\newcommand{\mdpTransition}{\mathcal{T}}

\newcommand{\policy}{\pi}

\newcommand{\mdpPath}{\xi}
\newcommand{\mdpPathDist}{\Gamma}
\newcommand{\mdpValue}{v}
\newcommand{\mdpStateActionProcess}{X}
\newcommand{\mdpStateActionProcessAlt}{Y}
\newcommand{\mdpMixedStateActionProcess}{\bar{\mdpStateActionProcess}}
\newcommand{\mdpStateRandomVar}{S}
\newcommand{\mdpActionRandomVar}{A}

\newcommand{\timeHorizon}{\mathcal{T}}

\newcommand{\joint}{joint}
\newcommand{\fullcommunication}{full}
\newcommand{\fullyimaginary}{full\text{ }img}
\newcommand{\imaginary}{img}
\newcommand{\intermittent}{int}

\newcommand{\game}{\bm{\mdp}}
\newcommand{\gameState}{\bm{\mdpState}}
\newcommand{\gameStateAlt}{\bm{\mdpStateAlt}}
\newcommand{\gameActionAlt}{\bm{\mdpActionAlt}}
\newcommand{\gameInitialState}{\bm{\mdpInitialState}}
\newcommand{\gameStateSpace}{\bm{\mdpStateSpace}}
\newcommand{\gameAction}{\bm{\mdpAction}}
\newcommand{\gameActionSpace}{\bm{\mdpActionSpace}}
\newcommand{\gameTransition}{\bm{\mdpTransition}}
\newcommand{\gameReward}{\bm{\mdpReward}}
\newcommand{\gameStateRandomVar}{\bm{\mdpStateRandomVar}}
\newcommand{\gameActionRandomVar}{\bm{\mdpActionRandomVar}}
\newcommand{\gameStateActionProcessAlt}{\bm{\mdpStateActionProcessAlt}}
\newcommand{\len}{len}
\newcommand{\expectedLength}{l}

\newcommand{\jointPolicy}{\policy_{\joint}}

\newcommand{\targetSet}{\gameStateSpace_{\mathcal{T}}}

\newcommand{\deadSet}{\gameStateSpace_{\mathcal{A}}}
\newcommand{\deadSetPrime}{\gameStateSpace_{\mathcal{D}}}
\newcommand{\doneSet}{\gameStateSpace_{\mathcal{E}}}
\newcommand{\gameProcessEndState}{\gameState_{\epsilonTransition}}
\newcommand{\mdpProcessEndState}{\mdpState_{\epsilonTransition}}

\newcommand{\gamePath}{\bm{\mdpPath}}
\newcommand{\gamePathDist}{\bm{\mdpPathDist}}
\newcommand{\gameValue}{\bm{\mdpValue}}
\newcommand{\gameStateActionProcess}{\bm{\mdpStateActionProcess}}
\newcommand{\totalCorrelation}{C}
\newcommand{\totalCorrelationUpperBound}{\bar{\totalCorrelation}}

\newcommand{\numAgents}{N}

\newcommand{\probabilityFailureForever}{p}
\newcommand{\probabilityFailureOneStep}{q}
\newcommand{\sequenceCommAvailibility}{\Lambda}
\newcommand{\oneStepCommAvailibility}{\lambda}

\newcommand{\expectedLengthCoef}{\delta}
\newcommand{\totalCorrelationCoef}{\beta}
\newcommand{\occupancyVar}{x}

\newcommand{\rover}{R}
\newcommand{\robot}{R}
\newcommand{\agent}{R}
\newcommand{\base}{B}
\newcommand{\goal}{T}
\blfootnote{*These authors contributed equally to this work.}
\section{Introduction}

In a cooperative multiagent systems, a team of decision-making agents aims to achieve a common objective through repeated interactions with each other and with a shared environment.
Such multiagent systems are ubiquitous; many applications of autonomous systems --- such as the coordination of autonomous vehicles, the control of networks of mobile sensors, or the control of traffic lights --- can be modeled as collections of interacting agents \cite{cao2012overview, parker2016multiple}. 

Inter-agent communication plays an essential role in the successful deployment of such multiagent systems.
In particular, the coordination between agents via communication -- their agreement upon the particular actions to collectively take at any given point in time -- is often necessary for the successful implementation of an optimal joint policy \cite{boutilier1996planning}.
However, many possible sources of communication disruption exist in practice, such as radio interference, hardware failure, or even adversarial attacks intended to sabotage the team.
Lost or unreliable communication can result in substantial degradation of the team's performance, because it removes the agents' ability to coordinate. 
Despite this reliance of the team's performance on communication, multiagent learning and planning algorithms typically do not offer robustness guarantees against possible losses in communication.

In this work, we study multiagent systems that are robust to such communication losses.
In detail, we study sequential multiagent decision problems formulated as cooperative Markov games with reach-avoid objectives \cite{littman1994markov, baier2008principles}.
We make the following contributions.

    \textit{(1) Imaginary Play for Policy Execution During Intermittent Communication.}
    During periods of lost communication, each agent maintains imaginary versions of their teammates' states and actions using the pre-agreed-upon joint policy and a model of the environment's stochastic dynamics.
    By maintaining such imaginary copies of their teammates, each agent may act according to a model of how their teammates are likely to behave, without receiving any communicated information from them.
    Once communication is re-established the agents share updates, correct their imaginary models, and proceed with policy execution as normal until communication is lost again, or until the team's task is complete.

    \textit{(2) Theoretical Results: Total Correlation as a Measure of Policy Robustness to Communication Loss.}
    We use the total correlation~\cite{watanabe1960information} -- a generalization of the mutual information -- of the stochastic state-action process induced by the joint policy as a measure of how reliant that particular policy is on communication.
    To relate this measure to the performance of the policy, we provide lower bounds on the value function achieved during intermittent communication, in terms of the total correlation of the policy and the value function it achieves when communication is available.
    In addition to the policy synthesis algorithm described below, this lower bound provides a means to select communication resources that are sufficient to achieve a particular performance while using noisy communication channels.

    \textit{(3) Synthesis of Policies Robust to Intermittent Communication.}
    To synthesize \textit{minimum-dependency policies} that remain performant under intermittent communication, we present an algorithm that maximizes a proxy to the lower bound described above.
    This optimization problem is formulated as a difference of convex terms.
    We solve for local optima using the convex-concave procedure \cite{yuille2002concave}.

Numerical results empirically demonstrate the effectiveness of the proposed algorithms for communication-free policy execution and for the synthesis of minimum-dependency joint policies.
When communication is not restricted, the synthesized minimum-dependency policies enjoy task performance that is similar to a baseline policy that does not take potential communication losses into account.
However, the minimum-dependency policies require minimal coordination between agents; the total correlation value of their joint state-action processes is roughly one fifth of the total correlation value of the process induced by the baseline policy.

As a result, the performance of the minimum-dependency policies remain constant, even when communication between agents is restricted to be entirely unavailable. 
By contrast, we observe a twenty percent degradation in the performance of the baseline policy when communication is lost.

\paragraph{Outline.}
In \S \ref{sec:related_work} we discuss related work.
In \S \ref{sec:preliminaries} we introduce preliminary background material as well as the notation used throughout the paper. 
We present our problem statement and an illustrative running example in \S \ref{sec:multiagent_planning_and_comms}. 
The proposed algorithms for policy execution during communication losses are presented in \S \ref{sec:imaginary_play}.
The paper's theoretical results and their implications are discussed in \S\ref{sec:measuring_dependencies} and \S \ref{sec:theoretical_results}.
We leave the details of the proofs of these theoretical results to the supplementary material.
In \S \ref{sec:policy_synthesis} we present the proposed formulation and solution to the policy synthesis problem, before presenting the experimental results in \S \ref{sec:experiments}.

\section{Related Work}

\label{sec:related_work}

Multiagent decision-making problems have been formulated using several models, e.g., multiagent Markov decision processes (MDPs) \cite{boutilier1996planning}.
Our problem setting, in which each agent has independent transitions and may only observe their own local state, is most similar to transition-independent decentralized MDPs (Dec-MDPs) \cite{becker2003transition}.
However, while this work considers the fully decentralized setting -- the agents cannot communicate at all -- we consider the setting in which communication is allowed but unreliable.
We note that Dec-MDPs are a special case of decentralized partially observable MDPs (Dec-POMDPs) \cite{oliehoek2016concise}, which are notoriously difficult to solve in general when the agents cannot communicate.
In fact, even policy synthesis for finite-horizon transition-independent Dec-MDPs without communication is NP-complete \cite{goldman2004decentralized}.

Prior work for multiagent systems considers imposing specific communication structures between the agents, either as a dependency graph \cite{guestrin2002multiagent}, or as a subset of joint states at which the agents may communicate \cite{melo2011decentralized}.
In addition to these fixed communication structures, the papers \cite{becker2009analyzing, wu2011online} consider communication as an explicit action that can be taken by the agents, leading to dynamic communication structures that change over time.
While all of the above works consider synthesizing optimal behavior according to specific communication structures, 
our work studies multiagent systems that are robust to unpredictable communication losses.

To render the multiagent systems robust to communication loss, our work aims to minimize intrinsic dependencies between the agents.
As a measure of such dependencies, we use the total correlation \cite{watanabe1960information} -- an information theoretic measure -- of the state-action process induced by the joint policy.
Information theoretic measures have been studied in single-agent MDPs \cite{savas2019entropy, leibfried2020mutual, tanaka2021transfer, eysenbach2021robust}.
In particular, \cite{tanaka2021transfer} synthesizes single-agent policies that minimize the transfer entropy from the state process to the action process with the purpose of minimizing the reliance of the policy on the underlying state process.
By contrast, our work considers a multiagent setting and introduces information theoretic measures with the specific purpose of providing guarantees on the performance of the team under communication loss.
The paper \cite{eysenbach2021robust} proposes to minimize the mutual information between the underlying state process and the agent's action process in the context of single-agent reinforcement learning.  
By contrast, we study the multiagent setting and provide bounds on the performance of the entire team, when the agents have intermittent communication.
Furthermore, we provide an optimization problem to synthesize joint policies that are robust to communication loss.
In the multiagent reinforcement learning setting, \cite{wang2020learning} consider minimizing the mutual information between the state processes and the messages shared between the agents, but do not provide theoretical result on the performance of the team when communication is only intermittently available.

The centralized training decentralized execution paradigm in has recently drawn attention in multiagent reinforcement learning \cite{rashid2018qmix, sunehag2018value, son2019qtran, mahajan2019maven}. 
These works enforce independence between the agents by imposing that the team's value function can be decomposed into local functions for each of the agents.
In our work, we do not consider decomposition of the value function, but instead directly synthesize a joint policy that leads to intrinsic independence between agents.
Another method to compute policies for decentralized execution, is to post-process a given joint policy.
For example, \cite{dobbe2017fully} uses the rate-distortion framework \cite{cover1991elements} for this purpose.
Our work does not assume a joint policy to be given a priori; we instead directly synthesize a joint policy that minimizes dependencies.

As discussed above, prior works tackle communication loss by making the policies fully decentralized \cite{rashid2018qmix, son2019qtran}, or by having the agents maintain beliefs about their teammates \cite{becker2009analyzing, wu2011online}.
While belief-based myopic approaches lead to high reward for a single step, they do not guarantee optimality over entire paths. 
Instead of maintaining such belief distributions, in our work each agent creates imaginary copies of its teammates when communication is lost; this idea is similar in spirit to the concept of digital twins \cite{boschert2016digital}. 
Combined with total correlation, the proposed imaginary play algorithm leads to performance guarantees over the entire path.

\section{Preliminaries}

\label{sec:preliminaries}

In this section, we outline several definitions and notation used throughout the paper.
Given a finite collection of \(\numAgents\) agents -- which we index by \(i \in \{1, 2, \ldots, \numAgents\}\) -- we model the dynamics of each individual agent using a Markov decision process (MDP) \(\mdp^i\).
An MDP is a tuple \(\mdp^i = (\mdpStateSpace^i, \mdpInitialState^i, \mdpActionSpace^i, \mdpTransition^i)\).
Here, \(\mdpStateSpace^i\) is a finite set of states, \(\mdpInitialState^i \in \mdpStateSpace^i\) is an initial state, \(\mdpActionSpace^i\) is a finite set of available actions, and \(\mdpTransition^i : \mdpStateSpace^i \times \mdpActionSpace^i \to \distribution(\mdpStateSpace^i)\) is a transition probability function.
We use \(\distribution(\mdpStateSpace^i)\) to denote the set of all probability distributions over the state space \(\mdpStateSpace^i\). For notational brevity, we use \(\mdpTransition^i(\mdpState^{i}, \mdpAction^i, \mdpStateAlt^i)\) to denote the probability of \(\mdpStateAlt^{i}\) given by the distribution \(\mdpTransition^i(\mdpState^{i}, \mdpAction^i)\).
A path \(\mdpPath^i\) in the MDP \(\mdp^i\) is an infinite sequence \(\mdpPath^i = \mdpState^i_0 \mdpAction^i_0 \mdpState^i_1 \mdpAction^i_1 \ldots\) of state-action pairs such that for every \(t = 0,1,\ldots\), \(\mdpTransition^i(\mdpState^i_t, \mdpAction^i_t, \mdpState^i_{t+1}) >0\).

Given such collection of agents, along with their corresponding MDPs \(\mdp^i\), we formulate the team's decision problem as a cooperative Markov game \(\game\). The game is considered cooperative because all agents share a common objective.
A cooperative Markov game involving \(\numAgents\) agents, each of which is modeled by an MDP \(\mdp^i = (\mdpStateSpace^i, \mdpInitialState^i, \mdpActionSpace^i, \mdpTransition^i)\), is given by the tuple \(\game = (\gameStateSpace, \gameInitialState, \gameActionSpace, \gameTransition)\).
Here, \(\gameStateSpace = \mdpStateSpace^1 \times \mdpStateSpace^2 \times \ldots \times \mdpStateSpace^{\numAgents}\) is the finite set of joint states, \(\gameInitialState = (\mdpInitialState^1, \ldots, \mdpInitialState^{\numAgents})\) is the joint initial state, \(\gameActionSpace = \mdpActionSpace^1 \times \mdpActionSpace^2 \times \ldots \times \mdpActionSpace^{\numAgents}\) is the finite set of joint actions, and \(\gameTransition : \gameStateSpace \times \gameActionSpace \to \distribution(\gameStateSpace)\) is the joint transition probability function. 
For notational brevity, we use \(\gameTransition(\gameState, \gameAction, \gameStateAlt)\) to denote the probability of \(\gameStateAlt\) in the distribution \(\gameTransition(\gameState, \gameAction)\).
The joint transition function \(\gameTransition\) is defined as \(\gameTransition(\gameState, \gameAction, \gameStateAlt) = \prod_{i=1}^{\numAgents} \mdpTransition(\mdpState^i, \mdpAction^i, \mdpStateAlt^i)\) for all \(\gameState = (\mdpState^1, \ldots, \mdpState^{\numAgents}), \gameStateAlt = (\mdpStateAlt^{1}, \ldots, \mdpStateAlt^{\numAgents}) \in \gameStateSpace\) and \(\gameAction = (\mdpAction^1, \ldots, \mdpAction^{\numAgents}) \in \gameActionSpace\). We note that the definition of the joint transition function \(\gameTransition\) assumes that the dynamics of the individual agents are independent.
We use \(\gamePath = \gameState_{0} \gameAction_{0} \gameState_{1} \gameAction_{1} \ldots\) to denote the  joint path of the agents. 
The joint path \(\gamePath\) is the union of individual paths \(\mdpPath^{1},\ldots, \mdpPath^{\numAgents}\).

A (stationary) joint policy \(\jointPolicy : \gameStateSpace \to \distribution(\gameActionSpace)\) is a mapping from a particular joint state to a probability distribution over joint actions.
We use \(\jointPolicy(\gameState, \gameAction)\) to denote the probability that action \(\gameAction\) is selected by \(\jointPolicy\) given the team is in joint state \(\gameState\).

In this work we consider team reach-avoid problems. 
That is, the team's objective is to collectively reach some target set \(\targetSet \subseteq \gameStateSpace\) of states, while avoiding a set \(\deadSet \subseteq \gameStateSpace\) of states.
The centralized planning problem then, is to solve for a team policy \(\jointPolicy\) maximizing the probability of reaching \(\targetSet\) from the team's initial joint state \(\gameInitialState\), while avoiding \(\deadSet\).
We call this probability value the reach-avoid probability.
More formally, we say that a \textit{path} \(\gamePath = \gameState_1 \gameAction_1 \gameState_2 \gameAction_2 \ldots\) successfully reaches the target set \(\targetSet\) if there exists some time \(M\) such that \(\gameState_{M} \in \targetSet\) and for all \(t < M\), \(\gameState_t \not \in \deadSet\).

For notational convenience, we use \(\gameState^{-i} \in \mdpStateSpace^1 \times \ldots \times \mdpStateSpace^{i-1} \times \mdpStateSpace^{i+1} \times \ldots \times \mdpStateSpace^{\numAgents}\) to denote the states of agent \(i\)'s teammates, excluding agent \(i\) itself.
By \(\gameStateSpace^{-i} = \mdpStateSpace^1 \times \ldots \times \mdpStateSpace^{i-1} \times \mdpStateSpace^{i+1} \times \ldots \mdpStateSpace^{\numAgents}\), we denote the set of all collections of the states of agent \(i\)'s teammates.
We similarly use \(\gameAction^{-i}\) and \(\gameActionSpace^{-i}\) to denote the actions of agent \(i\)'s teammates and the set of all possible such collections of actions, respectively.

We use \(\occupancyVar_{\gameState, \gameAction}\) to denote the occupancy measure of the state-action pair \((\gameState, \gameAction)\), i.e., the expected number of times that action \(\gameAction\) is taken at state \(\gameState\). Similarly, \(\occupancyVar_{\mdpState^{i}, \mdpAction^{i}}\) denotes the the occupancy measure of the state-action pair \((\mdpState^{i}, \mdpAction^{i})\) for agent \(i\). We note that \(\occupancyVar_{\mdpState^{i}, \mdpAction^{i}} = \sum_{\gameState^{-i} \in \gameStateSpace^{-i}} \sum_{\gameAction^{-i}\in\gameActionSpace^{-i}} \occupancyVar_{\mdpState^i, \gameState^{-i}, \mdpAction^{i}, \gameAction^{-i}}\). 

The entropy~\cite{cover1991elements} of a discrete random variable \(\genericRandomVar\) with a support \(\genericRandomVarSupport\) is \(\entropy(\genericRandomVar) = -\sum_{y\in \genericRandomVarSupport} \Pr(\genericRandomVar = y) \log(\Pr(\genericRandomVar = y)).\)  The Kullback-Leibler (KL) divergence~\cite{cover1991elements} between discrete probability distributions \(\genericDistribution^{1}\) and \(\genericDistribution^{2}\) with supports \(\genericDistributionSupport^{1}\) and  \(\genericDistributionSupport^{2}\), respectively, is \[\kl(\genericDistribution^{1} || \genericDistribution^{2}) = \sum_{q\in \genericDistributionSupport^{1}} \genericDistribution^{1}(q) \log\left(\frac{\genericDistribution^{1}(q)}{\genericDistribution^{2}(q)}\right).\]

\section{Multiagent Planning and Communication}

\label{sec:multiagent_planning_and_comms}

In this work we study the cooperative execution of joint policies when communication between the agents is intermittent, and in some cases entirely absent.
We begin by discussing the inter-agent communication that is necessary, in general, for team policy execution before we present the problem statement.

\begin{figure}[t]
    \centering
    \tikzset{>=stealth}
\resizebox{1\columnwidth}{!}{
\begin{tikzpicture}
    \contourlength{0mm};
    \draw[dotted] (-0.5,0.75) -- (0,0.75);
    \draw[dotted] (10.5,0.75) -- (11,0.75);
    \draw[ fill=blue, fill opacity=0.2, text opacity = 1] (0,0) rectangle (1.75,1.5) node[pos=.5] {\small  {\shortstack{Share \(\mdpState^{i}_{t}\) \\with other\\ agents}}};\
    \draw[ fill=red, fill opacity=0.2, text opacity = 1] (1.75,0) rectangle (3.5,1.5) node[pos=.5] { \small  {\shortstack{Jointly\\decide on\\action \(\gameAction_{t}\)}}};
    \draw[fill=green!50!black, fill opacity=0.2, text opacity = 1] (3.5,0) rectangle (5.25,1.5) node[pos=.5] {\small \shortstack{Transition\\to \(\mdpState^{i}_{t+1}\)}};
    \draw[ fill=blue, fill opacity=0.2, text opacity = 1] (5.25,0) rectangle (7,1.5) node[pos=.5] {\small  {\shortstack{Share \(\mdpState^{i}_{t+1}\) \\with other\\ agents}}};
    \draw[ fill=red, fill opacity=0.2, text opacity = 1] (7,0) rectangle (8.75,1.5) node[pos=.5] {\small {\shortstack{Jointly\\decide on\\action \(\gameAction_{t+1}\)}}};
    \draw[ fill=green!50!black, fill opacity=0.2, text opacity = 1] (8.75,0) rectangle (10.5,1.5) node[pos=.5] {\small \shortstack{Transition\\to \(\mdpState^{i}_{t+2}\)}};
    
    \filldraw[black] (-0.3,-0.5) circle (0pt) node[anchor=north west]{\shortstack{Check communication\\availability at time \(t\)}};
    \filldraw[black] (5.25,-0.5) circle (0pt) node[anchor=north]{\shortstack{Check  communication\\availability at time \(t+1\)}};
    \filldraw[black] (10.8,-0.5) circle (0pt) node[anchor=north east]{\shortstack{Check  communication\\availability at time \(t+2\)}};
    
    \draw[-stealth, thick]  (0,-0.5) -- (0,0);
    \draw[-stealth, thick]  (5.25,-0.5) -- (5.25,0);
    \draw[-stealth, thick]  (10.5,-0.5) -- (10.5,0);

\end{tikzpicture}}
    \caption{An illustration of the procedure for joint policy execution. At each decision step, all agents simultaneously check whether communication is available. If it is available, the agents share their local states in order to obtain the current joint state \(\gameState_t\) before agreeing upon a joint action \(\gameAction_t\) sampled from the joint policy \(\jointPolicy\). Otherwise, the agents execute \(\jointPolicy\) using imaginary play, outlined in Algorithm \ref{alg:imaginaryplay}.}
    \label{fig:policy_execution}
\end{figure}
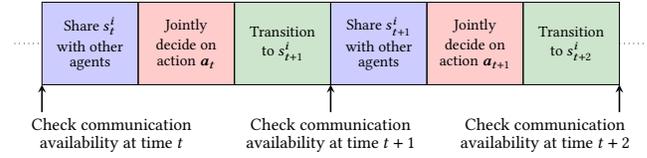

The agents operate in the environment by collectively executing a joint policy \(\jointPolicy\).
Each agent only has access to its own local state and action information. 
The agents must communicate their local states \(\mdpState^i_t\) at each timestep \(t\) and use \(\jointPolicy\) to collectively decide on a joint action \(\gameAction\), as is illustrated in Figure \ref{fig:policy_execution}.
Each agent executes its own local component \(\mdpAction^{i}\) of the selected joint action and resultingly transitions to its next local state \(\mdpState^i_{t+1}\).

We note that the joint policy requires communication between the agents at every time step. 
On the other hand, if the team suffers a communication failure at any given timestep, then they will not be able to share the necessary information to execute the joint policy in the manner described above.

\paragraph{Problem statement.} 
\((1)\) Create a planning algorithm that enables the agents to perform decentralized execution of the joint policy when communication is lost.
\((2)\) Quantify the performance of the team when such an algorithm is used during communication losses.
\((3)\) Synthesize joint policies that remain performant, even when communication is lost.
\begin{figure}[t]
    \centering
    \includegraphics[height=5cm] {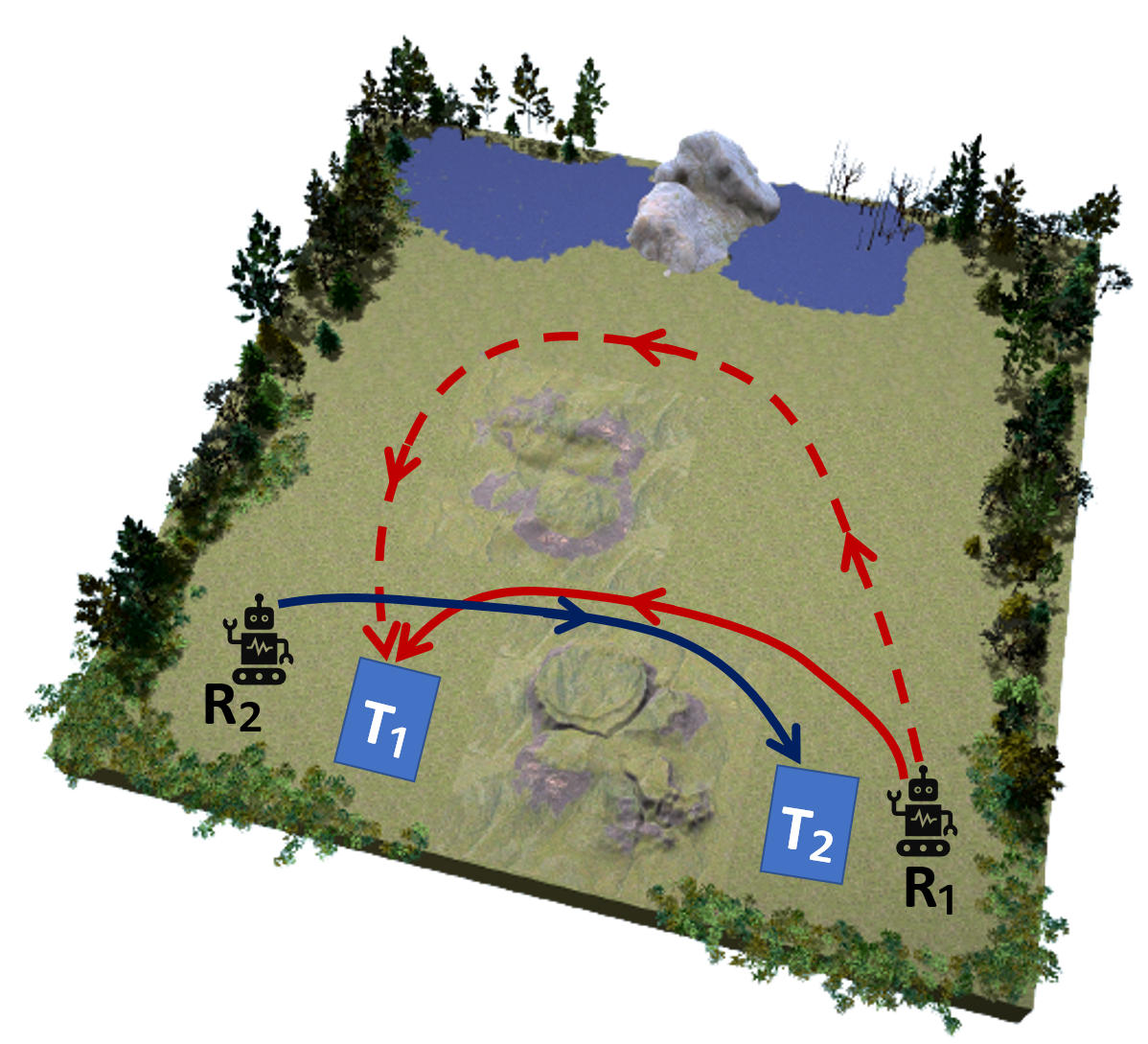}
    \caption{A two-agent navigation example. Two robots, \(\robot_1\) and \(\robot_2\), must navigate to their respective targets, \(\goal_1\) and \(\goal_2\), while avoiding collisions with each other.
    The terrain necessitates that each robot navigates through one of two valleys, while avoiding the water at the top of the map.
    During policy execution, each robot may only observe its own location,
    however, the agents communicate their locations with each other when such communication is possible.
    The colored curves illustrate different paths that the robots might to take, depending on the availability of communication.}
    \label{fig:running_example}
\end{figure}

 \paragraph{An illustrative example.} 
We present the running example illustrated in Figure \ref{fig:running_example} to help motivate the above problems.
Two robots \(\robot_1\) and \( \robot_2\) must simultaneously navigate to their respective targets \(\goal_1\) and \(\goal_2\).
The robots must also maintain a pre-specified minimum distance from each other during navigation to reduce the risk of the robots colliding.
Furthermore, rough terrain makes large portions of the navigation environment impassable, requiring the robots to navigate through one of two narrow valleys in order to reach their targets.
Finally, a lake of water presents risk to the robots; if either of them accidentally falls into the water, then the team fails its task.
The team's task is only considered complete once both robots have safely navigated to their respective targets. The objective of the agents is to complete this task with as high a probability as possible.

Given this team task, the robots may both choose to navigate through the bottom valley in order to reach their targets.
This route is shorter than traveling through the top valley for both robots, and it avoids passing near the dangerous body of water.
However, they must take turns when passing through the shared bottom valley to ensure that the robots never get too close to each other.
Such behavior requires communication; both agents should share their current location and intended next action in order to avoid simultaneously entering the valley.

By contrast, if no communication is available, the robots may instead choose to navigate through different valleys altogether.
This joint behavior increases the risk that one might fall into the water, but it removes the requirement that the robots communicate.
\section{Decentralized Policy Execution Under Communication Loss}
\label{sec:imaginary_play}

Consider a scenario in which the team of agents lose communication during the execution of a joint policy. 
Under such circumstances, the agents cannot execute the policy as outlined in the previous section and as illustrated in Figure \ref{fig:policy_execution}.
Each agent must instead decide on a local action for itself, without knowing the local states or actions of its teammates.
To achieve this decentralized execution of the joint policy, we propose to use \textit{imaginary play}; each agent maintains imaginary copies of its teammates during periods of communication loss.
That is, given the joint policy, the stochastic dynamics of the Markov game, and the states of their teammates at the last timestep before communication was lost, the agents maintain simulated copies of their teammates' states.
Each agent then uses its own imaginary version of the entire team to sample a joint action from the policy, executes its own local component of that joint action, and then simulates the next states of its imaginary teammates.
In the next time step, this process repeats.

\begin{algorithm}[t]
    \DontPrintSemicolon 
    \SetKwBlock{DoParallel}{For every \(i \in [\numAgents]\) do in parallel}{end}
    \(t_{loss} = \infty\).
    
    \For{\( t = 0,1,\ldots \)}{
        
        \eIf{Communication is possible}{
        
            \DoParallel{
                Share $\mdpState^{i}_{t}$ with other agents.\;
                
                Set $\hat{\mdpState}^{j}_{t,i} = \mdpState^{j}_{t}$ for all $j\neq i$.\;
                
                Jointly decide on an action $\gameAction_{t}  \sim \jointPolicy(\gameState_{t})$.\;
                
                Set $\hat{\gameAction}_{t,i} = \gameAction_{t}$.\;
                
                Execute $\mdpAction^{i}_{t}$ and transition to $\mdpState^{i}_{t+1} \sim \mdpTransition(\mdpState^{i}_{t},\mdpAction^{i}_{t})$.\;
            }
        }{ 
        
            Set \(t_{loss} = t\). \textbf{break}\;
            
            }
            
       }
    
    \For{\( t = t_{loss},t_{loss}+1,\ldots\)}{
            
        \eIf{t = 0}{
            Set $\hat{\mdpState}^{j}_{t,i} = \mdpInitialState^{j}$ for all $j \neq i$. \;
        }{
            Sample $\hat{\mdpState}^{j}_{t,i} \sim \mdpTransition^{j}(\hat{\mdpState}^{j}_{t-1,i}, \hat{\mdpAction}^{j}_{t-1,i})$ for all $j \neq i$. \;
        }
        \DoParallel{
            Sample $\hat{\mdpState}^{j}_{t,i} \sim \mdpTransition^{j}(\hat{\mdpState}^{j}_{t-1,i}, \hat{\mdpAction}^{j}_{t,i})$ for all $j \neq i$. \;
            
            Decide on an action $\hat{\gameAction}_{t,i}  \sim \jointPolicy(\hat{\mdpState}^{1}_{t,i}, \ldots, \hat{\mdpState}^{i-1}_{t,i}, \mdpState^{i}_{t}, \hat{\mdpState}^{i+1}_{t,i}, \ldots, \hat{\mdpState}^{N}_{t,i})$.\;
    
            Execute $\hat{\mdpAction}^{i}_{t,i}$ and transition to $\mdpState^{i}_{t+1} \sim \mdpTransition(\mdpState^{i}_{t},\hat{\mdpAction}^{i}_{t,i})$.\;
        }
    }
    \caption{Policy Execution with Imaginary Play}
    \label{alg:imaginaryplay}
\end{algorithm}

\begin{algorithm}[t]
    \DontPrintSemicolon 
    \SetKwBlock{DoParallel}{For every \(i \in [\numAgents]\) do in parallel}{end}
    
    \For{\( t = 0,1,\ldots \)}{
        
        \eIf{Communication is possible}{
        
            \DoParallel{
                Share $\mdpState^{i}_{t}$ with other agents.\;
                
                Set $\hat{\mdpState}^{j}_{t,i} = \mdpState^{j}_{t}$ for all $j\neq i$.\;
                
                Jointly decide on an action $\gameAction_{t}  \sim \jointPolicy(\gameState_{t})$. \;
                
                Set $\hat{\gameAction}_{t,i} = \gameAction_{t}$.\;
                
                Execute $\mdpAction^{i}_{t}$ and transition to $\mdpState^{i}_{t+1} \sim \mdpTransition(\mdpState^{i}_{t},\mdpAction^{i}_{t})$.\;
            }
        }{

            \DoParallel{
                    \eIf{t = 0}{
            Set $\hat{\mdpState}^{j}_{t,i} = \mdpInitialState^{j}$ for all $j \neq i$. \;
        }{
            Sample $\hat{\mdpState}^{j}_{t,i} \sim \mdpTransition^{j}(\hat{\mdpState}^{j}_{t-1,i}, \hat{\mdpAction}^{j}_{t-1,i})$ for all $j \neq i$. \;
        }
                Decide on an action $\hat{\gameAction}_{t,i}  \sim \jointPolicy(\hat{\mdpState}^{1}_{t,i}, \ldots, \hat{\mdpState}^{i-1}_{t,i}, \mdpState^{i}_{t}, \hat{\mdpState}^{i+1}_{t,i}, \ldots, \hat{\mdpState}^{N}_{t,i})$.\;
        
                Execute $\hat{\mdpAction}^{i}_{t,i}$ and transition to $\mdpState^{i}_{t+1} \sim \mdpTransition(\mdpState^{i}_{t},\hat{\mdpAction}^{i}_{t,i})$.\;
            }
        }
    }

    \caption{Policy Execution with Intermittent Communication}
    \label{alg:intermittent}
\end{algorithm}

Algorithm \ref{alg:imaginaryplay} details this process of joint policy execution through imaginary play.
Before the communication breaks, every Agent \(i\) shares its state \(\mdpState^{i}_{t}\) with its teammates at every time step, and the agents collectively decide on a joint action \(\gameAction_{t}\). 
When the communication breaks at time \(t_{loss}\), every agent \(i\) starts to play with imaginary teammates. 
That is, based on the last joint action \(\hat{\gameAction}_{t_{loss - 1},i}\) prior to communication loss, every Agent \( i \) uses the joint transition function \(\gameTransition\) to sample an imaginary state \(\hat{\mdpState}^{j}_{t_{loss - 1}+1,i}\) for each of its teammates. 
Here, \(\hat{\mdpState}^{j}_{t,i}\) denotes Agent \(i\)'s belief on Agent \(j\)'s state at time \(t\), and  \(\hat{\gameAction}_{t,i}\) denotes Agent \(i\)'s belief on the joint action at time \(t\).
Then, at every time step \(t \geq t_{loss}\), every agent \(i\) samples a joint game action \(\hat{\gameAction}_{t,i}\) using the joint policy and these imagined teammate states \( \hat{\mdpState}^{1}_{t,i}, \ldots, \hat{\mdpState}^{\numAgents}_{t,i} \). 
Every agent \(i\) then executes the local part \(\hat{\mdpAction}^{i}_{t,i} \) of its joint action \(\hat{\gameAction}_{t,i}\) and transitions to its next local state \(\mdpState^{i}_{t+1,i}\). 
Based on its imagined joint action \(\hat{\gameAction}_{t,i}\) and the previous imagined teammate states \( \hat{\mdpState}^{1}_{t,i}, \ldots, \hat{\mdpState}^{\numAgents}_{t,i} \), every Agent \( i \) also samples next imaginary states \( \hat{\mdpState}^{1}_{t+1,i}, \ldots, \hat{\mdpState}^{\numAgents}_{t+1,i} \) for its teammates.

We remark that while every agent operates cooperatively with its imaginary teammates under a communication loss, the objective of the team is evaluated with respect to the true joint state.

In some scenarios, communication failures may be intermittent as opposed to being persistent. That is, the agents may re-gain communication capabilities after periods of communication loss.
For such scenarios, we propose that the agents follow imaginary play whenever communication is lost,
update their imaginary representations when communication is re-established,
and coordinate directly with their real teammates for as long as communication remains available.
Algorithm \ref{alg:intermittent} describes this proposed approach for policy execution with intermittent communication.

\section{Measuring the Intrinsic Dependencies Between the Agents}

\label{sec:measuring_dependencies}

Given a joint policy, the team's performance under imaginary play will differ from the performance that would have been achieved under full communication.
Recall that we measure the team's performance as their probability of reaching the set \(\targetSet \subseteq \gameStateSpace\) of target joint states from the initial joint state \(\gameInitialState\), while avoiding \(\deadSet \subseteq \gameStateSpace\).

Intuitively, the team's performance under imaginary play will depend on how much the behavior of any particular agent changes according to the behavior of its teammates, as well as on how much the behavior of an agent's imaginary teammates differs from that of its actual teammates.
In other words, if the joint policy induces high intrinsic dependencies between the agents, then policy execution using imaginary play will lead to different outcomes than policy execution with fully available communication.

Total correlation~\cite{watanabe1960information} measures the amount of information shared between multiple random variables. Let \(\mdpStateActionProcess^{i}\) be a random variable over the paths \(\mdpPath^{i} = \mdpState^i_0 \mdpAction^i_0 \mdpState^i_1 \mdpAction^i_1\ldots\) of Agent \(i\) and \(\gameStateActionProcess\) be a random variable over the joint paths \(\gamePath = \gameState_0 \gameAction_0 \gameState_1 \gameAction_1 \ldots\) of all agents induced by the joint policy \(\jointPolicy\) under full communication. 
We refer to the total correlation \(\totalCorrelation_{\jointPolicy}\) of joint policy \(\jointPolicy\) as
\[\totalCorrelation_{\jointPolicy} =  \left[ \sum_{i=1}^{N} \entropy(\mdpStateActionProcess^{i} ) \right]  - \entropy(\gameStateActionProcess) . \]
There are two contributing factors to the value of the total correlation.
Firstly, if the actions of a particular agent depend on the local states of its teammates, then this will increase the value of the total correlation.
Secondly, if the joint policy is randomized and the agents need to coordinate on an action -- the action of each agent depends on the actions simultaneously selected by its teammates -- then this will also increase the value of the total correlation.

If the total correlation is \( 0 \), then there are no dependencies between the agents, i.e., the path of any given agent is independent from those of its teammates. As the dependencies between the agents increase, so too does the value of the total correlation.
We additionally remark that when there are only two agents, the total correlation between the state-action processes of the agents is equivalent to the mutual information between them.

We accordingly propose to use total correlation as measure of the intrinsic dependencies between the agents induced by a particular joint policy. 
In the next section, we relate the value of total correlation to the team's performance under communication loss.

\section{Performance Guarantees Under Communication Loss}

\label{sec:theoretical_results}

In this section, we provide lower bounds on the team's performance under a particular joint policy during communication loss.
These theoretical results are accomplished by relating the total correlation of the joint policy to the distribution over paths induced by executing that policy using imaginary play. The proofs of all results are included in the supplementary material.

\paragraph{Relating total correlation to imaginary play} Let \( \gamePathDist^{\fullcommunication}\) be the distribution of joint paths induced by the joint policy executed with full communication. Also, let \(\gamePathDist^{\imaginary}_{0}\) be the distribution of joint paths under imaginary play with no communication, i.e., \(t_{loss} = 0 \) in Algorithm \ref{alg:imaginaryplay}. 
By the definition of total correlation, we have \[\totalCorrelation_{\jointPolicy} =  \left[ \sum_{i=1}^{N} \entropy(\mdpStateActionProcess^{i} ) \right]  - \entropy(\gameStateActionProcess) =  \kl(\gamePathDist^{\fullcommunication} || \gamePathDist^{\imaginary}_{0}).\] 
From this definition, we observe that when \(\totalCorrelation_{\jointPolicy} = 0\), the induced distributions \(\gamePathDist^{\fullcommunication}, \gamePathDist^{\imaginary}_{0}\) must be the same since \(\kl(\gamePathDist^{\fullcommunication} || \gamePathDist^{\imaginary}_{0})=0\).
Furthermore, as the value of \(\totalCorrelation_{\jointPolicy}\) increases, the KL divergence between \(\gamePathDist^{\fullcommunication}\) and \(\gamePathDist^{\imaginary}_{0}\) increases as well.

\paragraph{On the closeness between path distributions induced by different communication availabilities.}
The value of \(\totalCorrelation_{\jointPolicy}\) measures how much the distribution over paths \(\gamePathDist^{\imaginary}_0\) differs from \(\gamePathDist^{\fullcommunication}\) in the setting where the agents never communicate, i.e. \(t_{loss} = 0\).
We now consider a scenario in which the agents communicate and operate together for some time, then lose communication and switch to imaginary play at time \(t_{loss} >0\). 
Let \( \gamePathDist^{\imaginary}_{t_{loss}} \) be the distribution of joint paths for an arbitrary positive value of \(t_{loss}\). 
Intuitively, we expect that the initial period of communication should not increase the KL divergence between \(\gamePathDist^{\fullcommunication}\) and \(\gamePathDist^{\imaginary}_{t_{loss}}\) in comparison with the case when \(t_{loss} = 0\).
Lemma \ref{lemma:imaginary} confirms this intuition.

\begin{lemma} \label{lemma:imaginary}
For every \(t_{loss} \in \lbrace 0, 1, \ldots \rbrace \cup \lbrace \infty \rbrace\) in Algorithm \ref{alg:imaginaryplay}, \[\kl(\gamePathDist^{\fullcommunication} || \gamePathDist^{\imaginary}_{0}) \geq  \kl(\gamePathDist^{\fullcommunication} || \gamePathDist^{\imaginary}_{t_{loss}}). \] 
\end{lemma}

We can similarly show that arbitrary intermittent communication does not increase the KL divergence between the induced path distributions. 
Let \(\sequenceCommAvailibility = \oneStepCommAvailibility_{0}, \oneStepCommAvailibility_{1},\ldots\) be a sequence of binary values such that  \(\oneStepCommAvailibility_{t} \)= 1 if and only if communication is available at time \(t\). 
The KL divergence between \( \gamePathDist^{\fullcommunication}\) and \(\gamePathDist^{\imaginary}_{0} \) is not higher than that between \( \gamePathDist^{\fullcommunication}\) and \(\gamePathDist^{\intermittent}_{\sequenceCommAvailibility} \), where \(\gamePathDist^{\intermittent}_{\sequenceCommAvailibility}\) is the distribution of paths under intermittent communication with an arbitrary sequence \(\sequenceCommAvailibility\) of communication availability. 
Furthermore, as shown in the second half of Lemma \ref{lemma:intermittent}, when \(\sequenceCommAvailibility = \lambda_0, \lambda_1, \ldots\) is a random sequence of communication availabilities, the communication dropout rate \(\probabilityFailureOneStep\) is related to the KL divergence between the distributions. 
\begin{lemma} \label{lemma:intermittent}
Let \(\sequenceCommAvailibility = \oneStepCommAvailibility_{0}, \oneStepCommAvailibility_{1},\ldots\) be an arbitrary sequence of communication availability in Algorithm \ref{alg:intermittent}. 
Then, \[\kl(\gamePathDist^{\fullcommunication} || \gamePathDist^{\imaginary}_{0}) \geq  \kl(\gamePathDist^{\fullcommunication} || \gamePathDist^{\intermittent}_{\sequenceCommAvailibility}). \] 

Let \(\sequenceCommAvailibility = \lambda_0, \lambda_1, \ldots\) be a random sequence of binary values such that every \(\oneStepCommAvailibility_{t}\) is independently sampled from a Bernoulli random variable with parameter \(1-\probabilityFailureOneStep\), and \(\gamePathDist^{\intermittent} = \expectation_{\sequenceCommAvailibility}\left[\gamePathDist^{\intermittent}_{\sequenceCommAvailibility}\right]\). Then,  \[\kl(\gamePathDist^{\fullcommunication} || \gamePathDist^{\imaginary}_{0}) \geq  \kl(\gamePathDist^{\fullcommunication} || \gamePathDist^{\intermittent})/\probabilityFailureOneStep. \]
\end{lemma}

Lemmas \ref{lemma:imaginary} and \ref{lemma:intermittent} bound the KL divergence between path distributions when the communication availability is independent from the histories of the agents. 
In  practice, communication availability may depend on the state-action processes of the agents. 
For example, in the multiagent navigation task depicted in Figure \ref{fig:running_example}, the agents may not be able to communicate if they do not have line-of-sight, e.g., when they are on the opposite sides of the mountains. Lemma \ref{lemma:imaginarystronger} shows a stronger result: The distribution over joint paths under imaginary play is close to \(\gamePathDist^{\fullcommunication}\) even when the communication availability is a function of the agents' histories.
\begin{lemma} \label{lemma:imaginarystronger}
Let \(\genericFunction:(\gameStateSpace\times\gameActionSpace)^{*}\to \lbrace 0, 1\rbrace\) be an arbitrary function that determines the communication availability based on the team's joint history such that \(\oneStepCommAvailibility_{0} = \genericFunction(\emptyString)\) and \(\oneStepCommAvailibility_{t} = \genericFunction(\gameState_{0}\gameAction_{0}\ldots\gameState_{t-1}\gameAction_{t-1})\). 
Let \(\gamePathDist^{\imaginary}_{\genericFunction} \) be the distribution over joint paths induced by imaginary play (Algorithm \ref{alg:imaginaryplay}) and communication availability dictated by \(\genericFunction\). Then, \[\kl(\gamePathDist^{\fullcommunication} || \gamePathDist^{\imaginary}_{0}) \geq  \kl(\gamePathDist^{\fullcommunication} || \gamePathDist^{\imaginary}_{\genericFunction}). \] 
\end{lemma}

We remark that Algorithms \ref{alg:imaginaryplay} and \ref{alg:intermittent} are agnostic to when the future communication failures happen. The lemmas do not assume a priori knowledge of the sequence of communication availability. 

\paragraph{On the value of the reach-avoid probability under communication loss}
We use the above results on the KL divergence between distributions of paths to derive bounds on the reach-avoid probability achieved by a particular joint policy under communication loss.

Let $\gameValue^{\fullcommunication}$ be the reach-avoid probability induced by a joint policy with full communication,
$\gameValue^{\imaginary}$ be the reach-avoid probability of the same policy under imaginary play (Algorithm \ref{alg:imaginaryplay}), and 
$\gameValue^{\intermittent}$ be the reach-avoid probability under intermittent communication (Algorithm \ref{alg:intermittent}). 
Also, let \( \deadSetPrime\) be the states from which the probability of reaching \(\targetSet\) is \(0\) under the joint policy. Define \(\len(\gamePath = \gameState_0 \gameAction_0\ldots) = \min\lbrace t+1 | \gameState_{t} \in \targetSet \cup \deadSetPrime \rbrace \) and \(\expectedLength^{\fullcommunication} = \expectation[\len(\gamePath) | \gamePath \sim \gamePathDist^{\fullcommunication}]\).

Theorem \ref{theorem:imaginarystonger} shows that the reach-avoid probability of a joint policy under imaginary play is lower-bounded by a function of the policy's reach-avoid probability with full communication and the value of \(\totalCorrelation_{\jointPolicy}\), even when the communication availability depends on the agents' histories. 

\begin{theorem} \label{theorem:imaginarystonger}
Let \(\genericFunction:(\gameStateSpace\times\gameActionSpace)^{*}\to \lbrace 0, 1\rbrace\) be an arbitrary function that determines the communication availability based on the history of the agents such that \(\oneStepCommAvailibility_{t} = \genericFunction(\gameState_{0}\gameAction_{0}\ldots\gameState_{t-1}\gameAction_{t-1})\).  For this system,
\[\gameValue^{\imaginary} \geq \gameValue^{\fullcommunication} - \sqrt{1-\exp(-\totalCorrelation_{\jointPolicy})}. \]
\end{theorem}

We now consider the setting in which the team's communication fails at some random time \(t_{loss} \geq 0\) and does not recover thereafter.
When \(t_{loss}\) follows a geometric distribution, we derive a stronger bound that relates the probability of communication failure at each time step to the reach-avoid probability under imaginary play.
\begin{theorem} \label{theorem:imaginary}
 Consider a communication system that fails with probability \(\probabilityFailureForever\) at any communication step and never recovers, i.e.,  \(\Pr(t_{loss} = t) = (1-\probabilityFailureForever)^{t}\probabilityFailureForever\) in Algorithm \ref{alg:imaginaryplay}. For this system, 
\[\gameValue^{\imaginary} \geq \max\big(\gameValue^{\fullcommunication} - \sqrt{1-\exp(-\totalCorrelation_{\jointPolicy})},  \gameValue^{\fullcommunication}  (1-\probabilityFailureForever)^{\frac{\expectedLength^{\fullcommunication}}{\gameValue^{\fullcommunication}}} \big). \]
\end{theorem}

When communication availability is intermittent, and can be modeled by a Bernoulli process, the reach-avoid probability under intermittent communication is directly lower-bounded by a function of the communication dropout rate \(\probabilityFailureOneStep\). 
We remark that the lower bound provides a means to select communication resources that are sufficient to achieve a particular performance while using noisy communication channels. 
In detail, consider a noisy communication channel on which the team must communicate. 
The code rate~\cite{cover1991elements} can be adjusted according to the desired value of \(\probabilityFailureOneStep\), which in turn determines the value of the lower bound on \(\gameValue^{\intermittent}\).
\begin{theorem}\label{theorem:intermittentstructured}
Consider a communication system that fails with probability \(\probabilityFailureOneStep\) at any communication step independent from the other communication steps. For this system, \[\gameValue^{\intermittent} \geq \max\big(\gameValue^{\fullcommunication} - \sqrt{1-\exp(-\probabilityFailureOneStep\totalCorrelation_{\jointPolicy})},  \gameValue^{\fullcommunication}  (1-\probabilityFailureOneStep)^{\frac{\expectedLength^{\fullcommunication}}{\gameValue^{\fullcommunication}}} \big).\]
\end{theorem}

The lower bounds in Theorems \ref{theorem:imaginarystonger}, \ref{theorem:imaginary}, and \ref{theorem:intermittentstructured} show that the reach-avoid probability of a joint policy under communication loss depends on the total correlation of the joint policy, the reach-avoid probability achieved with full communication, the communication dropout rate, and the expected path length under the joint policy. 
When the total correlation is \(0\), the reach-avoid probability under communication loss is the same as the reach-avoid probability with full communication. 
As the total correlation of the joint policy increases, the values of the lower bounds decrease. 
During intermittent communication, if value of the dropout rate is \(0\), then the reach-avoid probability of the joint policy executed using imaginary play (Algorithm \ref{alg:imaginaryplay}) or intermittent communication (Algorithm \ref{alg:intermittent}) is the same as when the policy is executed with full communication.
When the communication dropout rates are \(1\), the reach-avoid probability under communication loss depends on the value of the total correlation. We note that the bounds are tight when either the communication dropout rate or the total correlation is \(0\).

\section{Joint Policy Synthesis}

\label{sec:policy_synthesis}

In this section, we discuss the synthesis of minimum-dependency joint policies \(\policy_{MD}\) that are robust to communication failures. 
\paragraph{Entropy of paths for a single agent}
Given the Markov game, a stationary joint policy \(\jointPolicy\) induces a Markov chain. This Markov chain generates a stationary process \(\gameStateActionProcess\), which is the joint path of the agents. The entropy \(\entropy(\gameStateActionProcess)\) of a stationary process has a closed form expression in terms of the occupancy measures \(\occupancyVar_{\gameState, \gameAction}\) of the joint state-action pairs \((\gameState, \gameAction)\) \cite{savas2019entropy}. The path of a single agent, on the other hand, follows a hidden Markov model where \(\gameStateActionProcess\) is the underlying process and \(\mdpStateActionProcess^{i}\) is the observed process.  However, the entropy \(\entropy(\mdpStateActionProcess^{i})\) of a process that follows a hidden Markov model does not admit a closed-form expression. 

Let \(\occupancyVar_{\mdpState^{i},\mdpAction^{i}}\) be the occupancy measure for the state-action pair \((\mdpState^{i}, \mdpAction^{i}) \in \mdpStateSpace^{i} \times \mdpActionSpace^{i}\) under the joint policy \( \jointPolicy \). Consider a stationary process \(\mdpMixedStateActionProcess^{i}\) that induces the same occupancy measures \(\occupancyVar_{\mdpState^{i},\mdpAction^{i}}\) as the joint policy. 
The entropy \(\entropy(\mdpMixedStateActionProcess^{i})\) of the stationary process is greater than or equal to the entropy \(\entropy(\mdpStateActionProcess^{i})\) of the original process \cite{savas2019entropy}. 
Since \(\entropy(\mdpStateActionProcess^{i})\) does not admit a closed form expression, we instead upper bound \(\totalCorrelation_{\jointPolicy}\) using \(\entropy(\mdpMixedStateActionProcess^{i})\). 
Formally, we have \[\totalCorrelationUpperBound_{\jointPolicy} = \left[  \sum_{i=1}^{\numAgents} \entropy(\mdpMixedStateActionProcess^{i}) \right]  - \entropy(\gameStateActionProcess) \geq \totalCorrelation_{\jointPolicy} =  \left[ \sum_{i=1}^{N} \entropy(\mdpStateActionProcess^{i} ) \right]  - \entropy(\gameStateActionProcess). \]

\paragraph{The policy synthesis optimization problem.} 
To optimize the reach-avoid probability under communication loss, we would like to maximize the lower bound given in Theorem \ref{theorem:imaginary}.
However, due to the complex nature of this lower bound, we propose to instead use the following optimization problem as a proxy to the original problem:
\begin{equation}\sup_{\jointPolicy} \gameValue^{\fullcommunication} - \expectedLengthCoef \expectedLength^{\fullcommunication} - \totalCorrelationCoef \totalCorrelationUpperBound_{\jointPolicy} 
    \label{optproblemgeneric}
\end{equation}
 where \(\expectedLengthCoef > 0\) and \(\totalCorrelationCoef > 0\) are constants. 

We now represent \eqref{optproblemgeneric} in terms of occupancy measures and construct the optimization problem for synthesis. 
We first preprocess \(\game\) to ensure that \(\totalCorrelationUpperBound_{\jointPolicy}\) is well-defined. 
Define \(\deadSetPrime = \lbrace \gameState | \max_{\jointPolicy} \gameValue_{\joint} = 0 \text{ when the path begins at } \gameState \rbrace\), the set of all states from which the reach-avoid task is violated with probability \(1\). 
We note that \(\deadSetPrime \supseteq \deadSet\). For synthesis, we add an absorbing end state \(\gameProcessEndState = (\mdpProcessEndState^{1}, \ldots, \mdpProcessEndState^{\numAgents})\) and a joint action \(\epsilonTransition = (\epsilonTransition^{1}, \ldots, \epsilonTransition^{\numAgents})\) to \(\game\), which represent the end of the game in terms of the reach-avoid objective. 
Every \(\gameState \in \targetSet \cup \deadSetPrime\) has a single action \(\epsilonTransition\), and \(\gameTransition(\gameState,\epsilonTransition,\gameProcessEndState) = 1\) for all \(\gameState \in \targetSet \cup \deadSetPrime\), i.e., the states in \(\targetSet \cup \deadSetPrime\) deterministically transitions to \(\gameProcessEndState\). For synthesis, we assume that every \(\gameState \in \gameStateSpace \setminus (\targetSet \cup \deadSetPrime)\) has a finite occupancy measure, i.e., \(\sum_{\gameAction \in \gameActionSpace} \occupancyVar(\gameState,\gameAction) \leq \constantNumber \) for some \(\constantNumber \geq 0\).

In the previous sections, we assumed that the joint policy is stationary. The following proposition shows that stationary policies suffice to maximize \eqref{optproblemgeneric} after the preprocessing step.
\begin{proposition} \label{proposition:stationaryissufficient}
There exists a stationary joint policy that is a solution to \eqref{optproblemgeneric}.
\end{proposition}

Given that the stationary policies suffice, we can rewrite \eqref{optproblemgeneric} as an optimization problem in terms of the occupancy measures \(\occupancyVar_{\gameState, \gameAction}\). 
The constraints of this optimization problem are as follows. State \(\gameProcessEndState\) has an occupancy measure of zero, i.e. \(\occupancyVar_{\gameProcessEndState, \gameAction} = 0\) for all \(\gameAction \in \gameActionSpace \cup \lbrace \epsilonTransition\rbrace\). 
The other states have nonnegative occupancy measures, i.e., \(\occupancyVar_{\gameState, \gameAction} \geq 0\) for all \(\gameState \in \gameStateSpace, \gameAction \in \gameActionSpace \cup \lbrace \epsilonTransition \rbrace\). 
The occupancy measures satisfy the flow equations \(\sum_{\gameAction \in \gameActionSpace \cup \lbrace \epsilonTransition \rbrace} \occupancyVar_{\gameState, \gameAction} = \sum_{\substack{\gameStateAlt \in \gameStateSpace \\ \gameActionAlt \in \gameActionSpace \cup \lbrace \epsilonTransition \rbrace}} \occupancyVar_{\gameStateAlt, \gameActionAlt}\gameTransition(\gameStateAlt, \gameActionAlt, \gameState) + \mathbbm{1}_{\{\gameInitialState = \gameState\}}\) for all \( \gameState \in \gameStateSpace.\)
The objective function is \[    \max_{\occupancyVar} \quad \gameValue^{\fullcommunication} - \expectedLengthCoef \expectedLength^{\fullcommunication} - \totalCorrelationCoef \left( \sum_{i=1}^{\numAgents} \entropy(\mdpMixedStateActionProcess^{i}) - \entropy(\gameStateActionProcess) \right).\] 
The reach-avoid probability \(\gameValue^{\fullcommunication}\) can be expressed as \(\gameValue^{\fullcommunication} = \sum_{\gameState \in \gameStateSpace \setminus (\deadSetPrime \cup \targetSet)} \sum_{\gameAction \in \gameActionSpace} \sum_{\gameStateAlt \in \targetSet} \occupancyVar_{\gameState, \gameAction} \gameTransition(\gameState, \gameAction, \gameStateAlt).\) 
The expected path length is the expected time spent in the transient states, i.e., \(\expectedLength^{\fullcommunication} = \sum_{\gameState \in \gameStateSpace} \sum_{\gameAction \in \gameActionSpace \cup \lbrace \epsilonTransition \rbrace} \occupancyVar_{\gameState, \gameAction}.\) 
The entropy \(\entropy(\gameStateActionProcess) \)~\cite{savas2019entropy} of the joint state-action process until reaching state \(\gameProcessEndState\) is 
\[
     \sum_{\substack{\gameState \in \gameStateSpace  \\ \gameAction \in \gameActionSpace}} \occupancyVar_{\gameState, \gameAction} \log\left(\frac{\underset{\gameActionAlt \in \gameActionSpace}{\sum} \occupancyVar_{\gameState, \gameActionAlt}}{\occupancyVar_{\gameState, \gameAction}}\right) + \sum_{\substack{\gameState \in \gameStateSpace  \\ \gameAction \in \gameActionSpace}} \occupancyVar_{\gameState, \gameAction}\sum_{\gameStateAlt \in \gameStateSpace}  \gameTransition(\gameState, \gameAction, \gameStateAlt) \log\left(\frac{1}{\gameTransition(\gameState, \gameAction, \gameStateAlt)}  \right).    \]
The entropy \(\entropy(\mdpMixedStateActionProcess^{i})\)~\cite{savas2019entropy} of the stationary state-action process \(\mdpMixedStateActionProcess^{i}\) until reaching state \(\gameProcessEndState\) is

\begin{align*}
    & \sum_{\substack{\mdpState^{i} \in \mdpStateSpace^{i}  \\ \mdpAction^{i} \in \mdpActionSpace^{i} \cup \lbrace \epsilonTransition^{i} \rbrace }}  \occupancyVar_{\mdpState^{i}, \mdpAction^{i}}  \log\left(\frac{\underset{\mdpActionAlt^{i} \in \mdpActionSpace^{i}}{\sum} \occupancyVar_{\mdpState^{i}, \mdpActionAlt^{i}}}{\occupancyVar_{\mdpState^{i}, \mdpAction^{i}}}\right) \\
    & + \sum_{\substack{\mdpState^{i} \in \mdpStateSpace^{i} \\ \mdpAction^{i} \in \mdpActionSpace^{i} \cup \lbrace \epsilonTransition^{i} \rbrace }} \occupancyVar_{\mdpState^{i}, \mdpAction^{i}} \sum_{\mdpStateAlt^i \in \mdpStateSpace^{i} \cup \lbrace\mdpProcessEndState^{i}\rbrace} \mdpTransition^{i}(\mdpState^{i}, \mdpAction^{i}, \mdpStateAlt^{i}) \log\left( \frac{1}{\mdpTransition^{i}(\mdpState^{i}, \mdpAction^{i}, \mdpStateAlt^{i})} \right).
\end{align*}

The objective function of the optimization problem consists of convex, concave, and linear functions of the occupancy measures.
\(\gameValue^{\fullcommunication}\) and \(-\expectedLengthCoef \expectedLength^{\fullcommunication}\) are linear functions of the occupancy measures. \(\totalCorrelationCoef \entropy(\gameStateActionProcess)\) is a concave function of occupancy measures, and  \(-\totalCorrelationCoef \sum_{i=1}^{\numAgents} \entropy(\mdpMixedStateActionProcess^{i})\) is a convex function of occupancy measures. 
Furthermore, the problem's constraints are linear.
We use the concave-convex procedure \cite{lanckriet2009convergence,yuille2002concave} to solve for a local optimum.

After solving for the optimal values \(\occupancyVar^*_{\gameState, \gameAction}\) of the occupancy measure variables, we define the minimum-dependency joint policy as \(\policy_{MD}(\gameState, \gameAction) =\) \(\occupancyVar^*(\gameState,\gameAction) / \sum_{\gameActionAlt \in \gameActionSpace}\occupancyVar^*(\gameState,\gameActionAlt)\) for all \(\gameState \in \gameStateSpace \setminus(\targetSet \cup \deadSetPrime), \gameAction \in \gameActionSpace\)  such that \(\sum_{\gameActionAlt \in \gameActionSpace}\occupancyVar^*(\gameState,\gameActionAlt) >0\), and \(\policy_{MD}(\gameState, \gameAction) = 1/|\gameActionSpace| \) otherwise \cite{puterman2014markov}. We note that \(\policy_{MD}\) is stationary in the joint state space \(\gameStateSpace\).

\section{Numerical Experiments}

\label{sec:experiments}

\label{sec:experiments_multiagent_navigation}

We apply the proposed policy synthesis algorithm to the two-agent navigation example illustrated in Figure \ref{fig:running_example}.
The setup and objective of this task are as described in \S \ref{sec:multiagent_planning_and_comms}.
In all of the experiments, we compare the results of the minimum-dependency policy \(\policy_{MD}\), synthesized by the algorithm presented in \S \ref{sec:policy_synthesis}, to a baseline policy \(\policy_{base}\) which does not take communication into account; the baseline policy maximizes the probability that the team will complete its task, while assuming that communication will always be available.
For further details surrounding the synthesis of the baseline policy and for an additional three-agent experiment, we refer the reader to the supplementary material.
Project code is available at \href{https://github.com/cyrusneary/multi-agent-comms}{github.com/cyrusneary/multi-agent-comms}.

The common environment of the agents, illustrated in Figure \ref{fig:running_example}, is discretized into a grid of cells, each of which corresponds to an individual local state.
At any given timestep, each agent takes one of five separate actions: move left, move right, move up, move down, or remain in place.
Each agent slips with probability \(0.05\) every time it takes an action, resulting in the agent moving instead to another one of its valid neighboring states.
The resulting optimization problem has \(15,625\) variables and \(16,087\) constraints.

In all experiments, the values of the coefficients \(\expectedLengthCoef\) and \(\totalCorrelationCoef\) in the objective of the policy synthesis problem are set to \(0.01\) and \(0.4\) respectively. 
These values were selected to strike a balance between the optimization objective's three competing terms.

\begin{figure}
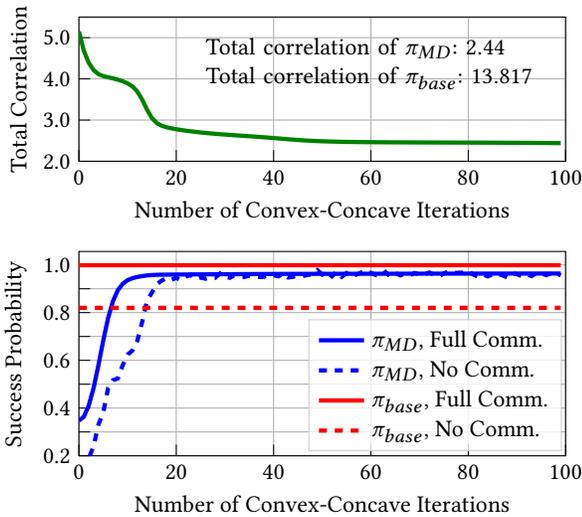

    \centering
\begin{tikzpicture}

\begin{groupplot}[group style={group name = plots, group size=1 by 2, vertical sep=1.2cm, horizontal sep=0.0cm}]

\nextgroupplot[
width=0.95\columnwidth, 
height=3.5cm,
legend cell align={left},
legend style={
  fill opacity=0.8,
  draw opacity=1,
  text opacity=1,
  at={(4.5cm, 0.4cm)},
  anchor=south west,
  draw=white!80!black
},
tick align=inside,
tick pos=left,
x grid style={white!69.0196078431373!black},
xlabel={Number of Convex-Concave Iterations},
xmajorgrids,
xmin=0.0, xmax=100.0,
xtick style={color=black},
y grid style={white!69.0196078431373!black},
ylabel={Total Correlation},
ymajorgrids,
ymin=2.0, ymax=5.5,
ytick style={color=black},
ytick={2.0, 3.0, 4.0, 5.0},
yticklabels={2.0, 3.0, 4.0, 5.0},
]

\input{tikz/total_corr_vs_iters_aux_action}

\node at (4.2cm, 1.3cm) [text width=5cm] {Total correlation of \(\policy_{MD}\): \(2.44\) \\ Total correlation of \(\policy_{base}\): \(13.817\)};

\nextgroupplot[
width=0.95*\columnwidth, 
height=4.3cm,
legend cell align={left},
legend columns = 1,
legend style={
  fill opacity=0.8,
  draw opacity=1,
  text opacity=1,
  at={(3.05cm, 0.2cm)},
  anchor=south west,
  draw=white!80!black
},
tick align=inside,
tick pos=left,
x grid style={white!69.0196078431373!black},
xlabel={Number of Convex-Concave Iterations},
xmajorgrids,
xmin=0.0, xmax=100.0,
xtick style={color=black},
y grid style={white!69.0196078431373!black},
ylabel={Success Probability},
ymajorgrids,
ymin=0.2, ymax=1.0562822829432,
ytick style={color=black},
ytick={0.2, 0.3, 0.4, 0.5, 0.6, 0.7, 0.8, 0.9, 1.0},
yticklabels={0.2, ,0.4, ,0.6, , 0.8, , 1.0},
]

\input{tikz/success_prob_vs_iters_aux_action}

\end{groupplot}

\end{tikzpicture}
    \caption{(Top) Total correlation value of the minimum-dependency policy \(\policy_{MD}\) as a function of the number of elapsed iterations of the convex-concave optimization procedure. (Bottom) Probability of task success for \(\policy_{MD}\). 
    For comparison, we plot the success probability resulting from both imaginary play execution (no communication) and centralized execution (full communication).
    To estimate the probability of task success, we perform rollouts of the joint policy and compute the empirical rate at which the team accomplishes its objective.}
    \label{fig:multiagent_navigation_results}
\end{figure}

\subsection{Fully Imaginary Play}
Figure \ref{fig:multiagent_navigation_results} compares the results of the minimum-dependency policy \(\policy_{MD}\) and the baseline policy \(\policy_{base}\) in two scenarios: when communication is either fully available or never available.

We observe from the top figure that the proposed policy synthesis algorithm is effective at reducing the total correlation of the induced stochastic state-action process; the total correlation value of \(\policy_{MD}\) is three orders of magnitude smaller than that of \(\policy_{base}\).

The bottom figure shows the strong performance of \(\policy_{MD}\) when no communication is available between the agents.
In particular, we observe that \(\policy_{MD}\) achieves a probability of task success of \(0.97\), regardless of whether the agents are able to communicate.
That is, by minimizing the total correlation of the policy, \(\policy_{MD}\) ensures the agents may successfully execute the policy without communicating during execution.
Conversely, while \(\policy_{base}\) achieves a \(0.99\) probability of task success when communication is available, this value falls to \(0.82\) if the agents lose the ability to communicate.
This experiment empirically demonstrates the intuition of Theorem \ref{theorem:imaginary}.

In addition to the quantitative results illustrated by Figure \ref{fig:multiagent_navigation_results}, we observe an interesting quantitative change in behavior between \(\policy_{base}\) and \(\policy_{MD}\).
In particular, \(\policy_{base}\) results in both of the agents navigating through the lower valley in order to arrive at their targets.
This route relies heavily on teammate coordination; the agents must communicate at each timestep in order to safely take turns passing through the valley without colliding. 
By contrast, \(\policy_{MD}\) results in agent \(\agent_1\) navigating through the top valley while \(\agent_0\) takes the bottom valley.
Intuitively, by navigating through separate valleys, this team behavior is much less likely to result in collisions even if the agent's don't share their locations with each other.
As a result, teammate coordination is much less important to successfully execute the behavior of \(\policy_{MD}\) than it is to execute that of \(\policy_{base}\).

\subsection{Intermittent Communication}
While the previous discussion focused on the empirical performance of \(\policy_{MD}\) in the setting where the agents cannot communicate at all, we now examine the setting in which random intermittent communication is available.
More specifically, we assume that at each timestep communication fails with probability \(q\), independently of whether or not communication is available during the other timesteps.
In this setting, the agents execute the joint policy according to Algorithm \ref{alg:intermittent}.
That is, if communication is available at a given timestep, all agents collectively share their local states and decide on a joint action.
Conversely, when communication is not available, the agents execute the policy using imaginary play.

Figure \ref{fig:intermittent_play_success_prob} plots the team's probability of task success when they execute either \(\policy_{MD}\) or \(\policy_{base}\) using Algorithm \ref{alg:intermittent}, as a function of the probability of communication failure \(\probabilityFailureOneStep\).
We observe that the probability of task success of the baseline policy \(\policy_{base}\) is very high when \(\probabilityFailureOneStep = 0\), however, it begins to significantly decrease as \(\probabilityFailureOneStep\) increases beyond \(0.4\).
Conversely, the proposed minimum-dependency policy \(\policy_{MD}\) does not suffer such a drop in performance; as \(\probabilityFailureOneStep\) increases and communication becomes more sparse the task success probability of policy \(\policy_{MD}\) remains constant.

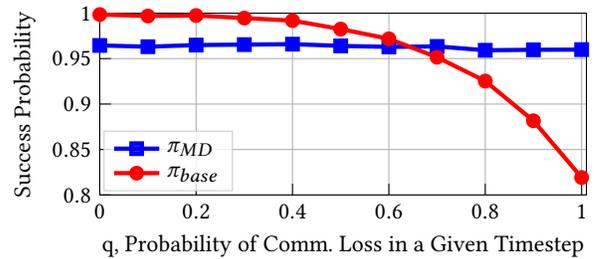
\begin{figure}
    \centering
\begin{tikzpicture}

\begin{axis}[
tick align=outside,
tick pos=left,
width=0.95\columnwidth, 
height=4cm,
x grid style={white!69.0196078431373!black},
legend cell align={left},
legend columns = 1,
legend style={
  fill opacity=0.8,
  draw opacity=1,
  text opacity=1,
  at={(0.05cm, 0.05cm)},
  anchor=south west,
  draw=white!80!black
},
xlabel={q, Probability of Comm. Loss in a Given Timestep},
xmajorgrids,
xmin=0, xmax=1.01,
xtick style={color=black},
xtick align=inside,
y grid style={white!69.0196078431373!black},
ylabel={Success Probability},
ymajorgrids,
ymin=0.80, ymax=1,
ytick style={color=black},
ytick align=inside,
]
\addplot [ultra thick, blue, mark=square*]
table {%
0 0.9645
0.1 0.9631
0.2 0.9649
0.3 0.9654
0.4 0.9659
0.5 0.964
0.6 0.9629
0.7 0.9634
0.8 0.9592
0.9 0.9597
1 0.9599
};\addlegendentry{\(\policy_{MD}\)}
\addplot [ultra thick, red, mark=*]
table {%
0 0.9985
0.1 0.9969
0.2 0.9974
0.3 0.9946
0.4 0.9918
0.5 0.9825
0.6 0.9718
0.7 0.9515
0.8 0.9253
0.9 0.8816
1 0.8192
}; \addlegendentry{\(\policy_{base}\)}
\end{axis}

\end{tikzpicture}
    \caption{Success probability of intermittent communication for different values of \(\probabilityFailureOneStep\), which represents the probability of communication unavailability during any given timestep. 
    When \(\probabilityFailureOneStep = 0\) communication is available at every timestep, and when \(\probabilityFailureOneStep = 1\) communication is never available.}
    \label{fig:intermittent_play_success_prob}
\end{figure}

\section{Conclusions}

\label{sec:conclusions}

In this work, we develop multiagent systems that are robust to communication loss.
We provide algorithms for decentralized policy execution when communication is lost and relate the performance of these algorithms to the performance achieved by the same policy under full communication.
Using these theoretical results, we propose an optimization algorithm for the synthesis of joint policies that are robust to potential losses in communication.
While the policy synthesis algorithm directly operates on the joint state space, future work will aim to use abstractions of the joint state-action space to scale the policy synthesis to larger problems.



\begin{acks}
This work was supported in part by AFRL FA9550-19-1-0169, ARL ACC-APG-RTP W911NF1920333, and ARO W911NF-20-1-0140.
\end{acks}


\newpage
\newpage

\balance

\bibliographystyle{ACM-Reference-Format} 
\bibliography{sample}


\begin{thebibliography}{33}


\ifx \showCODEN    \undefined \def \showCODEN     #1{\unskip}     \fi
\ifx \showDOI      \undefined \def \showDOI       #1{#1}\fi
\ifx \showISBNx    \undefined \def \showISBNx     #1{\unskip}     \fi
\ifx \showISBNxiii \undefined \def \showISBNxiii  #1{\unskip}     \fi
\ifx \showISSN     \undefined \def \showISSN      #1{\unskip}     \fi
\ifx \showLCCN     \undefined \def \showLCCN      #1{\unskip}     \fi
\ifx \shownote     \undefined \def \shownote      #1{#1}          \fi
\ifx \showarticletitle \undefined \def \showarticletitle #1{#1}   \fi
\ifx \showURL      \undefined \def \showURL       {\relax}        \fi
\providecommand\bibfield[2]{#2}
\providecommand\bibinfo[2]{#2}
\providecommand\natexlab[1]{#1}
\providecommand\showeprint[2][]{arXiv:#2}

\bibitem[\protect\citeauthoryear{Altman}{Altman}{1999}]%
        {altman1999constrained}
\bibfield{author}{\bibinfo{person}{Eitan Altman}.}
  \bibinfo{year}{1999}\natexlab{}.
\newblock \bibinfo{booktitle}{\emph{Constrained {Markov} decision processes}}.
  Vol.~\bibinfo{volume}{7}.
\newblock \bibinfo{publisher}{CRC Press}.
\newblock


\bibitem[\protect\citeauthoryear{Aps}{Aps}{2020}]%
        {aps2020mosek}
\bibfield{author}{\bibinfo{person}{MOSEK Aps}.}
  \bibinfo{year}{2020}\natexlab{}.
\newblock \showarticletitle{MOSEK Optimizer {API} for {Python}}.
\newblock \bibinfo{journal}{\emph{Software Package, Ver}}  \bibinfo{volume}{9}
  (\bibinfo{year}{2020}).
\newblock


\bibitem[\protect\citeauthoryear{Baier and Katoen}{Baier and Katoen}{2008}]%
        {baier2008principles}
\bibfield{author}{\bibinfo{person}{Christel Baier} {and}
  \bibinfo{person}{Joost-Pieter Katoen}.} \bibinfo{year}{2008}\natexlab{}.
\newblock \bibinfo{booktitle}{\emph{Principles of model checking}}.
\newblock \bibinfo{publisher}{MIT Press}.
\newblock


\bibitem[\protect\citeauthoryear{Becker, Carlin, Lesser, and
  Zilberstein}{Becker et~al\mbox{.}}{2009}]%
        {becker2009analyzing}
\bibfield{author}{\bibinfo{person}{Raphen Becker}, \bibinfo{person}{Alan
  Carlin}, \bibinfo{person}{Victor Lesser}, {and} \bibinfo{person}{Shlomo
  Zilberstein}.} \bibinfo{year}{2009}\natexlab{}.
\newblock \showarticletitle{Analyzing myopic approaches for multi-agent
  communication}.
\newblock \bibinfo{journal}{\emph{Computational Intelligence}}
  \bibinfo{volume}{25}, \bibinfo{number}{1} (\bibinfo{year}{2009}),
  \bibinfo{pages}{31--50}.
\newblock


\bibitem[\protect\citeauthoryear{Becker, Zilberstein, Lesser, and
  Goldman}{Becker et~al\mbox{.}}{2003}]%
        {becker2003transition}
\bibfield{author}{\bibinfo{person}{Raphen Becker}, \bibinfo{person}{Shlomo
  Zilberstein}, \bibinfo{person}{Victor Lesser}, {and}
  \bibinfo{person}{Claudia~V Goldman}.} \bibinfo{year}{2003}\natexlab{}.
\newblock \showarticletitle{Transition-independent decentralized {Markov}
  decision processes}. In \bibinfo{booktitle}{\emph{Proceedings of the 2nd
  International Conference on Autonomous Agents and Multiagent Systems}}.
  \bibinfo{pages}{41--48}.
\newblock


\bibitem[\protect\citeauthoryear{Boschert and Rosen}{Boschert and
  Rosen}{2016}]%
        {boschert2016digital}
\bibfield{author}{\bibinfo{person}{Stefan Boschert} {and}
  \bibinfo{person}{Roland Rosen}.} \bibinfo{year}{2016}\natexlab{}.
\newblock \showarticletitle{Digital twin—the simulation aspect}.
\newblock In \bibinfo{booktitle}{\emph{Mechatronic futures}}.
  \bibinfo{publisher}{Springer}, \bibinfo{pages}{59--74}.
\newblock


\bibitem[\protect\citeauthoryear{Boutilier}{Boutilier}{1996}]%
        {boutilier1996planning}
\bibfield{author}{\bibinfo{person}{Craig Boutilier}.}
  \bibinfo{year}{1996}\natexlab{}.
\newblock \showarticletitle{Planning, learning and coordination in multiagent
  decision processes}. In \bibinfo{booktitle}{\emph{Proceedings of the 6th
  Conference on Theoretical Aspects of Rationality and Knowledge}},
  Vol.~\bibinfo{volume}{96}. \bibinfo{pages}{195--210}.
\newblock


\bibitem[\protect\citeauthoryear{Boyd and Vandenberghe}{Boyd and
  Vandenberghe}{2004}]%
        {boyd2004convex}
\bibfield{author}{\bibinfo{person}{Stephen Boyd} {and} \bibinfo{person}{Lieven
  Vandenberghe}.} \bibinfo{year}{2004}\natexlab{}.
\newblock \bibinfo{booktitle}{\emph{Convex optimization}}.
\newblock \bibinfo{publisher}{Cambridge University Press}.
\newblock


\bibitem[\protect\citeauthoryear{Bretagnolle and Huber}{Bretagnolle and
  Huber}{1979}]%
        {bretagnolle1979estimation}
\bibfield{author}{\bibinfo{person}{Jean Bretagnolle} {and}
  \bibinfo{person}{Catherine Huber}.} \bibinfo{year}{1979}\natexlab{}.
\newblock \showarticletitle{Estimation des densit{\'e}s: risque minimax}.
\newblock \bibinfo{journal}{\emph{Zeitschrift f{\"u}r
  Wahrscheinlichkeitstheorie und verwandte Gebiete}} \bibinfo{volume}{47},
  \bibinfo{number}{2} (\bibinfo{year}{1979}), \bibinfo{pages}{119--137}.
\newblock


\bibitem[\protect\citeauthoryear{Cao, Yu, Ren, and Chen}{Cao
  et~al\mbox{.}}{2012}]%
        {cao2012overview}
\bibfield{author}{\bibinfo{person}{Yongcan Cao}, \bibinfo{person}{Wenwu Yu},
  \bibinfo{person}{Wei Ren}, {and} \bibinfo{person}{Guanrong Chen}.}
  \bibinfo{year}{2012}\natexlab{}.
\newblock \showarticletitle{An overview of recent progress in the study of
  distributed multi-agent coordination}.
\newblock \bibinfo{journal}{\emph{IEEE Transactions on Industrial Informatics}}
  \bibinfo{volume}{9}, \bibinfo{number}{1} (\bibinfo{year}{2012}),
  \bibinfo{pages}{427--438}.
\newblock


\bibitem[\protect\citeauthoryear{Cover and Thomas}{Cover and Thomas}{1991}]%
        {cover1991elements}
\bibfield{author}{\bibinfo{person}{Thomas~M Cover} {and} \bibinfo{person}{Joy~A
  Thomas}.} \bibinfo{year}{1991}\natexlab{}.
\newblock \bibinfo{booktitle}{\emph{Elements of Information Theory}}.
\newblock \bibinfo{publisher}{John Wiley \& Sons}.
\newblock


\bibitem[\protect\citeauthoryear{Diamond and Boyd}{Diamond and Boyd}{2016}]%
        {diamond2016cvxpy}
\bibfield{author}{\bibinfo{person}{Steven Diamond} {and}
  \bibinfo{person}{Stephen Boyd}.} \bibinfo{year}{2016}\natexlab{}.
\newblock \showarticletitle{{CVXPY}: {A} {P}ython-embedded modeling language
  for convex optimization}.
\newblock \bibinfo{journal}{\emph{Journal of Machine Learning Research}}
  \bibinfo{volume}{17}, \bibinfo{number}{83} (\bibinfo{year}{2016}),
  \bibinfo{pages}{1--5}.
\newblock


\bibitem[\protect\citeauthoryear{Dobbe, Fridovich-Keil, and Tomlin}{Dobbe
  et~al\mbox{.}}{2017}]%
        {dobbe2017fully}
\bibfield{author}{\bibinfo{person}{Roel Dobbe}, \bibinfo{person}{David
  Fridovich-Keil}, {and} \bibinfo{person}{Claire Tomlin}.}
  \bibinfo{year}{2017}\natexlab{}.
\newblock \showarticletitle{Fully Decentralized Policies for Multi-Agent
  Systems: An Information Theoretic Approach}. In
  \bibinfo{booktitle}{\emph{Proceedings of the 31st Conference on Neural
  Information Processing Systems}}. \bibinfo{pages}{2945–2954}.
\newblock
\showISBNx{9781510860964}


\bibitem[\protect\citeauthoryear{Eysenbach, Salakhutdinov, and
  Levine}{Eysenbach et~al\mbox{.}}{2021}]%
        {eysenbach2021robust}
\bibfield{author}{\bibinfo{person}{Ben Eysenbach}, \bibinfo{person}{Russ~R
  Salakhutdinov}, {and} \bibinfo{person}{Sergey Levine}.}
  \bibinfo{year}{2021}\natexlab{}.
\newblock \showarticletitle{Robust Predictable Control}.
\newblock \bibinfo{journal}{\emph{Pre-proceedings of the 35th Conference on
  Neural Information Processing Systems}}.
\newblock


\bibitem[\protect\citeauthoryear{Goldman and Zilberstein}{Goldman and
  Zilberstein}{2004}]%
        {goldman2004decentralized}
\bibfield{author}{\bibinfo{person}{Claudia~V Goldman} {and}
  \bibinfo{person}{Shlomo Zilberstein}.} \bibinfo{year}{2004}\natexlab{}.
\newblock \showarticletitle{Decentralized control of cooperative systems:
  Categorization and complexity analysis}.
\newblock \bibinfo{journal}{\emph{Journal of Artificial Intelligence Research}}
   \bibinfo{volume}{22} (\bibinfo{year}{2004}), \bibinfo{pages}{143--174}.
\newblock


\bibitem[\protect\citeauthoryear{Guestrin, Koller, and Parr}{Guestrin
  et~al\mbox{.}}{2002}]%
        {guestrin2002multiagent}
\bibfield{author}{\bibinfo{person}{Carlos Guestrin}, \bibinfo{person}{Daphne
  Koller}, {and} \bibinfo{person}{Ronald Parr}.}
  \bibinfo{year}{2002}\natexlab{}.
\newblock \showarticletitle{Multiagent planning with factored {MDPs}}. In
  \bibinfo{booktitle}{\emph{Proceedings of the 14th Conference on Neural
  Information Processing Systems}}. \bibinfo{pages}{1523--1530}.
\newblock


\bibitem[\protect\citeauthoryear{Lanckriet and Sriperumbudur}{Lanckriet and
  Sriperumbudur}{2009}]%
        {lanckriet2009convergence}
\bibfield{author}{\bibinfo{person}{Gert Lanckriet} {and}
  \bibinfo{person}{Bharath~K Sriperumbudur}.} \bibinfo{year}{2009}\natexlab{}.
\newblock \showarticletitle{On the convergence of the concave-convex
  procedure}. In \bibinfo{booktitle}{\emph{Proceedings of the 21st Conference
  on Neural Information Processing Systems}}. \bibinfo{pages}{1759--1767}.
\newblock


\bibitem[\protect\citeauthoryear{Leibfried and Grau-Moya}{Leibfried and
  Grau-Moya}{2020}]%
        {leibfried2020mutual}
\bibfield{author}{\bibinfo{person}{Felix Leibfried} {and}
  \bibinfo{person}{Jordi Grau-Moya}.} \bibinfo{year}{2020}\natexlab{}.
\newblock \showarticletitle{Mutual-information regularization in {Markov}
  decision processes and actor-critic learning}. In
  \bibinfo{booktitle}{\emph{Proceedings of the 3rd Conference on Robot
  Learning}}. \bibinfo{pages}{360--373}.
\newblock


\bibitem[\protect\citeauthoryear{Littman}{Littman}{1994}]%
        {littman1994markov}
\bibfield{author}{\bibinfo{person}{Michael~L Littman}.}
  \bibinfo{year}{1994}\natexlab{}.
\newblock \showarticletitle{Markov games as a framework for multi-agent
  reinforcement learning}. In \bibinfo{booktitle}{\emph{Proceedings of the 11th
  International Conference on Machine Learning}}. \bibinfo{pages}{157--163}.
\newblock


\bibitem[\protect\citeauthoryear{Mahajan, Rashid, Samvelyan, and
  Whiteson}{Mahajan et~al\mbox{.}}{2019}]%
        {mahajan2019maven}
\bibfield{author}{\bibinfo{person}{Anuj Mahajan}, \bibinfo{person}{Tabish
  Rashid}, \bibinfo{person}{Mikayel Samvelyan}, {and} \bibinfo{person}{Shimon
  Whiteson}.} \bibinfo{year}{2019}\natexlab{}.
\newblock \showarticletitle{Maven: Multi-agent variational exploration}. In
  \bibinfo{booktitle}{\emph{Prooceedings of the 33rd Conference on Neural
  Information Processing Systems}}. \bibinfo{pages}{7613--7624}.
\newblock


\bibitem[\protect\citeauthoryear{Melo and Veloso}{Melo and Veloso}{2011}]%
        {melo2011decentralized}
\bibfield{author}{\bibinfo{person}{Francisco~S Melo} {and}
  \bibinfo{person}{Manuela Veloso}.} \bibinfo{year}{2011}\natexlab{}.
\newblock \showarticletitle{Decentralized {MDPs} with sparse interactions}.
\newblock \bibinfo{journal}{\emph{Artificial Intelligence}}
  \bibinfo{volume}{175}, \bibinfo{number}{11} (\bibinfo{year}{2011}),
  \bibinfo{pages}{1757--1789}.
\newblock


\bibitem[\protect\citeauthoryear{Oliehoek and Amato}{Oliehoek and
  Amato}{2016}]%
        {oliehoek2016concise}
\bibfield{author}{\bibinfo{person}{Frans~A Oliehoek} {and}
  \bibinfo{person}{Christopher Amato}.} \bibinfo{year}{2016}\natexlab{}.
\newblock \bibinfo{booktitle}{\emph{A concise introduction to decentralized
  {POMDPs}}}.
\newblock \bibinfo{publisher}{Springer}.
\newblock


\bibitem[\protect\citeauthoryear{Parker, Rus, and Sukhatme}{Parker
  et~al\mbox{.}}{2016}]%
        {parker2016multiple}
\bibfield{author}{\bibinfo{person}{Lynne~E Parker}, \bibinfo{person}{Daniela
  Rus}, {and} \bibinfo{person}{Gaurav~S Sukhatme}.}
  \bibinfo{year}{2016}\natexlab{}.
\newblock \showarticletitle{Multiple mobile robot systems}.
\newblock In \bibinfo{booktitle}{\emph{Springer Handbook of Robotics}}.
  \bibinfo{publisher}{Springer}, \bibinfo{pages}{1335--1384}.
\newblock


\bibitem[\protect\citeauthoryear{Puterman}{Puterman}{2014}]%
        {puterman2014markov}
\bibfield{author}{\bibinfo{person}{Martin~L Puterman}.}
  \bibinfo{year}{2014}\natexlab{}.
\newblock \bibinfo{booktitle}{\emph{Markov decision processes: discrete
  stochastic dynamic programming}}.
\newblock \bibinfo{publisher}{John Wiley \& Sons}.
\newblock


\bibitem[\protect\citeauthoryear{Rashid, Samvelyan, Schroeder, Farquhar,
  Foerster, and Whiteson}{Rashid et~al\mbox{.}}{2018}]%
        {rashid2018qmix}
\bibfield{author}{\bibinfo{person}{Tabish Rashid}, \bibinfo{person}{Mikayel
  Samvelyan}, \bibinfo{person}{Christian Schroeder}, \bibinfo{person}{Gregory
  Farquhar}, \bibinfo{person}{Jakob Foerster}, {and} \bibinfo{person}{Shimon
  Whiteson}.} \bibinfo{year}{2018}\natexlab{}.
\newblock \showarticletitle{Qmix: Monotonic value function factorisation for
  deep multi-agent reinforcement learning}. In
  \bibinfo{booktitle}{\emph{Proceedings of the 35th International Conference on
  Machine Learning}}. \bibinfo{pages}{4295--4304}.
\newblock


\bibitem[\protect\citeauthoryear{Savas, Ornik, Cubuktepe, Karabag, and
  Topcu}{Savas et~al\mbox{.}}{2019}]%
        {savas2019entropy}
\bibfield{author}{\bibinfo{person}{Yagiz Savas}, \bibinfo{person}{Melkior
  Ornik}, \bibinfo{person}{Murat Cubuktepe}, \bibinfo{person}{Mustafa~O
  Karabag}, {and} \bibinfo{person}{Ufuk Topcu}.}
  \bibinfo{year}{2019}\natexlab{}.
\newblock \showarticletitle{Entropy maximization for Markov decision processes
  under temporal logic constraints}.
\newblock \bibinfo{journal}{\emph{IEEE Trans. Automat. Control}}
  \bibinfo{volume}{65}, \bibinfo{number}{4} (\bibinfo{year}{2019}),
  \bibinfo{pages}{1552--1567}.
\newblock


\bibitem[\protect\citeauthoryear{Son, Kim, Kang, Hostallero, and Yi}{Son
  et~al\mbox{.}}{2019}]%
        {son2019qtran}
\bibfield{author}{\bibinfo{person}{Kyunghwan Son}, \bibinfo{person}{Daewoo
  Kim}, \bibinfo{person}{Wan~Ju Kang}, \bibinfo{person}{David~Earl Hostallero},
  {and} \bibinfo{person}{Yung Yi}.} \bibinfo{year}{2019}\natexlab{}.
\newblock \showarticletitle{Qtran: Learning to factorize with transformation
  for cooperative multi-agent reinforcement learning}. In
  \bibinfo{booktitle}{\emph{Proceedings of the 36th International Conference on
  Machine Learning}}. \bibinfo{pages}{5887--5896}.
\newblock


\bibitem[\protect\citeauthoryear{Sunehag, Lever, Gruslys, Czarnecki, Zambaldi,
  Jaderberg, Lanctot, Sonnerat, Leibo, Tuyls, et~al\mbox{.}}{Sunehag
  et~al\mbox{.}}{2018}]%
        {sunehag2018value}
\bibfield{author}{\bibinfo{person}{Peter Sunehag}, \bibinfo{person}{Guy Lever},
  \bibinfo{person}{Audrunas Gruslys}, \bibinfo{person}{Wojciech~Marian
  Czarnecki}, \bibinfo{person}{Vinicius Zambaldi}, \bibinfo{person}{Max
  Jaderberg}, \bibinfo{person}{Marc Lanctot}, \bibinfo{person}{Nicolas
  Sonnerat}, \bibinfo{person}{Joel~Z Leibo}, \bibinfo{person}{Karl Tuyls},
  {et~al\mbox{.}}} \bibinfo{year}{2018}\natexlab{}.
\newblock \showarticletitle{Value-decomposition networks for cooperative
  multi-agent learning based on team reward}. In
  \bibinfo{booktitle}{\emph{Proceedings of the 17th International Conference on
  Autonomous Agents and Multiagent Systems}}. \bibinfo{pages}{2085--2087}.
\newblock


\bibitem[\protect\citeauthoryear{Tanaka, Sandberg, and Skoglund}{Tanaka
  et~al\mbox{.}}{2021}]%
        {tanaka2021transfer}
\bibfield{author}{\bibinfo{person}{Takashi Tanaka}, \bibinfo{person}{Henrik
  Sandberg}, {and} \bibinfo{person}{Mikael Skoglund}.}
  \bibinfo{year}{2021}\natexlab{}.
\newblock \showarticletitle{Transfer-entropy-regularized {Markov} decision
  processes}.
\newblock \bibinfo{journal}{\emph{IEEE Trans. Automat. Control}}
  (\bibinfo{year}{2021}).
\newblock


\bibitem[\protect\citeauthoryear{Wang, He, Yu, Qiu, An, and Rabinovich}{Wang
  et~al\mbox{.}}{2020}]%
        {wang2020learning}
\bibfield{author}{\bibinfo{person}{Rundong Wang}, \bibinfo{person}{Xu He},
  \bibinfo{person}{Runsheng Yu}, \bibinfo{person}{Wei Qiu}, \bibinfo{person}{Bo
  An}, {and} \bibinfo{person}{Zinovi Rabinovich}.}
  \bibinfo{year}{2020}\natexlab{}.
\newblock \showarticletitle{Learning efficient multi-agent communication: An
  information bottleneck approach}. In \bibinfo{booktitle}{\emph{Proceedings of
  the 37th International Conference on Machine Learning}}.
  \bibinfo{pages}{9908--9918}.
\newblock


\bibitem[\protect\citeauthoryear{Watanabe}{Watanabe}{1960}]%
        {watanabe1960information}
\bibfield{author}{\bibinfo{person}{Satosi Watanabe}.}
  \bibinfo{year}{1960}\natexlab{}.
\newblock \showarticletitle{Information theoretical analysis of multivariate
  correlation}.
\newblock \bibinfo{journal}{\emph{IBM Journal of research and development}}
  \bibinfo{volume}{4}, \bibinfo{number}{1} (\bibinfo{year}{1960}),
  \bibinfo{pages}{66--82}.
\newblock


\bibitem[\protect\citeauthoryear{Wu, Zilberstein, and Chen}{Wu
  et~al\mbox{.}}{2011}]%
        {wu2011online}
\bibfield{author}{\bibinfo{person}{Feng Wu}, \bibinfo{person}{Shlomo
  Zilberstein}, {and} \bibinfo{person}{Xiaoping Chen}.}
  \bibinfo{year}{2011}\natexlab{}.
\newblock \showarticletitle{Online planning for multi-agent systems with
  bounded communication}.
\newblock \bibinfo{journal}{\emph{Artificial Intelligence}}
  \bibinfo{volume}{175}, \bibinfo{number}{2} (\bibinfo{year}{2011}),
  \bibinfo{pages}{487--511}.
\newblock


\bibitem[\protect\citeauthoryear{Yuille and Rangarajan}{Yuille and
  Rangarajan}{2002}]%
        {yuille2002concave}
\bibfield{author}{\bibinfo{person}{Alan~L Yuille} {and} \bibinfo{person}{Anand
  Rangarajan}.} \bibinfo{year}{2002}\natexlab{}.
\newblock \showarticletitle{The concave-convex procedure {(CCCP)}}. In
  \bibinfo{booktitle}{\emph{Proceedings of the 14th Conference on Neural
  Information Processing Systems}}. \bibinfo{pages}{1033--1040}.
\newblock


\end{thebibliography}


\newpage

\onecolumn
\appendix

\begin{center} 
    \begin{LARGE}
        \textbf{Planning Not to Talk: Multiagent Systems that are Robust to Communication Loss}  \\  \end{LARGE}
    \begin{Large}
        \textbf{Supplementary Material}
    \end{Large}
\end{center}

\section{Proofs for technical results}
We first define some notation to be used in the notation and provide different expressions of total correlation.

\paragraph{Notation.} Under the joint policy \(\jointPolicy\) with full communication, let \(\gameStateRandomVar_{t}\) be a random variable denoting the joint state of the agents at time \(t\), \(\gameActionRandomVar_{t}\) be a random variable denoting the joint action of the agents at time \(t\), \(\mdpStateRandomVar^{i}_{t}\) be a random variable denoting the state of Agent \(i\) at time \(t\), and \(\mdpActionRandomVar^{i}_{t}\) be a random variable denoting the action of Agent \(i\) at time \(t\). 

We use \(\probabilityMeasure^{\fullcommunication}\) to denote the probability measure over the (finite or infinite) state-action process under the joint policy with full communication. \(\probabilityMeasure^{\imaginary}_{t_{loss}}\) denotes the probability measure over the (finite or infinite) state-action process under the imaginary play (under Algorithm \ref{alg:imaginaryplay}) where the first communication loss happens at time \(t_{loss}\). \(\probabilityMeasure^{\imaginary}_{\genericFunction}\) denotes the probability measure over the (finite or infinite) state-action process under the imaginary play (under Algorithm \ref{alg:imaginaryplay}) where \(\genericFunction:(\gameStateSpace\times\gameActionSpace)^{*}\to \lbrace 0, 1\rbrace\) determines the communication availability based on the team's joint history. \(\probabilityMeasure^{\intermittent}_{\sequenceCommAvailibility}\) denotes the probability measure over the (finite or infinite) state-action process under the intermittent communication (under Algorithm \ref{alg:intermittent}) with a sequence \(\sequenceCommAvailibility\) of communication availability.

The Kleene star applied to a set \(\genericSet\) of symbols is the set \(\genericSet^{*} = \bigcup_{i \geq 0} \genericSet^{i}\) of all finite-length words where  \(\genericSet^{0} = \lbrace \emptyString \rbrace\) and \(\emptyString\) is the empty string. The set of all infinite-length words is denoted by \(\genericSet^{\omega}\).

\paragraph{Different expressions of total correlation.} 
The total correlation~\cite{watanabe1960information} of joint policy \(\jointPolicy\) is \[\totalCorrelation_{\jointPolicy}  =\kl(\gamePathDist^{\fullcommunication} || \gamePathDist^{\imaginary}_{0}) =  \left[ \sum_{i=1}^{N} \entropy(\mdpStateActionProcess^{i}) \right]  - \entropy(\gameStateActionProcess).\] By the chain rule of entropy~\cite{cover1991elements} and the fact that \(\gameInitialState\) is a common knowledge, we have
\[\totalCorrelation_{\jointPolicy}= \left[ \sum_{i=1}^{N} \entropy(\mdpStateRandomVar^{i}_{0}\mdpActionRandomVar^{i}_{0}|\gameStateRandomVar_{0}) \right]  - \entropy(\gameStateRandomVar_{0}\gameActionRandomVar_{0}|\gameStateRandomVar_{0}) + \sum_{t=0}^{\infty} \left[ \left[ \sum_{i=1}^{N} \entropy(\mdpStateRandomVar^{i}_{t}\mdpActionRandomVar^{i}_{t}|\gameStateRandomVar_{0}\mdpActionRandomVar^{i}_{0}  \ldots \mdpStateRandomVar^{i}_{t-1}\mdpActionRandomVar^{i}_{t-1}) \right]  - \entropy(\gameStateRandomVar_{t}\gameActionRandomVar_{t}|\gameStateRandomVar_{0}\gameActionRandomVar_{0} \ldots \gameStateRandomVar_{t-1}\gameActionRandomVar_{t-1} ) \right].\] We note that for all \(t=1,2,\ldots\)
\begin{subequations}\label{stagecorrelationgeqzero}
\begin{align}
     \left[ \sum_{i=1}^{N} \entropy(\mdpStateRandomVar^{i}_{t}\mdpActionRandomVar^{i}_{t}|\gameStateRandomVar_{0}\mdpActionRandomVar^{i}_{0}  \ldots \mdpStateRandomVar^{i}_{t-1}\mdpActionRandomVar^{i}_{t-1}) \right]  - \entropy(\gameStateRandomVar_{t}\gameActionRandomVar_{t}|\gameStateRandomVar_{0}\gameActionRandomVar_{0} \ldots \gameStateRandomVar_{t-1}\gameActionRandomVar_{t-1} ) &\geq  \left[ \sum_{i=1}^{N} \entropy(\mdpStateRandomVar^{i}_{t}\mdpActionRandomVar^{i}_{t}|\gameStateRandomVar_{0}\gameActionRandomVar_{0} \ldots \gameStateRandomVar_{t-1}\gameActionRandomVar_{t-1}) \right]  - \entropy(\gameStateRandomVar_{t}\gameActionRandomVar_{t}|\gameStateRandomVar_{0}\gameActionRandomVar_{0} \ldots \gameStateRandomVar_{t-1}\gameActionRandomVar_{t-1}) \label{conditioningreducesentropy0}
     \\
     &\geq 0 \label{subaddtivityofentropy}
\end{align}
\end{subequations} where \eqref{conditioningreducesentropy0} is because conditioning (extra information) reduces entropy and \eqref{subaddtivityofentropy} is due to the subadditivity of entropy. Similarly, 
\begin{equation}  \label{stagecorrelation0geqzero}
    \left[ \sum_{i=1}^{N} \entropy(\mdpStateRandomVar^{i}_{0}\mdpActionRandomVar^{i}_{0}|\gameStateRandomVar_{0}) \right]  - \entropy(\gameStateRandomVar_{0}\gameActionRandomVar_{0}|\gameStateRandomVar_{0}) \leq 0
\end{equation}

\newpage
\begin{proof}[Proof of Lemma \ref{lemma:imaginary}]
We consider three cases of \(t_{loss}\) to prove the lemma: \(t_{loss} = 1,2, \ldots\), \(t_{loss} = 0\), and \(t_{loss} = \infty\). 

If \(t_{loss} = 0\), the statement trivially holds since \(\gamePathDist^{\imaginary}_{0} = \gamePathDist^{\imaginary}_{t_{loss}}\). In this case,  \(\kl(\gamePathDist^{\fullcommunication} || \gamePathDist^{\imaginary}_{0}) = \kl(\gamePathDist^{\fullcommunication} || \gamePathDist^{\imaginary}_{t_{loss}})\).

If \(t_{loss} = \infty\), the statement holds since there is always communication and \(\gamePathDist^{\fullcommunication} = \gamePathDist^{\imaginary}_{\infty}\). In this case,  \(\kl(\gamePathDist^{\fullcommunication} || \gamePathDist^{\imaginary}_{\infty}) = 0 \leq \kl(\gamePathDist^{\fullcommunication} || \gamePathDist^{\imaginary}_{0})\).

Let \(t_{loss}\geq 1\) be an arbitrary integer. We have 

\begin{subequations}
\begin{align}
&\kl(\gamePathDist^{\fullcommunication} || \gamePathDist^{\imaginary}_{t_{loss}}) \nonumber
\\
&= \sum_{\gameState_{0}\gameAction_{0}\gameState_{1}\gameAction_{1}\ldots \in (\gameStateSpace \times \gameActionSpace)^{\omega}} \probabilityMeasure^{\fullcommunication}(\gameState_{0}\gameAction_{0}\gameState_{1}\gameAction_{1}\ldots) \log\left(\frac{\probabilityMeasure^{\fullcommunication}(\gameState_{0}\gameAction_{0}\gameState_{1}\gameAction_{1}\ldots)}{\probabilityMeasure^{\imaginary}_{t_{loss}}(\gameState_{0}\gameAction_{0}\gameState_{1}\gameAction_{1}\ldots)}\right)
\\
&= \sum_{\genericString = \gameState_{0}\gameAction_{0}\ldots \gameState_{t_{loss}-1}\gameAction_{t_{loss}-1} \in (\gameStateSpace \times \gameActionSpace)^{t_{loss}}} \probabilityMeasure^{\fullcommunication}(\genericString) \log\left(\frac{\probabilityMeasure^{\fullcommunication}(\genericString)}{\probabilityMeasure^{\imaginary}_{t_{loss}}(\genericString)}\right) \nonumber
\\
&\quad+  \sum_{\genericString = \gameState_{0}\gameAction_{0}\ldots \gameState_{t_{loss}-1}\gameAction_{t_{loss}-1} \in (\gameStateSpace \times \gameActionSpace)^{t_{loss}}} \quad  \sum_{\genericString' =\gameState_{t_{loss}}\gameAction_{t_{loss}}\gameState_{t_{loss}}\gameAction_{t_{loss}}\ldots \in (\gameStateSpace \times \gameActionSpace)^{\omega}} \probabilityMeasure^{\fullcommunication}(\genericString\genericString') \log\left(\frac{\probabilityMeasure^{\fullcommunication}(\genericString'| \genericString)}{\probabilityMeasure^{\imaginary}_{t_{loss}}(\genericString'| \genericString)}\right) 
\\
&=\sum_{\genericString = \gameState_{0}\gameAction_{0}\ldots \gameState_{t_{loss}-1}\gameAction_{t_{loss}-1} \in (\gameStateSpace \times \gameActionSpace)^{t_{loss}}} \quad \sum_{\genericString' =\gameState_{t_{loss}}\gameAction_{t_{loss}}\gameState_{t_{loss}+1}\gameAction_{t_{loss}+1}\ldots \in (\gameStateSpace \times \gameActionSpace)^{\omega}} \probabilityMeasure^{\fullcommunication}(\genericString \genericString') \log\left(\frac{\probabilityMeasure^{\fullcommunication}(\genericString'| \genericString)}{\probabilityMeasure^{\imaginary}_{t_{loss}}(\genericString'| \genericString)}\right) 
\label{eq:jointimgeqaulbeforetloss}
\\
&=\sum_{\genericString = \gameState_{0}\gameAction_{0}\ldots \gameState_{t_{loss}-1}\gameAction_{t_{loss}-1} \in (\gameStateSpace \times \gameActionSpace)^{t_{loss}}} \quad  \sum_{\genericString' =\gameState_{t_{loss}}\gameAction_{t_{loss}}\gameState_{t_{loss}+1}\gameAction_{t_{loss}+1}\ldots \in (\gameStateSpace \times \gameActionSpace)^{\omega}} \probabilityMeasure^{\fullcommunication}(\genericString \genericString') \log\left(\frac{\probabilityMeasure^{\fullcommunication}(\genericString'| \genericString)}{\prod_{i=1}^{\numAgents}\probabilityMeasure^{\fullcommunication}_{t_{loss}}(\mdpState^{i}_{t_{loss}}\mdpAction^{i}_{t_{loss}}\mdpState^{i}_{t_{loss}+1}\mdpAction^{i}_{t_{loss}+1}\ldots|\genericString)}\right) \label{eq:imgaftertloss}
\end{align}
\end{subequations}
where \eqref{eq:jointimgeqaulbeforetloss} is because the imaginary play is the same with the joint policy for \(t = 0, \ldots, t_{loss}-1\) and \eqref{eq:imgaftertloss} is because under the imaginary play, the agents are fully independent for \(t \geq t_{loss}\).

By the definition of conditional entropy,
\begin{subequations}
\begin{align}
    &\kl(\gamePathDist^{\fullcommunication} || \gamePathDist^{\imaginary}_{t_{loss}}) \nonumber
    \\
    &=\sum_{\genericString = \gameState_{0}\gameAction_{0}\ldots \gameState_{t_{loss}-1}\gameAction_{t_{loss}-1} \in (\gameStateSpace \times \gameActionSpace)^{t_{loss}}} \quad  \sum_{\genericString' =\gameState_{t_{loss}}\gameAction_{t_{loss}}\gameState_{t_{loss}+1}\gameAction_{t_{loss}+1}\ldots \in (\gameStateSpace \times \gameActionSpace)^{\omega}} \probabilityMeasure^{\fullcommunication}(\genericString \genericString') \log\left(\frac{\probabilityMeasure^{\fullcommunication}(\genericString'| \genericString)}{\prod_{i=1}^{\numAgents}\probabilityMeasure^{\fullcommunication}_{t_{loss}}(\mdpState^{i}_{t_{loss}}\mdpAction^{i}_{t_{loss}}\mdpState^{i}_{t_{loss}+1}\mdpAction^{i}_{t_{loss}+1}\ldots|\genericString)}\right)
    \\
    &= \left[ \sum_{i=1}^{\numAgents} \entropy(\mdpStateRandomVar^{i}_{t_{loss}}\mdpActionRandomVar^{i}_{t_{loss} }\mdpStateRandomVar^{i}_{t_{loss}+1}\mdpActionRandomVar^{i}_{t_{loss}+1} \ldots| \gameStateRandomVar_{0}\gameActionRandomVar_{0}\ldots\gameStateRandomVar_{t_{loss}-1}\gameActionRandomVar_{t_{loss}-1}) \right]    \nonumber 
    \\
    &\quad -  \entropy(\gameStateRandomVar_{t_{loss}}\gameActionRandomVar_{t_{loss} }\gameStateRandomVar_{t_{loss}+1}\gameActionRandomVar_{t_{loss}+1} \ldots| \gameStateRandomVar_{0}\gameActionRandomVar_{0}\ldots\gameStateRandomVar_{t_{loss}-1}\gameActionRandomVar_{t_{loss}-1}) 
    \\
    &\leq \left[ \sum_{i=1}^{\numAgents} \entropy(\mdpStateRandomVar^{i}_{t_{loss}}\mdpActionRandomVar^{i}_{t_{loss} }\mdpStateRandomVar^{i}_{t_{loss}+1}\mdpActionRandomVar^{i}_{t_{loss}+1} \ldots| \gameStateRandomVar_{0}\mdpActionRandomVar^{i}_{0}\ldots\mdpStateRandomVar^{i}_{t_{loss}-1}\mdpActionRandomVar^{i}_{t_{loss}-1}) \right] \nonumber 
    \\
    &\quad -  \entropy(\gameStateRandomVar_{t_{loss}}\gameActionRandomVar_{t_{loss} }\gameStateRandomVar_{t_{loss}+1}\gameActionRandomVar_{t_{loss}+1} \ldots| \gameStateRandomVar_{0}\gameActionRandomVar_{0}\ldots\gameStateRandomVar_{t_{loss}-1}\gameActionRandomVar_{t_{loss}-1}) \label{conditioningreducesentropy}
    \\
    &=  \sum_{t=t_{loss}}^{\infty} \left[ \left[ \sum_{i=1}^{N} \entropy(\mdpStateRandomVar^{i}_{t}\mdpActionRandomVar^{i}_{t}|\gameStateRandomVar_{0}\mdpActionRandomVar^{i}_{0}  \ldots \mdpStateRandomVar^{i}_{t-1}\mdpActionRandomVar^{i}_{t-1}) \right]  - \entropy(\gameStateRandomVar_{t}\gameActionRandomVar_{t}|\gameStateRandomVar_{0}\gameActionRandomVar_{0} \ldots \gameStateRandomVar_{t-1}\gameActionRandomVar_{t-1} ) \right] \label{chainruleofentropy}
\end{align}
\end{subequations} where \eqref{conditioningreducesentropy} is because conditioning reduces entropy and \eqref{chainruleofentropy} is due to the chain rule of entropy. Finally, combining \eqref{stagecorrelationgeqzero},\eqref{stagecorrelation0geqzero}, and the definition of \(\totalCorrelation_{\jointPolicy}\), we have
\begin{subequations}
\begin{align}
    \kl(\gamePathDist^{\fullcommunication} || \gamePathDist^{\imaginary}_{t_{loss}}) &\leq   \sum_{t=t_{loss}}^{\infty} \left[ \left[ \sum_{i=1}^{N} \entropy(\mdpStateRandomVar^{i}_{t}\mdpActionRandomVar^{i}_{t}|\gameStateRandomVar_{0}\mdpActionRandomVar^{i}_{0}  \ldots \mdpStateRandomVar^{i}_{t-1}\mdpActionRandomVar^{i}_{t-1}) \right] - \entropy(\gameStateRandomVar_{t}\gameActionRandomVar_{t}|\gameStateRandomVar_{0}\gameActionRandomVar_{0} \ldots \gameStateRandomVar_{t-1}\gameActionRandomVar_{t-1} ) \right] 
    \\
    &\leq \left[ \sum_{i=1}^{N} \entropy(\mdpStateRandomVar^{i}_{0}\mdpActionRandomVar^{i}_{0}|\gameStateRandomVar_{0}) \right] - \entropy(\gameStateRandomVar_{0}\gameActionRandomVar_{0}|\gameStateRandomVar_{0}) + \sum_{t=0}^{\infty} \left[ \left[ \sum_{i=1}^{N} \entropy(\mdpStateRandomVar^{i}_{t}\mdpActionRandomVar^{i}_{t}|\gameStateRandomVar_{0}\mdpActionRandomVar^{i}_{0}  \ldots \mdpStateRandomVar^{i}_{t-1}\mdpActionRandomVar^{i}_{t-1}) \right]  - \entropy(\gameStateRandomVar_{t}\gameActionRandomVar_{t}|\gameStateRandomVar_{0}\gameActionRandomVar_{0} \ldots \gameStateRandomVar_{t-1}\gameActionRandomVar_{t-1} ) \right]
    \\
    &=\totalCorrelation_{\jointPolicy}
    \\
    &=\kl(\gamePathDist^{\fullcommunication} || \gamePathDist^{\imaginary}_{0}).
\end{align}
\end{subequations}

Hence, for every \(t_{loss} \in \lbrace 0, 1, \ldots \rbrace \cup \lbrace \infty \rbrace\) in Algorithm \ref{alg:imaginaryplay}, \[\kl(\gamePathDist^{\fullcommunication} || \gamePathDist^{\imaginary}_{0}) \geq  \kl(\gamePathDist^{\fullcommunication} || \gamePathDist^{\imaginary}_{t_{loss}}).\]
\end{proof}

\newpage
\begin{proof}[Proof of Lemma \ref{lemma:intermittent}]
We first show that \(\kl(\gamePathDist^{\fullcommunication} || \gamePathDist^{\imaginary}_{0}) \geq  \kl(\gamePathDist^{\fullcommunication} || \gamePathDist^{\intermittent}_{\sequenceCommAvailibility}) \) for an arbitrary sequence of \(\sequenceCommAvailibility = \oneStepCommAvailibility_{0}, \oneStepCommAvailibility_{1},\ldots\) communication availability. 

Define \(\oneStepCommAvailibility_{-1} = 1\). Let \(l_{j}\) denote the starting time index of \(j\)-th period that communication is not available. Formally, \(l_{1} = \min\lbrace i|\oneStepCommAvailibility_{i}=0, \oneStepCommAvailibility_{i-1}=1, i\geq0\rbrace\) and \(l_{j} = \min\lbrace i|\oneStepCommAvailibility_{i}=0, \oneStepCommAvailibility_{i-1}=1, i> l_{j-1}\rbrace\) for all \(j\geq2\). Similarly, let \(r_{j}\) denote the starting time index of \(j\)-th period that communication is available again. Formally, \(r_{1} = \min\lbrace i|\oneStepCommAvailibility_{i}=1, \oneStepCommAvailibility_{i-1}=0, i\geq0\rbrace\) and \(r_{j} = \min\lbrace i|\oneStepCommAvailibility_{i}=1, \oneStepCommAvailibility_{i-1}=0, i> r_{j-1}\rbrace\) for all \(j\geq2\). For example, for \(\sequenceCommAvailibility = \oneStepCommAvailibility_{0}, \oneStepCommAvailibility_{1},\ldots = 0,1,1,0,0,0,1,1, \ldots\), we have \(l_{1} = 0\), \(r_{1} = 1\), \(l_{2} = 3\), and \(r_{2} = 6\). For \(\sequenceCommAvailibility = \oneStepCommAvailibility_{0}, \oneStepCommAvailibility_{1},\ldots = 1,1,0,0,0,1, \ldots\), we have \(l_{1} = 2\) and \(r_{1} = 5\).

We consider two different cases of \(\oneStepCommAvailibility_{0}\) separately. First, assume that \(\oneStepCommAvailibility_{0}= 1\), i.e., the communication is available at time \(0\).
\begin{subequations}
\begin{align}
    \kl(\gamePathDist^{\fullcommunication} || \gamePathDist^{\intermittent}_{\sequenceCommAvailibility}) &= \sum_{\gameState_{0}\gameAction_{0}\gameState_{1}\gameAction_{1}\ldots \in (\gameStateSpace \times \gameActionSpace)^{\omega}} \probabilityMeasure^{\fullcommunication}(\gameState_{0}\gameAction_{0}\gameState_{1}\gameAction_{1}\ldots) \log\left(\frac{\probabilityMeasure^{\fullcommunication}(\gameState_{0}\gameAction_{0}\gameState_{1}\gameAction_{1}\ldots)}{\probabilityMeasure^{\intermittent}_{\sequenceCommAvailibility}(\gameState_{0}\gameAction_{0}\gameState_{1}\gameAction_{1}\ldots)}\right)
    \\
    &= \sum_{\gameState_{0}\gameAction_{0}\gameState_{1}\gameAction_{1}\ldots \in (\gameStateSpace \times \gameActionSpace)^{\omega}}  \probabilityMeasure^{\fullcommunication}(\gameState_{0}\gameAction_{0}\gameState_{1}\gameAction_{1}\ldots) \log\left(\frac{\probabilityMeasure^{\fullcommunication}(\gameState_{0}\gameAction_{0}\gameState_{1}\gameAction_{1}\ldots\gameState_{l_{1}-1}\gameAction_{l_{1}-1})}{\probabilityMeasure^{\intermittent}_{\sequenceCommAvailibility}(\gameState_{0}\gameAction_{0}\gameState_{1}\gameAction_{1}\ldots\gameState_{l_{1}-1}\gameAction_{l_{1}-1})}\right) \nonumber 
    \\
    &\quad +\sum_{j=1}^{\infty} \quad \sum_{\gameState_{0}\gameAction_{0}\gameState_{1}\gameAction_{1}\ldots \in (\gameStateSpace \times \gameActionSpace)^{\omega}}  \probabilityMeasure^{\fullcommunication}(\gameState_{0}\gameAction_{0}\gameState_{1}\gameAction_{1}\ldots) \log\left(\frac{\probabilityMeasure^{\fullcommunication}(\gameState_{l_{j}}\gameAction_{l_{j}}\gameState_{l_{j}+1}\gameAction_{l_{j}+1}\ldots\gameState_{l_{j+1}-1}\gameAction_{l_{j+1}-1}|\gameState_{0}\gameAction_{0}\gameState_{1}\gameAction_{1}\ldots\gameState_{l_{j}-1}\gameAction_{l_{j}-1})}{\probabilityMeasure^{\intermittent}_{\sequenceCommAvailibility}(\gameState_{l_{j}}\gameAction_{l_{j}}\gameState_{l_{j}+1}\gameAction_{l_{j}+1}\ldots\gameState_{l_{j+1}-1}\gameAction_{l_{j+1}-1}|\gameState_{0}\gameAction_{0}\gameState_{1}\gameAction_{1}\ldots\gameState_{l_{j}-1}\gameAction_{l_{j}-1})}\right) 
    \\
    &= \sum_{\gameState_{0}\gameAction_{0}\gameState_{1}\gameAction_{1}\ldots \in (\gameStateSpace \times \gameActionSpace)^{\omega}}  \probabilityMeasure^{\fullcommunication}(\gameState_{0}\gameAction_{0}\gameState_{1}\gameAction_{1}\ldots) \log\left(\frac{\probabilityMeasure^{\fullcommunication}(\gameState_{0}\gameAction_{0}\gameState_{1}\gameAction_{1}\ldots\gameState_{l_{1}-1}\gameAction_{l_{1}-1})}{\probabilityMeasure^{\intermittent}_{\sequenceCommAvailibility}(\gameState_{0}\gameAction_{0}\gameState_{1}\gameAction_{1}\ldots\gameState_{l_{1}-1}\gameAction_{l_{1}-1})}\right) \nonumber 
    \\
    &\quad +\sum_{j=1}^{\infty} \quad \sum_{\gameState_{0}\gameAction_{0}\gameState_{1}\gameAction_{1}\ldots \in (\gameStateSpace \times \gameActionSpace)^{\omega}}  \probabilityMeasure^{\fullcommunication}(\gameState_{0}\gameAction_{0}\gameState_{1}\gameAction_{1}\ldots) \log\left(\frac{\probabilityMeasure^{\fullcommunication}(\gameState_{l_{j}}\gameAction_{l_{j}}\gameState_{l_{j}+1}\gameAction_{l_{j}+1}\ldots\gameState_{r_{j}-1}\gameAction_{r_{j}-1}|\gameState_{0}\gameAction_{0}\gameState_{1}\gameAction_{1}\ldots\gameState_{l_{j}-1}\gameAction_{l_{j}-1})}{\probabilityMeasure^{\intermittent}_{\sequenceCommAvailibility}(\gameState_{l_{j}}\gameAction_{l_{j}}\gameState_{l_{j}+1}\gameAction_{l_{j}+1}\ldots\gameState_{r_{j}-1}\gameAction_{r_{j}-1}|\gameState_{0}\gameAction_{0}\gameState_{1}\gameAction_{1}\ldots\gameState_{l_{j}-1}\gameAction_{l_{j}-1})}\right) \nonumber 
    \\
    &\quad +\sum_{j=1}^{\infty} \quad \sum_{\gameState_{0}\gameAction_{0}\gameState_{1}\gameAction_{1}\ldots \in (\gameStateSpace \times \gameActionSpace)^{\omega}}  \probabilityMeasure^{\fullcommunication}(\gameState_{0}\gameAction_{0}\gameState_{1}\gameAction_{1}\ldots) \log\left(\frac{\probabilityMeasure^{\fullcommunication}(\gameState_{r_{j}}\gameAction_{r_{j}}\gameState_{r_{j}+1}\gameAction_{r_{j}+1}\ldots\gameState_{l_{j+1}-1}\gameAction_{l_{j+1}-1}|\gameState_{0}\gameAction_{0}\gameState_{1}\gameAction_{1}\ldots\gameState_{r_{j}-1}\gameAction_{r_{j}-1})}{\probabilityMeasure^{\intermittent}_{\sequenceCommAvailibility}(\gameState_{r_{j}}\gameAction_{r_{j}}\gameState_{r_{j}+1}\gameAction_{r_{j}+1}\ldots\gameState_{l_{j+1}-1}\gameAction_{l_{j+1}-1}|\gameState_{0}\gameAction_{0}\gameState_{1}\gameAction_{1}\ldots\gameState_{r_{j}-1}\gameAction_{r_{j}-1})}\right).
\end{align}
\end{subequations}
We note that when the communication is available the state-action process under the intermittent communication and  the state-action process under the joint policy with full communication follow the same Markov chain. Also note that communication is available between \([0,l_{1}]\) and \([r_{j},l_{j+1}-1]\) for all \(j\geq 1\). Consequently, 
\begin{subequations}
\begin{align}
    \kl(\gamePathDist^{\fullcommunication} || \gamePathDist^{\intermittent}_{\sequenceCommAvailibility})
    &=\sum_{\gameState_{0}\gameAction_{0}\gameState_{1}\gameAction_{1}\ldots \in (\gameStateSpace \times \gameActionSpace)^{\omega}}  \probabilityMeasure^{\fullcommunication}(\gameState_{0}\gameAction_{0}\gameState_{1}\gameAction_{1}\ldots) \log\left(\frac{\probabilityMeasure^{\fullcommunication}(\gameState_{0}\gameAction_{0}\gameState_{1}\gameAction_{1}\ldots\gameState_{l_{1}-1}\gameAction_{l_{1}-1})}{\probabilityMeasure^{\fullcommunication}(\gameState_{0}\gameAction_{0}\gameState_{1}\gameAction_{1}\ldots\gameState_{l_{1}-1}\gameAction_{l_{1}-1})}\right) \nonumber 
    \\
    &\quad +\sum_{j=1}^{\infty} \quad \sum_{\gameState_{0}\gameAction_{0}\gameState_{1}\gameAction_{1}\ldots \in (\gameStateSpace \times \gameActionSpace)^{\omega}}  \probabilityMeasure^{\fullcommunication}(\gameState_{0}\gameAction_{0}\gameState_{1}\gameAction_{1}\ldots) \log\left(\frac{\probabilityMeasure^{\fullcommunication}(\gameState_{l_{j}}\gameAction_{l_{j}}\gameState_{l_{j}+1}\gameAction_{l_{j}+1}\ldots\gameState_{r_{j}-1}\gameAction_{r_{j}-1}|\gameState_{0}\gameAction_{0}\gameState_{1}\gameAction_{1}\ldots\gameState_{l_{j}-1}\gameAction_{l_{j}-1})}{\probabilityMeasure^{\intermittent}_{\sequenceCommAvailibility}(\gameState_{l_{j}}\gameAction_{l_{j}}\gameState_{l_{j}+1}\gameAction_{l_{j}+1}\ldots\gameState_{r_{j}-1}\gameAction_{r_{j}-1}|\gameState_{0}\gameAction_{0}\gameState_{1}\gameAction_{1}\ldots\gameState_{l_{j}-1}\gameAction_{l_{j}-1})}\right) \nonumber 
    \\
    &\quad +\sum_{j=1}^{\infty} \quad \sum_{\gameState_{0}\gameAction_{0}\gameState_{1}\gameAction_{1}\ldots \in (\gameStateSpace \times \gameActionSpace)^{\omega}}  \probabilityMeasure^{\fullcommunication}(\gameState_{0}\gameAction_{0}\gameState_{1}\gameAction_{1}\ldots) \log\left(\frac{\probabilityMeasure^{\fullcommunication}(\gameState_{r_{j}}\gameAction_{r_{j}}\gameState_{r_{j}+1}\gameAction_{r_{j}+1}\ldots\gameState_{l_{j+1}-1}\gameAction_{l_{j+1}-1}|\gameState_{0}\gameAction_{0}\gameState_{1}\gameAction_{1}\ldots\gameState_{r_{j}-1}\gameAction_{r_{j}-1})}{\probabilityMeasure^{\fullcommunication}(\gameState_{r_{j}}\gameAction_{r_{j}}\gameState_{r_{j}+1}\gameAction_{r_{j}+1}\ldots\gameState_{l_{j+1}-1}\gameAction_{l_{j+1}-1}|\gameState_{0}\gameAction_{0}\gameState_{1}\gameAction_{1}\ldots\gameState_{r_{j}-1}\gameAction_{r_{j}-1})}\right)
    \\
    &=\sum_{j=1}^{\infty} \quad \sum_{\gameState_{0}\gameAction_{0}\gameState_{1}\gameAction_{1}\ldots \in (\gameStateSpace \times \gameActionSpace)^{\omega}}  \probabilityMeasure^{\fullcommunication}(\gameState_{0}\gameAction_{0}\gameState_{1}\gameAction_{1}\ldots) \log\left(\frac{\probabilityMeasure^{\fullcommunication}(\gameState_{l_{j}}\gameAction_{l_{j}}\gameState_{l_{j}+1}\gameAction_{l_{j}+1}\ldots\gameState_{r_{j}-1}\gameAction_{r_{j}-1}|\gameState_{0}\gameAction_{0}\gameState_{1}\gameAction_{1}\ldots\gameState_{l_{j}-1}\gameAction_{l_{j}-1})}{\probabilityMeasure^{\intermittent}_{\sequenceCommAvailibility}(\gameState_{l_{j}}\gameAction_{l_{j}}\gameState_{l_{j}+1}\gameAction_{l_{j}+1}\ldots\gameState_{r_{j}-1}\gameAction_{r_{j}-1}|\gameState_{0}\gameAction_{0}\gameState_{1}\gameAction_{1}\ldots\gameState_{l_{j}-1}\gameAction_{l_{j}-1})}\right).
\end{align}
\end{subequations}

By the same arguments,  when \(\oneStepCommAvailibility_{0}= 0\), i.e., the communication is not available at time \(0\), 
\begin{subequations}
\begin{align}
    \kl(\gamePathDist^{\fullcommunication} || \gamePathDist^{\intermittent}_{\sequenceCommAvailibility}) &= \sum_{\gameState_{0}\gameAction_{0}\gameState_{1}\gameAction_{1}\ldots \in (\gameStateSpace \times \gameActionSpace)^{\omega}} \probabilityMeasure^{\fullcommunication}(\gameState_{0}\gameAction_{0}\gameState_{1}\gameAction_{1}\ldots) \log\left(\frac{\probabilityMeasure^{\fullcommunication}(\gameState_{0}\gameAction_{0}\gameState_{1}\gameAction_{1}\ldots)}{\probabilityMeasure^{\intermittent}_{\sequenceCommAvailibility}(\gameState_{0}\gameAction_{0}\gameState_{1}\gameAction_{1}\ldots)}\right)
    \\
    &=\sum_{j=1}^{\infty} \quad \sum_{\gameState_{0}\gameAction_{0}\gameState_{1}\gameAction_{1}\ldots \in (\gameStateSpace \times \gameActionSpace)^{\omega}}  \probabilityMeasure^{\fullcommunication}(\gameState_{0}\gameAction_{0}\gameState_{1}\gameAction_{1}\ldots) \log\left(\frac{\probabilityMeasure^{\fullcommunication}(\gameState_{l_{j}}\gameAction_{l_{j}}\gameState_{l_{j}+1}\gameAction_{l_{j}+1}\ldots\gameState_{r_{j}-1}\gameAction_{r_{j}-1}|\gameState_{0}\gameAction_{0}\gameState_{1}\gameAction_{1}\ldots\gameState_{l_{j}-1}\gameAction_{l_{j}-1})}{\probabilityMeasure^{\intermittent}_{\sequenceCommAvailibility}(\gameState_{l_{j}}\gameAction_{l_{j}}\gameState_{l_{j}+1}\gameAction_{l_{j}+1}\ldots\gameState_{r_{j}-1}\gameAction_{r_{j}-1}|\gameState_{0}\gameAction_{0}\gameState_{1}\gameAction_{1}\ldots\gameState_{l_{j}-1}\gameAction_{l_{j}-1})}\right) \nonumber 
    \\
    &\quad +\sum_{j=1}^{\infty} \quad \sum_{\gameState_{0}\gameAction_{0}\gameState_{1}\gameAction_{1}\ldots \in (\gameStateSpace \times \gameActionSpace)^{\omega}}  \probabilityMeasure^{\fullcommunication}(\gameState_{0}\gameAction_{0}\gameState_{1}\gameAction_{1}\ldots) \log\left(\frac{\probabilityMeasure^{\fullcommunication}(\gameState_{r_{j}}\gameAction_{r_{j}}\gameState_{r_{j}+1}\gameAction_{r_{j}+1}\ldots\gameState_{l_{j+1}-1}\gameAction_{l_{j+1}-1}|\gameState_{0}\gameAction_{0}\gameState_{1}\gameAction_{1}\ldots\gameState_{r_{j}-1}\gameAction_{r_{j}-1})}{\probabilityMeasure^{\intermittent}_{\sequenceCommAvailibility}(\gameState_{r_{j}}\gameAction_{r_{j}}\gameState_{r_{j}+1}\gameAction_{r_{j}+1}\ldots\gameState_{l_{j+1}-1}\gameAction_{l_{j+1}-1}|\gameState_{0}\gameAction_{0}\gameState_{1}\gameAction_{1}\ldots\gameState_{r_{j}-1}\gameAction_{r_{j}-1})}\right)
    \\
    &=\sum_{j=1}^{\infty} \quad \sum_{\gameState_{0}\gameAction_{0}\gameState_{1}\gameAction_{1}\ldots \in (\gameStateSpace \times \gameActionSpace)^{\omega}}  \probabilityMeasure^{\fullcommunication}(\gameState_{0}\gameAction_{0}\gameState_{1}\gameAction_{1}\ldots) \log\left(\frac{\probabilityMeasure^{\fullcommunication}(\gameState_{l_{j}}\gameAction_{l_{j}}\gameState_{l_{j}+1}\gameAction_{l_{j}+1}\ldots\gameState_{r_{j}-1}\gameAction_{r_{j}-1}|\gameState_{0}\gameAction_{0}\gameState_{1}\gameAction_{1}\ldots\gameState_{l_{j}-1}\gameAction_{l_{j}-1})}{\probabilityMeasure^{\intermittent}_{\sequenceCommAvailibility}(\gameState_{l_{j}}\gameAction_{l_{j}}\gameState_{l_{j}+1}\gameAction_{l_{j}+1}\ldots\gameState_{r_{j}-1}\gameAction_{r_{j}-1}|\gameState_{0}\gameAction_{0}\gameState_{1}\gameAction_{1}\ldots\gameState_{l_{j}-1}\gameAction_{l_{j}-1})}\right).
\end{align}
\end{subequations}

Hence, for every value of \(\oneStepCommAvailibility_{0}\),
\[   \kl(\gamePathDist^{\fullcommunication} || \gamePathDist^{\intermittent}_{\sequenceCommAvailibility})=\sum_{j=1}^{\infty} \quad \sum_{\gameState_{0}\gameAction_{0}\gameState_{1}\gameAction_{1}\ldots \in (\gameStateSpace \times \gameActionSpace)^{\omega}}  \probabilityMeasure^{\fullcommunication}(\gameState_{0}\gameAction_{0}\gameState_{1}\gameAction_{1}\ldots) \log\left(\frac{\probabilityMeasure^{\fullcommunication}(\gameState_{l_{j}}\gameAction_{l_{j}}\gameState_{l_{j}+1}\gameAction_{l_{j}+1}\ldots\gameState_{r_{j}-1}\gameAction_{r_{j}-1}|\gameState_{0}\gameAction_{0}\gameState_{1}\gameAction_{1}\ldots\gameState_{l_{j}-1}\gameAction_{l_{j}-1})}{\probabilityMeasure^{\intermittent}_{\sequenceCommAvailibility}(\gameState_{l_{j}}\gameAction_{l_{j}}\gameState_{l_{j}+1}\gameAction_{l_{j}+1}\ldots\gameState_{r_{j}-1}\gameAction_{r_{j}-1}|\gameState_{0}\gameAction_{0}\gameState_{1}\gameAction_{1}\ldots\gameState_{l_{j}-1}\gameAction_{l_{j}-1})}\right).\]

Since the policy is stationary and the agents are fully independent between \([l_{j}, r_{j}-1]\) for all \(j \geq 0\), we have 
\begin{subequations}
\begin{align}
    \kl(\gamePathDist^{\fullcommunication} || \gamePathDist^{\intermittent}_{\sequenceCommAvailibility})&=\sum_{j=1}^{\infty} \quad \sum_{\gameState_{0}\gameAction_{0}\gameState_{1}\gameAction_{1}\ldots \in (\gameStateSpace \times \gameActionSpace)^{\omega}}  \probabilityMeasure^{\fullcommunication}(\gameState_{0}\gameAction_{0}\gameState_{1}\gameAction_{1}\ldots) \log\left(\frac{\probabilityMeasure^{\fullcommunication}(\gameState_{l_{j}}\gameAction_{l_{j}}\gameState_{l_{j}+1}\gameAction_{l_{j}+1}\ldots\gameState_{r_{j}-1}\gameAction_{r_{j}-1}|\gameState_{0}\gameAction_{0}\gameState_{1}\gameAction_{1}\ldots\gameState_{l_{j}-1}\gameAction_{l_{j}-1})}{\probabilityMeasure^{\intermittent}_{\sequenceCommAvailibility}(\gameState_{l_{j}}\gameAction_{l_{j}}\gameState_{l_{j}+1}\gameAction_{l_{j}+1}\ldots\gameState_{r_{j}-1}\gameAction_{r_{j}-1}|\gameState_{l_{j}-1}\gameAction_{l_{j}-1})}\right)
    \\
    &=\sum_{j=1}^{\infty} \quad \sum_{\gameState_{0}\gameAction_{0}\gameState_{1}\gameAction_{1}\ldots \in (\gameStateSpace \times \gameActionSpace)^{\omega}}  \probabilityMeasure^{\fullcommunication}(\gameState_{0}\gameAction_{0}\gameState_{1}\gameAction_{1}\ldots) \log\left(\frac{\probabilityMeasure^{\fullcommunication}(\gameState_{l_{j}}\gameAction_{l_{j}}\gameState_{l_{j}+1}\gameAction_{l_{j}+1}\ldots\gameState_{r_{j}-1}\gameAction_{r_{j}-1}|\gameState_{0}\gameAction_{0}\gameState_{1}\gameAction_{1}\ldots\gameState_{l_{j}-1}\gameAction_{l_{j}-1})}{\prod_{i=1}^{\numAgents}\probabilityMeasure^{\fullcommunication}(\mdpState^{i}_{l_{j}}\mdpAction^{i}_{l_{j}}\mdpState^{i}_{l_{j}+1}\mdpAction^{i}_{l_{j}+1}\ldots\mdpState^{i}_{r_{j}-1}\mdpAction^{i}_{r_{j}-1}|\gameState_{l_{j}-1}\gameAction_{l_{j}-1})}\right).
\end{align}
\end{subequations}

By the definition of conditional entropy, we have
\begin{subequations}
\begin{align}
    \kl(\gamePathDist^{\fullcommunication} || \gamePathDist^{\intermittent}_{\sequenceCommAvailibility})&=\sum_{j=1}^{\infty}  \left[ \left[ \sum_{i=1}^{N} \entropy(\mdpStateRandomVar^{i}_{l_{j}}\mdpActionRandomVar^{i}_{l_{j}}\mdpStateRandomVar^{i}_{l_{j}+1}\mdpActionRandomVar^{i}_{l_{j}+1}\ldots\mdpStateRandomVar^{i}_{r_{j}-1}\mdpActionRandomVar^{i}_{r_{j}-1}|\gameStateRandomVar_{l_{j}-1}\gameActionRandomVar_{l_{j}-1}) \right]  - \entropy(\gameStateRandomVar_{l_{j}}\gameActionRandomVar_{l_{j}}\gameStateRandomVar_{l_{j}+1}\gameActionRandomVar_{l_{j}+1}\ldots\gameStateRandomVar_{r_{j}-1}\gameActionRandomVar_{r_{j}-1}|\gameStateRandomVar_{0}\gameActionRandomVar_{0} \ldots \gameStateRandomVar_{l_{j}-1}\gameActionRandomVar_{l_{j}-1} ) \right]
    \\
    &=\sum_{j=1}^{\infty}  \left[ \left[ \sum_{i=1}^{N} \entropy(\mdpStateRandomVar^{i}_{l_{j}}\mdpActionRandomVar^{i}_{l_{j}}\mdpStateRandomVar^{i}_{l_{j}+1}\mdpActionRandomVar^{i}_{l_{j}+1}\ldots\mdpStateRandomVar^{i}_{r_{j}-1}\mdpActionRandomVar^{i}_{r_{j}-1}|\gameStateRandomVar_{0}\mdpActionRandomVar^{i}_{0}\mdpStateRandomVar^{i}_{1}\mdpActionRandomVar^{i}_{1}\ldots\mdpStateRandomVar^{i}_{l_{j}-2}\mdpActionRandomVar^{i}_{l_{j}-2} \gameStateRandomVar_{l_{j}-1}\gameActionRandomVar_{l_{j}-1}) \right] \right. \nonumber 
    \\
    &\quad - \left. \entropy(\gameStateRandomVar_{l_{j}}\gameActionRandomVar_{l_{j}}\gameStateRandomVar_{l_{j}+1}\gameActionRandomVar_{l_{j}+1}\ldots\gameStateRandomVar_{r_{j}-1}\gameActionRandomVar_{r_{j}-1}|\gameStateRandomVar_{0}\gameActionRandomVar_{0} \ldots \gameStateRandomVar_{l_{j}-1}\gameActionRandomVar_{l_{j}-1} )\right].
\end{align}
\end{subequations} where the second equality is due to the stationarity of \(\jointPolicy\), i.e., \(\mdpStateRandomVar^{i}_{l_{j}}\mdpActionRandomVar^{i}_{l_{j}}\mdpStateRandomVar^{i}_{l_{j}+1}\mdpActionRandomVar^{i}_{l_{j}+1}\ldots\mdpStateRandomVar^{i}_{r_{j}-1}\mdpActionRandomVar^{i}_{r_{j}-1}\) is independent of \(\gameStateRandomVar_{0}\mdpActionRandomVar^{i}_{0}\mdpStateRandomVar^{i}_{1}\mdpActionRandomVar^{i}_{1}\ldots\mdpStateRandomVar^{i}_{l_{j}-2}\mdpActionRandomVar^{i}_{l_{j}-2}\) given \(\gameStateRandomVar_{l_{j}-1}\gameActionRandomVar_{l_{j}-1}\).

Since conditioning reduces entropy,
\begin{subequations}
\begin{align}
    \kl(\gamePathDist^{\fullcommunication} || \gamePathDist^{\intermittent}_{\sequenceCommAvailibility}) &=\sum_{j=1}^{\infty}  \left[ \left[ \sum_{i=1}^{N} \entropy(\mdpStateRandomVar^{i}_{l_{j}}\mdpActionRandomVar^{i}_{l_{j}}\mdpStateRandomVar^{i}_{l_{j}+1}\mdpActionRandomVar^{i}_{l_{j}+1}\ldots\mdpStateRandomVar^{i}_{r_{j}-1}\mdpActionRandomVar^{i}_{r_{j}-1}|\gameStateRandomVar_{0}\mdpActionRandomVar^{i}_{0}\mdpStateRandomVar^{i}_{1}\mdpActionRandomVar^{i}_{1}\ldots \gameStateRandomVar_{l_{j}-1}\gameActionRandomVar_{l_{j}-1}) \right] \right. \nonumber 
    \\
    &\quad - \left. \entropy(\gameStateRandomVar_{l_{j}}\gameActionRandomVar_{l_{j}}\gameStateRandomVar_{l_{j}+1}\gameActionRandomVar_{l_{j}+1}\ldots\gameStateRandomVar_{r_{j}-1}\gameActionRandomVar_{r_{j}-1}|\gameStateRandomVar_{0}\gameActionRandomVar_{0} \ldots \gameStateRandomVar_{l_{j}-1}\gameActionRandomVar_{l_{j}-1} )\right]
    \\
     &\leq\sum_{j=1}^{\infty}  \left[ \left[ \sum_{i=1}^{N} \entropy(\mdpStateRandomVar^{i}_{l_{j}}\mdpActionRandomVar^{i}_{l_{j}}\mdpStateRandomVar^{i}_{l_{j}+1}\mdpActionRandomVar^{i}_{l_{j}+1}\ldots\mdpStateRandomVar^{i}_{r_{j}-1}\mdpActionRandomVar^{i}_{r_{j}-1}|\gameStateRandomVar_{0}\mdpActionRandomVar^{i}_{0}\mdpStateRandomVar^{i}_{1}\mdpActionRandomVar^{i}_{1}\ldots \mdpStateRandomVar^{i}_{l_{j}-1}\mdpActionRandomVar^{i}_{l_{j}-1}) \right] \right. \nonumber 
    \\
    &\quad - \left. \entropy(\gameStateRandomVar_{l_{j}}\gameActionRandomVar_{l_{j}}\gameStateRandomVar_{l_{j}+1}\gameActionRandomVar_{l_{j}+1}\ldots\gameStateRandomVar_{r_{j}-1}\gameActionRandomVar_{r_{j}-1}|\gameStateRandomVar_{0}\gameActionRandomVar_{0} \ldots \gameStateRandomVar_{l_{j}-1}\gameActionRandomVar_{l_{j}-1} )\right]
    \\
    &=\sum_{j=1}^{\infty} \sum_{t=l_{j}}^{r_{j}-1} \left[ \left[ \sum_{i=1}^{N} \entropy(\mdpStateRandomVar^{i}_{t}\mdpActionRandomVar^{i}_{t}|\gameStateRandomVar_{0}\mdpActionRandomVar^{i}_{0}  \ldots \mdpStateRandomVar^{i}_{t-1}\mdpActionRandomVar^{i}_{t-1}) \right]  - \entropy(\gameStateRandomVar_{t}\gameActionRandomVar_{t}|\gameStateRandomVar_{0}\gameActionRandomVar_{0} \ldots \gameStateRandomVar_{t-1}\gameActionRandomVar_{t-1} )\right]
\end{align}
\end{subequations}
where the last equality is due to the definition of conditional entropy.

We note that 
\begin{subequations}
\begin{align}
    \kl(\gamePathDist^{\fullcommunication} || \gamePathDist^{\intermittent}_{\sequenceCommAvailibility}) &\leq \sum_{j=1}^{\infty} \sum_{t=l_{j}}^{r_{j}-1} \left[ \left[ \sum_{i=1}^{N} \entropy(\mdpStateRandomVar^{i}_{t}\mdpActionRandomVar^{i}_{t}|\gameStateRandomVar_{0}\mdpActionRandomVar^{i}_{0}  \ldots \mdpStateRandomVar^{i}_{t-1}\mdpActionRandomVar^{i}_{t-1}) \right]  - \entropy(\gameStateRandomVar_{t}\gameActionRandomVar_{t}|\gameStateRandomVar_{0}\gameActionRandomVar_{0} \ldots \gameStateRandomVar_{t-1}\gameActionRandomVar_{t-1} )\right]
    \\
    &= \sum_{t=0}^{\infty}(1-\oneStepCommAvailibility_{t}) \left[ \left[ \sum_{i=1}^{N} \entropy(\mdpStateRandomVar^{i}_{t}\mdpActionRandomVar^{i}_{t}|\gameStateRandomVar_{0}\mdpActionRandomVar^{i}_{0}  \ldots \mdpStateRandomVar^{i}_{t-1}\mdpActionRandomVar^{i}_{t-1}) \right]  - \entropy(\gameStateRandomVar_{t}\gameActionRandomVar_{t}|\gameStateRandomVar_{0}\gameActionRandomVar_{0} \ldots \gameStateRandomVar_{t-1}\gameActionRandomVar_{t-1} )\right].
\end{align}
\end{subequations}

Combining \eqref{stagecorrelationgeqzero},\eqref{stagecorrelation0geqzero}, and the definition of \(\totalCorrelation_{\jointPolicy}\), we have
\begin{subequations}
\begin{align}
    \kl(\gamePathDist^{\fullcommunication} || \gamePathDist^{\intermittent}_{\sequenceCommAvailibility}) &\leq \sum_{t=0}^{\infty}(1-\oneStepCommAvailibility_{t}) \left[ \left[ \sum_{i=1}^{N} \entropy(\mdpStateRandomVar^{i}_{t}\mdpActionRandomVar^{i}_{t}|\gameStateRandomVar_{0}\mdpActionRandomVar^{i}_{0}  \ldots \mdpStateRandomVar^{i}_{t-1}\mdpActionRandomVar^{i}_{t-1}) \right]  - \entropy(\gameStateRandomVar_{t}\gameActionRandomVar_{t}|\gameStateRandomVar_{0}\gameActionRandomVar_{0} \ldots \gameStateRandomVar_{t-1}\gameActionRandomVar_{t-1} )\right] \label{klasasumofstepcorr}
    \\
    &\leq \kl(\gamePathDist^{\fullcommunication} || \gamePathDist^{\imaginary}_{0}).
\end{align}
\end{subequations}

We now show that if every \(\oneStepCommAvailibility_{t}\) is independently sampled from a Bernoulli random variable with parameter \(\probabilityFailureOneStep\),  \[\kl(\gamePathDist^{\fullcommunication} || \gamePathDist^{\imaginary}) \geq  \kl(\gamePathDist^{\fullcommunication} || \gamePathDist^{\intermittent})/\probabilityFailureOneStep \] where  \(\gamePathDist^{\intermittent} = \expectation_{\sequenceCommAvailibility}\left[\gamePathDist^{\intermittent}_{\sequenceCommAvailibility}\right]\). 

By the convexity of KL divergence~\cite{boyd2004convex} and \eqref{klasasumofstepcorr}, Jensen's inequality~\cite{boyd2004convex} yields
\begin{subequations}
\begin{align}
    \kl(\gamePathDist^{\fullcommunication} || \gamePathDist^{\intermittent}) &= \kl(\gamePathDist^{\fullcommunication} || \expectation_{\sequenceCommAvailibility}\left[\gamePathDist^{\intermittent}_{\sequenceCommAvailibility}\right]) 
    \\
    &\leq \expectation_{\sequenceCommAvailibility}\left[\kl(\gamePathDist^{\fullcommunication} || \gamePathDist^{\intermittent}_{\sequenceCommAvailibility}]) \right]
    \\
    &\leq \expectation_{\sequenceCommAvailibility}\left[  \sum_{t=0}^{\infty}(1-\oneStepCommAvailibility_{t}) \left[ \left[ \sum_{i=1}^{N} \entropy(\mdpStateRandomVar^{i}_{t}\mdpActionRandomVar^{i}_{t}|\gameStateRandomVar_{0}\mdpActionRandomVar^{i}_{0}  \ldots \mdpStateRandomVar^{i}_{t-1}\mdpActionRandomVar^{i}_{t-1}) \right]  - \entropy(\gameStateRandomVar_{t}\gameActionRandomVar_{t}|\gameStateRandomVar_{0}\gameActionRandomVar_{0} \ldots \gameStateRandomVar_{t-1}\gameActionRandomVar_{t-1} )\right] \right]
        \\
    &=  q \sum_{t=0}^{\infty} \left[ \left[ \sum_{i=1}^{N} \entropy(\mdpStateRandomVar^{i}_{t}\mdpActionRandomVar^{i}_{t}|\gameStateRandomVar_{0}\mdpActionRandomVar^{i}_{0}  \ldots \mdpStateRandomVar^{i}_{t-1}\mdpActionRandomVar^{i}_{t-1}) \right]  - \entropy(\gameStateRandomVar_{t}\gameActionRandomVar_{t}|\gameStateRandomVar_{0}\gameActionRandomVar_{0} \ldots \gameStateRandomVar_{t-1}\gameActionRandomVar_{t-1} )\right]
    \\
    &= q \kl(\gamePathDist^{\fullcommunication} || \gamePathDist^{\imaginary}_{0})
\end{align}
\end{subequations}
where the last equalities are due to the linearity of expectation and the independence of \(\oneStepCommAvailibility_{t}\) values from the state-action processes. Rearranging the terms yields \[\kl(\gamePathDist^{\fullcommunication} || \gamePathDist^{\imaginary}_{0}) \geq  \kl(\gamePathDist^{\fullcommunication} || \gamePathDist^{\intermittent})/\probabilityFailureOneStep. \]
\end{proof}

\newpage
\begin{proof}[Proof of Lemma \ref{lemma:imaginarystronger}] 
The proof is similar to the proof of Lemma 7.1.

If \(\genericFunction(\emptyString) = 0\), i.e., communication is not available at time \(0\), then the agents use the imaginary play for the whole path, i.e., \(t_{loss} = 0\), and the distribution \(\gamePathDist^{\imaginary}_{\genericFunction}\) of paths is the same as \(\gamePathDist^{\imaginary}_{0}\). Then, \[\kl(\gamePathDist^{\fullcommunication} || \gamePathDist^{\imaginary}_{0}) = \kl(\gamePathDist^{\fullcommunication} || \gamePathDist^{\imaginary}_{\genericFunction}).\] Without loss of generality, we will assume \(\genericFunction(\emptyString) \not= 0\) for the rest of the proof.

We first define some sets for ease of notation. Let  \(\genericString = \gameState_{0}\gameAction_{0}\gameState_{1}\gameAction_{1}\ldots \gameState_{t}\gameAction_{t}\) be a finite state-action sequence. Define \(Pref(\genericString)\) as the set of all strict prefixes of \(\genericString\) such that \(Pref(\genericString =  \gameState_{0}\gameAction_{0}\gameState_{1}\gameAction_{1}\ldots \gameState_{t}\gameAction_{t}) = \lbrace  \gameState_{0}\gameAction_{0}\gameState_{1}\gameAction_{1}\ldots \gameState_{t'}\gameAction_{t'}| t' = 0,\ldots, t-1 \rbrace \). Define \(Strings(\genericString)\) as the set of all finite state-action sequences that start with \(\genericString\) such that \(Strings(\genericString) = \lbrace \genericString w' | w' \in (\gameStateSpace \times \gameActionSpace)^{*}\rbrace \).

Let \(W_{loss}\) be the set of finite state-action sequences that lead to a communication loss for the first time. Formally, \[W_{loss} = \left\lbrace \genericString \in (\gameStateSpace \times \gameActionSpace)^{*} |\genericFunction(\genericString) =0 \text{, and } (\forall \genericString' \in Pref(W_{loss}), \genericFunction(\genericString') = 1 ) \right\rbrace.\] Note that there do not exist \(w,w' \in W_{loss}\) and \(w \neq w'\) such that \(w \in Pref(w')\) or \(w' \in Pref(w)\).

Let \(W_{\neg loss}\) be the set of finite shortest state-action sequences that guarantees the agents will not ever experience a communication loss. Formally, \[W_{\neg loss} = \left\lbrace w \in (\gameStateSpace \times \gameActionSpace)^{*} |(\forall w' \in Strings(w), \genericFunction(w') =1 )\text{, and } (\forall w' \in Pref(\genericString), \exists \bar{w} \in Strings(w'), \genericFunction(\bar{w}) =0)  \right\rbrace.\] Note that there do not exist \(w,w' \in W_{\neg loss}\) and \(w \neq w'\) such that \(w \in Pref(w')\) or \(w' \in Pref(w)\).

Note that \(W_{loss} \cap W_{\neg loss} = \emptyset\). Also, note that \[\bigcup_{\genericString \in W_{loss} \cup W_{\neg loss}} \left\lbrace \genericString \genericString' | \genericString' \in (\gameStateSpace \times \gameActionSpace)^{\omega} \right\rbrace  = (\gameStateSpace \times \gameActionSpace)^{\omega},\] i.e., every path starts with a finite state-action sequence from \(W_{loss}\) or \(W_{\neg loss}\). Let \(\tau\) denote the random hitting time to set \((W_{loss} \cup W_{\neg loss}\), i.e., \(\gameState_{0}\gameAction_{0}\gameState_{1}\gameAction_{1}\ldots \gameState_{\tau}\gameAction_{\tau} \in (W_{loss} \cup W_{\neg loss})\).

We have
\begin{subequations}
\begin{align}
\kl(\gamePathDist^{\fullcommunication} || \gamePathDist^{\imaginary}_{\genericFunction}) &= \sum_{\gameState_{0}\gameAction_{0}\gameState_{1}\gameAction_{1}\ldots \in (\gameStateSpace \times \gameActionSpace)^{\omega}} \probabilityMeasure^{\fullcommunication}(\gameState_{0}\gameAction_{0}\gameState_{1}\gameAction_{1}\ldots) \log\left(\frac{\probabilityMeasure^{\fullcommunication}(\gameState_{0}\gameAction_{0}\gameState_{1}\gameAction_{1}\ldots)}{\probabilityMeasure^{\imaginary}_{\genericFunction}(\gameState_{0}\gameAction_{0}\gameState_{1}\gameAction_{1}\ldots)}\right)
\\
&= \sum_{\genericString = \gameState_{0}\gameAction_{0}\gameState_{1}\gameAction_{1}\ldots \gameState_{\tau}\gameAction_{\tau} \in (W_{loss} \cup W_{\neg loss})} \probabilityMeasure^{\fullcommunication}(\genericString) \log\left(\frac{\probabilityMeasure^{\fullcommunication}(\genericString)}{\probabilityMeasure^{\imaginary}_{\genericFunction}(\genericString)}\right) \nonumber 
\\
&\quad +\sum_{\genericString = \gameState_{0}\gameAction_{0}\gameState_{1}\gameAction_{1}\ldots \gameState_{\tau}\gameAction_{\tau} \in (W_{loss} \cup W_{\neg loss})} \probabilityMeasure^{\fullcommunication}(\genericString) \sum_{\genericString'= \gameState_{\tau+1}\gameAction_{\tau+1}\gameState_{\tau+2}\gameAction_{\tau+2}\ldots \in (\gameStateSpace \times \gameActionSpace)^{\omega}} \probabilityMeasure^{\fullcommunication}(\genericString' | \genericString) \log\left(\frac{\probabilityMeasure^{\fullcommunication}(\genericString'|\genericString)}{\probabilityMeasure^{\imaginary}_{\genericFunction}(\genericString'|\genericString)}\right) 
\\
&= \sum_{\genericString = \gameState_{0}\gameAction_{0}\gameState_{1}\gameAction_{1}\ldots \gameState_{\tau}\gameAction_{\tau} \in (W_{loss} \cup W_{\neg loss})} \probabilityMeasure^{\fullcommunication}(\genericString) \log\left(\frac{\probabilityMeasure^{\fullcommunication}(\genericString)}{\probabilityMeasure^{\fullcommunication}(\genericString)}\right) \nonumber 
\\
&\quad +\sum_{\genericString = \gameState_{0}\gameAction_{0}\gameState_{1}\gameAction_{1}\ldots \gameState_{\tau}\gameAction_{\tau} \in (W_{loss} \cup W_{\neg loss})} \probabilityMeasure^{\fullcommunication}(\genericString) \sum_{\genericString'= \gameState_{\tau+1}\gameAction_{\tau+1}\gameState_{\tau+2}\gameAction_{\tau+2}\ldots \in (\gameStateSpace \times \gameActionSpace)^{\omega}} \probabilityMeasure^{\fullcommunication}(\genericString' | \genericString) \log\left(\frac{\probabilityMeasure^{\fullcommunication}(\genericString'|\genericString)}{\probabilityMeasure^{\imaginary}_{\genericFunction}(\genericString'|\genericString)}\right)  \label{eq:jointimgeqaulbeforehitting}
\\
&=\sum_{\genericString = \gameState_{0}\gameAction_{0}\gameState_{1}\gameAction_{1}\ldots \gameState_{\tau}\gameAction_{\tau} \in (W_{loss} \cup W_{\neg loss})} \probabilityMeasure^{\fullcommunication}(\genericString) \sum_{\genericString'= \gameState_{\tau+1}\gameAction_{\tau+1}\gameState_{\tau+2}\gameAction_{\tau+2}\ldots \in (\gameStateSpace \times \gameActionSpace)^{\omega}} \probabilityMeasure^{\fullcommunication}(\genericString' | \genericString) \log\left(\frac{\probabilityMeasure^{\fullcommunication}(\genericString'|\genericString)}{\probabilityMeasure^{\imaginary}_{\genericFunction}(\genericString'|\genericString)}\right) 
\end{align}
\end{subequations}
where \eqref{eq:jointimgeqaulbeforehitting} is because the imaginary play is the same with the joint policy for \(t = 0, \ldots, \tau\).

We have 
\begin{subequations}
\begin{align}
\kl(\gamePathDist^{\fullcommunication} || \gamePathDist^{\imaginary}_{\genericFunction}) 
&= \sum_{\genericString = \gameState_{0}\gameAction_{0}\gameState_{1}\gameAction_{1}\ldots \gameState_{\tau}\gameAction_{\tau} \in W_{loss}} \probabilityMeasure^{\fullcommunication}(\genericString) \sum_{\genericString'= \gameState_{\tau+1}\gameAction_{\tau+1}\gameState_{\tau+2}\gameAction_{\tau+2}\ldots \in (\gameStateSpace \times \gameActionSpace)^{\omega}} \probabilityMeasure^{\fullcommunication}(\genericString' | \genericString) \log\left(\frac{\probabilityMeasure^{\fullcommunication}(\genericString'|\genericString)}{\prod_{i=1}^{\numAgents}\probabilityMeasure^{\fullcommunication}(\mdpState^{i}_{\tau+1}\mdpAction^{i}_{\tau+1}\mdpState^{i}_{\tau+2}\mdpAction^{i}_{\tau+2}\ldots|\genericString)}\right)
\\
&\leq \sum_{\genericString = \gameState_{0}\gameAction_{0}\gameState_{1}\gameAction_{1}\ldots \gameState_{\tau}\gameAction_{\tau} \in W_{loss}} \probabilityMeasure^{\fullcommunication}(\genericString) \sum_{\genericString'= \gameState_{\tau+1}\gameAction_{\tau+1}\gameState_{\tau+2}\gameAction_{\tau+2}\ldots \in (\gameStateSpace \times \gameActionSpace)^{\omega}} \probabilityMeasure^{\fullcommunication}(\genericString' | \genericString) \log\left(\frac{\probabilityMeasure^{\fullcommunication}(\genericString'|\genericString)}{\prod_{i=1}^{\numAgents}\probabilityMeasure^{\fullcommunication}(\mdpState^{i}_{\tau+1}\mdpAction^{i}_{\tau+1}\mdpState^{i}_{\tau+2}\mdpAction^{i}_{\tau+2}\ldots|\genericString)}\right) \nonumber
\\
&\quad  + \sum_{\genericString = \gameState_{0}\gameAction_{0}\gameState_{1}\gameAction_{1}\ldots \gameState_{\tau}\gameAction_{\tau} \in W_{\neg loss}} \probabilityMeasure^{\fullcommunication}(\genericString) \sum_{\genericString'= \gameState_{\tau+1}\gameAction_{\tau+1}\gameState_{\tau+2}\gameAction_{\tau+2}\ldots \in (\gameStateSpace \times \gameActionSpace)^{\omega}} \probabilityMeasure^{\fullcommunication}(\genericString' | \genericString) \log\left(\frac{\probabilityMeasure^{\fullcommunication}(\genericString'|\genericString)}{\prod_{i=1}^{\numAgents}\probabilityMeasure^{\fullcommunication}(\mdpState^{i}_{\tau+1}\mdpAction^{i}_{\tau+1}\mdpState^{i}_{\tau+2}\mdpAction^{i}_{\tau+2}\ldots|\genericString)}\right) \nonumber
\\
&\quad  + \sum_{\genericString = \gameState_{0}\gameAction_{0}\gameState_{1}\gameAction_{1}\ldots \gameState_{\tau}\gameAction_{\tau} \in (W_{ loss} \cup W_{\neg loss})} \probabilityMeasure^{\fullcommunication}(\genericString)  \log\left(\frac{\probabilityMeasure^{\fullcommunication}(\genericString)}{\prod_{i=1}^{\numAgents}\probabilityMeasure^{\fullcommunication}(\mdpState^{i}_{0}\mdpAction^{i}_{0}\mdpState^{i}_{1}\mdpAction^{i}_{1}\ldots \mdpState^{i}_{\tau}\mdpAction^{i}_{\tau} )}\right) \label{ineq:klincreased}
\end{align}
\end{subequations}
since the additional terms in \eqref{ineq:klincreased} are KL divergences, which are always nonnegative.

By the definition of conditional entropy,
\begin{subequations}
\begin{align}
    \kl(\gamePathDist^{\fullcommunication} || \gamePathDist^{\imaginary}_{\genericFunction}) &\leq \left[ \sum_{i=1}^{\numAgents} \entropy(\mdpStateRandomVar^{i}_{\tau+1}\mdpActionRandomVar^{i}_{\tau+1 }\mdpStateRandomVar^{i}_{\tau+2}\mdpActionRandomVar^{i}_{\tau+2} \ldots| \gameStateRandomVar_{0}\gameActionRandomVar_{0}\ldots\gameStateRandomVar_{\tau}\gameActionRandomVar_{\tau}) +\entropy(\mdpStateRandomVar^{i}_{0}\mdpActionRandomVar^{i}_{0 } \ldots \mdpActionRandomVar^{i}_{\tau} \mdpActionRandomVar^{i}_{\tau}|\gameStateRandomVar_{0}) \right] \nonumber
    \\
    &\quad  -  \entropy(\gameStateRandomVar_{\tau+1}\gameActionRandomVar_{\tau+1 }\gameStateRandomVar_{\tau+2}\gameActionRandomVar_{\tau+2} \ldots| \gameStateRandomVar_{0}\gameActionRandomVar_{0}\ldots\gameStateRandomVar_{\tau}\gameActionRandomVar_{\tau}) - \entropy(\gameStateRandomVar_{0}\gameActionRandomVar_{0}\ldots\gameStateRandomVar_{\tau}\gameActionRandomVar_{\tau}|\gameStateRandomVar_{0})
    \\
    &\leq \left[ \sum_{i=1}^{\numAgents} \entropy(\mdpStateRandomVar^{i}_{\tau+1}\mdpActionRandomVar^{i}_{\tau+1 }\mdpStateRandomVar^{i}_{\tau+2}\mdpActionRandomVar^{i}_{\tau+2} \ldots| \gameStateRandomVar_{0}\mdpActionRandomVar^{i}_{0}\ldots\mdpStateRandomVar^{i}_{\tau}\mdpActionRandomVar^{i}_{\tau}) +\entropy(\mdpStateRandomVar^{i}_{0}\mdpActionRandomVar^{i}_{0 } \ldots \mdpActionRandomVar^{i}_{\tau} \mdpActionRandomVar^{i}_{\tau}|\gameStateRandomVar_{0}) \right] \nonumber \label{conditioningreducesentropy2}
    \\
    &\quad  -  \entropy(\gameStateRandomVar_{\tau+1}\gameActionRandomVar_{\tau+1 }\gameStateRandomVar_{\tau+2}\gameActionRandomVar_{\tau+2} \ldots| \gameStateRandomVar_{0}\gameActionRandomVar_{0}\ldots\gameStateRandomVar_{\tau}\gameActionRandomVar_{\tau}) - \entropy(\gameStateRandomVar_{0}\gameActionRandomVar_{0}\ldots\gameStateRandomVar_{\tau}\gameActionRandomVar_{\tau}|\gameStateRandomVar_{0})
\end{align}
\end{subequations} where \eqref{conditioningreducesentropy2} is because conditioning reduces entropy. Finally, we have
\begin{subequations}
\begin{align}
    \kl(\gamePathDist^{\fullcommunication} || \gamePathDist^{\imaginary}_{\genericFunction}) &\leq   \left[ \sum_{i=1}^{\numAgents} \entropy(\mdpStateRandomVar^{i}_{\tau+1}\mdpActionRandomVar^{i}_{\tau+1 }\mdpStateRandomVar^{i}_{\tau+2}\mdpActionRandomVar^{i}_{\tau+2} \ldots| \gameStateRandomVar_{0}\mdpActionRandomVar^{i}_{0}\ldots\mdpStateRandomVar^{i}_{\tau}\mdpActionRandomVar^{i}_{\tau}) +\entropy(\mdpStateRandomVar^{i}_{0}\mdpActionRandomVar^{i}_{0 } \ldots \mdpActionRandomVar^{i}_{\tau} \mdpActionRandomVar^{i}_{\tau}|\gameStateRandomVar_{0}) \right] \nonumber
    \\
    &\quad  -  \entropy(\gameStateRandomVar_{\tau+1}\gameActionRandomVar_{\tau+1 }\gameStateRandomVar_{\tau+2}\gameActionRandomVar_{\tau+2} \ldots| \gameStateRandomVar_{0}\gameActionRandomVar_{0}\ldots\gameStateRandomVar_{\tau}\gameActionRandomVar_{\tau}) - \entropy(\gameStateRandomVar_{0}\gameActionRandomVar_{0}\ldots\gameStateRandomVar_{\tau}\gameActionRandomVar_{\tau}|\gameStateRandomVar_{0}) 
    \\
    &= \left[ \sum_{i=1}^{N} \entropy(\mdpStateRandomVar^{i}_{0}\mdpActionRandomVar^{i}_{0}|\gameStateRandomVar_{0}) \right]  - \entropy(\gameStateRandomVar_{0}\gameActionRandomVar_{0}|\gameStateRandomVar_{0}) + \sum_{t=0}^{\infty} \left[ \left[ \sum_{i=1}^{N} \entropy(\mdpStateRandomVar^{i}_{t}\mdpActionRandomVar^{i}_{t}|\gameStateRandomVar_{0}\mdpActionRandomVar^{i}_{0}  \ldots \mdpStateRandomVar^{i}_{t-1}\mdpActionRandomVar^{i}_{t-1}) \right]  - \entropy(\gameStateRandomVar_{t}\gameActionRandomVar_{t}|\gameStateRandomVar_{0}\gameActionRandomVar_{0} \ldots \gameStateRandomVar_{t-1}\gameActionRandomVar_{t-1} ) \right] \label{chainruleofentropy2}
    \\
    &=\totalCorrelation_{\jointPolicy}
    \\
    &=\kl(\gamePathDist^{\fullcommunication} || \gamePathDist^{\imaginary}_{0})
\end{align}
\end{subequations}
where  \eqref{chainruleofentropy2} is due to the chain rule of entropy.

\end{proof}

\newpage
\begin{proof}[Proof of Theorem \ref{theorem:imaginarystonger}]
 Let \(R\) be the set of paths that reach \(\targetSet\). A path \(\gameState_{0}\gameAction_{0}\gameState_{1}\gameAction_{1}\ldots \in R\) if and only if there exists \(t \geq 0\) such that  \(\gameState_{t} \in R\). Also let \(R'\) be an arbitrary set of paths.
\begin{subequations}
\begin{align}
    \gameValue^{\fullcommunication} - \gameValue^{\imaginary}&=  \sum_{\gameState_{0}\gameAction_{0}\gameState_{1}\gameAction_{1}\ldots \in R} \probabilityMeasure^{\fullcommunication}(\gameState_{0}\gameAction_{0}\gameState_{1}\gameAction_{1}\ldots ) - \probabilityMeasure^{\imaginary}_{\genericFunction}(\gameState_{0}\gameAction_{0}\gameState_{1}\gameAction_{1}\ldots )
    \\
    &\leq \left | \sum_{\gameState_{0}\gameAction_{0}\gameState_{1}\gameAction_{1}\ldots \in R} \probabilityMeasure^{\fullcommunication}(\gameState_{0}\gameAction_{0}\gameState_{1}\gameAction_{1}\ldots ) - \probabilityMeasure^{\imaginary}_{\genericFunction}(\gameState_{0}\gameAction_{0}\gameState_{1}\gameAction_{1}\ldots )\right |
    \\
    &\leq \sup_{R'} \left |\sum_{\gameState_{0}\gameAction_{0}\gameState_{1}\gameAction_{1}\ldots \in R'} \probabilityMeasure^{\fullcommunication}(\gameState_{0}\gameAction_{0}\gameState_{1}\gameAction_{1}\ldots ) - \probabilityMeasure^{\imaginary}_{\genericFunction}(\gameState_{0}\gameAction_{0}\gameState_{1}\gameAction_{1}\ldots )\right |
    \\
    &\leq \sqrt{1-\exp(-\kl(\gamePathDist^{\fullcommunication} || \gamePathDist^{\imaginary}_{\genericFunction}))} \label{bretagnollehuber}
        \\
    &\leq \sqrt{1-\exp(-\totalCorrelation_{\jointPolicy})} \label{klislowerthanC}
\end{align}
\end{subequations}
where \eqref{bretagnollehuber} is due to Bretagnolle-Huber inequality~\cite{bretagnolle1979estimation} and \eqref{klislowerthanC} is due to Lemma \ref{lemma:imaginarystronger}. Rearranging the terms of \eqref{klislowerthanC} yields to the desired result.
\end{proof}

\newpage
\begin{proof}[Proof of Theorem \ref{theorem:imaginary}]
We first show that \[\gameValue^{\imaginary} \geq \gameValue^{\fullcommunication}  (1-\probabilityFailureForever)^{\frac{\expectedLength^{\fullcommunication}}{\gameValue^{\fullcommunication}}}.\] 

Remember that  \(\len(\gamePath = \gameState_0 \gameAction_0\ldots) = \min\lbrace t+1 | \gameState_{t} \in \targetSet \cup \deadSetPrime \rbrace \) and \(\expectedLength^{\fullcommunication} = \expectation[\len(\gamePath) | \gamePath \sim \gamePathDist^{\fullcommunication}]\). Let \(Success\) be an event that the path satisfies the reach-avoid specification. Define \(\expectedLength^{\fullcommunication}_{+} = \expectation[\len(\gamePath) | \gamePath \sim \gamePathDist^{\fullcommunication}, Success]\) and \(\expectedLength^{\fullcommunication}_{-} = \expectation[\len(\gamePath) | \gamePath \sim \gamePathDist^{\fullcommunication}, \neg Success]\). Note that \[\expectedLength^{\fullcommunication} = \expectedLength^{\fullcommunication}_{+} \gameValue^{\fullcommunication} + \expectedLength^{\fullcommunication}_{-} (1 - \gameValue^{full}) \geq \expectedLength^{\fullcommunication}_{+} \gameValue^{\fullcommunication}.  \] Also note that \[\gameValue^{\fullcommunication} = \sum_{t=0}^{\infty} \quad  \sum_{\substack{\gameState_{0}\gameAction_{0}\gameState_{1}\gameAction_{1}\ldots \gameState_{t} \in (\gameStateSpace \times \gameActionSpace)^{t} \times \gameStateSpace\\\gameState_{0},\ldots, \gameState_{t-1} \not \in \targetSet \cup \deadSetPrime\\ \gameState_{t} \in \targetSet }} \probabilityMeasure^{\fullcommunication}(\gameState_{0}\gameAction_{0}\gameState_{1}\gameAction_{1}\ldots\gameState_{t}), \] and \[\expectedLength^{\fullcommunication}_{+} = \frac{1}{{\gameValue^{\fullcommunication}}}\left(\sum_{t=0}^{\infty} (t+1) \sum_{\substack{\gameState_{0}\gameAction_{0}\gameState_{1}\gameAction_{1}\ldots \gameState_{t} \in (\gameStateSpace \times \gameActionSpace)^{t} \times \gameStateSpace\\\gameState_{0},\ldots, \gameState_{t-1} \not \in \targetSet \cup \deadSetPrime\\ \gameState_{t} \in \targetSet }} \probabilityMeasure^{\fullcommunication}(\gameState_{0}\gameAction_{0}\gameState_{1}\gameAction_{1}\ldots\gameState_{t})  \right).\]

Let \(L\) be the event that the agents experience a communication loss before they reach a state in \(\targetSet \cup \deadSetPrime\). Also let \(\probabilityMeasure^{\imaginary}\) denote the probability measure over the (finite or infinite) state-action process under the imaginary play (under Algorithm \ref{alg:imaginaryplay}) where \(\Pr(t_{loss} = t) = (1-\probabilityFailureForever)^{t}\probabilityFailureForever\). Since  \(L\) and  \(\neg L\) are disjoint events, \[
    \gameValue^{\imaginary} = {\Pr}^{\probabilityMeasure^{\imaginary}}(Success \; \& \; L) + {\Pr}^{\probabilityMeasure^{\imaginary}}(Success \; \& \; \neg L) \geq {\Pr}^{\probabilityMeasure^{\imaginary}}(Success \; \& \; \neg L).
\]

We have 
\begin{subequations}
\begin{align}
{\Pr}^{ \probabilityMeasure^{\imaginary}}(Success \; \& \; \neg L) &= \sum_{t=0}^{\infty} \quad  \sum_{\substack{\gameState_{0}\gameAction_{0}\gameState_{1}\gameAction_{1}\ldots \gameState_{t} \in (\gameStateSpace \times \gameActionSpace)^{t} \times \gameStateSpace \\\gameState_{0},\ldots, \gameState_{t-1} \not \in \targetSet \cup \deadSetPrime\\ \gameState_{t} \in \targetSet }} \probabilityMeasure^{\imaginary}(\gameState_{0}\gameAction_{0}\gameState_{1}\gameAction_{1}\ldots\gameState_{t} | \oneStepCommAvailibility_{0}=1, \ldots, \oneStepCommAvailibility_{t} = 1) \Pr(\oneStepCommAvailibility_{0}=1, \ldots, \oneStepCommAvailibility_{t} = 1)
\\
&= \sum_{t=0}^{\infty} \quad \sum_{\substack{\gameState_{0}\gameAction_{0}\gameState_{1}\gameAction_{1}\ldots \gameState_{t} \in (\gameStateSpace \times \gameActionSpace)^{t} \times \gameStateSpace \\\gameState_{0},\ldots, \gameState_{t-1} \not \in \targetSet \cup \deadSetPrime\\ \gameState_{t} \in \targetSet }} \probabilityMeasure^{\fullcommunication}(\gameState_{0}\gameAction_{0}\gameState_{1}\gameAction_{1}\ldots\gameState_{t}) (1-\probabilityFailureForever)^{t+1}
\end{align}
\end{subequations} since \(\probabilityMeasure^{\imaginary} = \probabilityMeasure^{\fullcommunication}\) if there is not a communication loss.

Let \(\genericFunctionAlt(t) = (1-\probabilityFailureForever)^{t+1}\) for \(t\geq 0\). We note that \(\genericFunctionAlt(t)\) is a convex function of \(t\). Also, let \(\genericDistribution\) be a probability distribution over \(0, 1\ldots\) such that \[\genericDistribution(t) = \frac{1}{{\gameValue^{\fullcommunication}}}\left( \sum_{\substack{\gameState_{0}\gameAction_{0}\gameState_{1}\gameAction_{1}\ldots \gameState_{t} \in (\gameStateSpace \times \gameActionSpace)^{t} \times \gameStateSpace\\\gameState_{0},\ldots, \gameState_{t-1} \not \in \targetSet \cup \deadSetPrime\\ \gameState_{t} \in \targetSet }} \probabilityMeasure^{\fullcommunication}(\gameState_{0}\gameAction_{0}\gameState_{1}\gameAction_{1}\ldots\gameState_{t})\right).\] Note that \(\expectation_{t \sim \genericDistribution}[t] = \expectedLength^{\fullcommunication}_{+} -1\)

We have \[{\Pr}^{ \probabilityMeasure^{\imaginary}}(Success \; \& \; \neg L) = \sum_{t=0}^{\infty} \quad  \sum_{\substack{\gameState_{0}\gameAction_{0}\gameState_{1}\gameAction_{1}\ldots \gameState_{t} \in (\gameStateSpace \times \gameActionSpace)^{t} \times \gameStateSpace \\\gameState_{0},\ldots, \gameState_{t-1} \not \in \targetSet \cup \deadSetPrime\\ \gameState_{t} \in \targetSet }} \probabilityMeasure^{\fullcommunication}(\gameState_{0}\gameAction_{0}\gameState_{1}\gameAction_{1}\ldots\gameState_{t}) (1-\probabilityFailureForever)^{t+1} =  \gameValue^{\fullcommunication} \expectation_{t \sim Q} \left[ \genericFunctionAlt(t) \right].\] Since \(\genericFunctionAlt(t)\) is a convex function of \(t\), we get  \[{\Pr}^{ \probabilityMeasure^{\imaginary}}(Success \; \& \; \neg L)  =  \gameValue^{\fullcommunication} \expectation_{t \sim Q} \left[ \genericFunctionAlt(t) \right] \geq \gameValue^{\fullcommunication}  \genericFunctionAlt(\expectation_{t \sim Q} \left[ t \right]) = \gameValue^{\fullcommunication} (1-\probabilityFailureForever)^{\expectation_{t \sim Q} \left[ t \right] + 1}  =  \gameValue^{\fullcommunication}(1-p)^{\expectedLength^{\fullcommunication}_{+}} .\]  by Jensen's inequality~\cite{boyd2004convex}.

Finally, using \(\gameValue^{\imaginary} \geq {\Pr}^{ \probabilityMeasure^{\imaginary}}(Success \; \& \; \neg L)\) and \(\expectedLength^{\fullcommunication} \geq \expectedLength^{\fullcommunication}_{+} \gameValue^{\fullcommunication}\), we get \[\gameValue^{\imaginary}  \geq \gameValue^{\fullcommunication}(1-\probabilityFailureForever)^{\frac{\expectedLength^{\fullcommunication}}{\gameValue^{\fullcommunication}}}.\]

The proof for \(\gameValue^{\imaginary} \geq \gameValue^{\fullcommunication} - \sqrt{1-\exp(-\totalCorrelation_{\jointPolicy})}\) follows the same structure with the proof of Theorem \ref{theorem:imaginarystonger} and have slight differences. We give the full proof for completeness.
 Let \(R\) be the set of paths that reach \(\targetSet\). A path \(\gameState_{0}\gameAction_{0}\gameState_{1}\gameAction_{1}\ldots \in R\) if and only if there exists \(t \geq 0\) such that  \(\gameState_{t} \in R\). Also let \(R'\) be an arbitrary set of paths. Define \(\gamePathDist^{\imaginary} = \expectation_{t_{loss}}[ \gamePathDist^{\imaginary}_{t_{loss}}]\).
\begin{subequations}
\begin{align}
    \gameValue^{\fullcommunication} - \gameValue^{\imaginary}&=  \sum_{\gameState_{0}\gameAction_{0}\gameState_{1}\gameAction_{1}\ldots \in R} \probabilityMeasure^{\fullcommunication}(\gameState_{0}\gameAction_{0}\gameState_{1}\gameAction_{1}\ldots ) - \probabilityMeasure^{\imaginary}(\gameState_{0}\gameAction_{0}\gameState_{1}\gameAction_{1}\ldots )
    \\
    &\leq \left |\sum_{\gameState_{0}\gameAction_{0}\gameState_{1}\gameAction_{1}\ldots \in R} \probabilityMeasure^{\fullcommunication}(\gameState_{0}\gameAction_{0}\gameState_{1}\gameAction_{1}\ldots ) - \probabilityMeasure^{\imaginary}(\gameState_{0}\gameAction_{0}\gameState_{1}\gameAction_{1}\ldots )\right |
    \\
    &\leq \sup_{R'} \left |\sum_{\gameState_{0}\gameAction_{0}\gameState_{1}\gameAction_{1}\ldots \in R'} \probabilityMeasure^{\fullcommunication}(\gameState_{0}\gameAction_{0}\gameState_{1}\gameAction_{1}\ldots ) - \probabilityMeasure^{\imaginary}(\gameState_{0}\gameAction_{0}\gameState_{1}\gameAction_{1}\ldots )\right |
    \\
    &\leq \sqrt{1-\exp(-\kl(\gamePathDist^{\fullcommunication} || \gamePathDist^{\imaginary}))} \label{bretagnollehuber2}
        \\
    &\leq \sqrt{1-\exp(-\expectation_{t_{loss}}[\kl(\gamePathDist^{\fullcommunication} || \gamePathDist^{\imaginary}_{t_{loss}}))])} \label{convexityofkl2}
        \\
    &\leq \sqrt{1-\exp(-\expectation_{t_{loss}}[\totalCorrelation_{\jointPolicy}])} \label{klislowerthanC2}
            \\
    &= \sqrt{1-\exp(-\totalCorrelation_{\jointPolicy})}\label{valuebound2}
\end{align}
\end{subequations}
where \eqref{bretagnollehuber2} is due to Bretagnolle-Huber inequality~\cite{bretagnolle1979estimation}, \eqref{convexityofkl2} is due to the convexity of the KL divergence, and \eqref{klislowerthanC2} is due to Lemma \ref{lemma:imaginary}. Rearranging the terms of \eqref{valuebound2} yields to the desired result.
\end{proof}

\newpage
\begin{proof}[Proof of Theorem \ref{theorem:intermittentstructured}]
The proof of Theorem \ref{theorem:intermittentstructured} follows the same structure with the proof of Theorem \ref{theorem:imaginary} and have slight differences. We give the full proof for completeness.

We first show that \[\gameValue^{\intermittent} \geq \gameValue^{\fullcommunication}  (1-\probabilityFailureOneStep)^{\frac{\expectedLength^{\fullcommunication}}{\gameValue^{\fullcommunication}}}.\] 

Remember that  \(\len(\gamePath = \gameState_0 \gameAction_0\ldots) = \min\lbrace t+1 | \gameState_{t} \in \targetSet \cup \deadSetPrime \rbrace \) and \(\expectedLength^{\fullcommunication} = \expectation[\len(\gamePath) | \gamePath \sim \gamePathDist^{\fullcommunication}]\). Let \(Success\) be an event that the path satisfies the reach-avoid specification. Define \(\expectedLength^{\fullcommunication}_{+} = \expectation[\len(\gamePath) | \gamePath \sim \gamePathDist^{\fullcommunication}, Success]\) and \(\expectedLength^{\fullcommunication}_{-} = \expectation[\len(\gamePath) | \gamePath \sim \gamePathDist^{\fullcommunication}, \neg Success]\). Note that \[\expectedLength^{\fullcommunication} = \expectedLength^{\fullcommunication}_{+} \gameValue^{\fullcommunication} + \expectedLength^{\fullcommunication}_{-} (1 - \gameValue^{full}) \geq \expectedLength^{\fullcommunication}_{+} \gameValue^{\fullcommunication}.  \] Also note that \[\gameValue^{\fullcommunication} = \sum_{t=0}^{\infty}  \quad  \sum_{\substack{\gameState_{0}\gameAction_{0}\gameState_{1}\gameAction_{1}\ldots \gameState_{t} \in (\gameStateSpace \times \gameActionSpace)^{t} \times \gameStateSpace\\\gameState_{0},\ldots, \gameState_{t-1} \not \in \targetSet \cup \deadSetPrime\\ \gameState_{t} \in \targetSet }} \probabilityMeasure^{\fullcommunication}(\gameState_{0}\gameAction_{0}\gameState_{1}\gameAction_{1}\ldots\gameState_{t}), \] and \[\expectedLength^{\fullcommunication}_{+} = \frac{1}{{\gameValue^{\fullcommunication}}}\left(\sum_{t=0}^{\infty} (t+1) \sum_{\substack{\gameState_{0}\gameAction_{0}\gameState_{1}\gameAction_{1}\ldots \gameState_{t} \in (\gameStateSpace \times \gameActionSpace)^{t} \times \gameStateSpace\\\gameState_{0},\ldots, \gameState_{t-1} \not \in \targetSet \cup \deadSetPrime\\ \gameState_{t} \in \targetSet }} \probabilityMeasure^{\fullcommunication}(\gameState_{0}\gameAction_{0}\gameState_{1}\gameAction_{1}\ldots\gameState_{t})  \right).\]

Let \(L\) be the event that the agents experience a communication loss before they reach a state in \(\targetSet \cup \deadSet\). Also let \(\probabilityMeasure^{\intermittent}\) denote the probability measure over the (finite or infinite) state-action process under intermittent communication (under Algorithm \ref{alg:intermittent}) where \(\sequenceCommAvailibility = \lambda_0, \lambda_1, \ldots\) is a random sequence of binary values such that every \(\oneStepCommAvailibility_{t}\) is independently sampled from a Bernoulli random variable with parameter \(1-\probabilityFailureOneStep\).  Since  \(L\) and  \(\neg L\) are disjoint events, \[
    \gameValue^{\intermittent} = {\Pr}^{\probabilityMeasure^{\intermittent}}(Success \; \& \; L) + {\Pr}^{\probabilityMeasure^{\intermittent}}(Success \; \& \; \neg L) \geq {\Pr}^{\probabilityMeasure^{\intermittent}}(Success \; \& \; \neg L).
\]

We have 
\begin{subequations}
\begin{align}
{\Pr}^{\probabilityMeasure^{\intermittent}}(Success \; \& \; \neg L) &= \sum_{t=0}^{\infty} \quad \sum_{\substack{\gameState_{0}\gameAction_{0}\gameState_{1}\gameAction_{1}\ldots \gameState_{t} \in (\gameStateSpace \times \gameActionSpace)^{t} \times \gameStateSpace \\\gameState_{0},\ldots, \gameState_{t-1} \not \in \targetSet \cup \deadSetPrime\\ \gameState_{t} \in \targetSet }} \probabilityMeasure^{\intermittent}(\gameState_{0}\gameAction_{0}\gameState_{1}\gameAction_{1}\ldots\gameState_{t} | \oneStepCommAvailibility_{0}=1, \ldots, \oneStepCommAvailibility_{t} = 1) \Pr(\oneStepCommAvailibility_{0}=1, \ldots, \oneStepCommAvailibility_{t} = 1)
\\
&= \sum_{t=0}^{\infty} \quad \sum_{\substack{\gameState_{0}\gameAction_{0}\gameState_{1}\gameAction_{1}\ldots \gameState_{t} \in (\gameStateSpace \times \gameActionSpace)^{t} \times \gameStateSpace \\\gameState_{0},\ldots, \gameState_{t-1} \not \in \targetSet \cup \deadSetPrime\\ \gameState_{t} \in \targetSet }} \probabilityMeasure^{\fullcommunication}(\gameState_{0}\gameAction_{0}\gameState_{1}\gameAction_{1}\ldots\gameState_{t}) (1-\probabilityFailureOneStep)^{t+1}
\end{align}
\end{subequations} since \(\probabilityMeasure^{\intermittent} = \probabilityMeasure^{\fullcommunication}\) if there is not a communication loss.

Let \(\genericFunctionAlt(t) = (1-\probabilityFailureOneStep)^{t+1}\) for \(t\geq 0\). We note that \(\genericFunctionAlt(t)\) is a convex function of \(t\). Also, let \(\genericDistribution\) be a probability distribution over \(0, 1\ldots\) such that \[\genericDistribution(t) = \frac{1}{{\gameValue^{\fullcommunication}}}\left( \sum_{\substack{\gameState_{0}\gameAction_{0}\gameState_{1}\gameAction_{1}\ldots \gameState_{t} \in (\gameStateSpace \times \gameActionSpace)^{t} \times \gameStateSpace\\\gameState_{0},\ldots, \gameState_{t-1} \not \in \targetSet \cup \deadSetPrime\\ \gameState_{t} \in \targetSet }} \probabilityMeasure^{\fullcommunication}(\gameState_{0}\gameAction_{0}\gameState_{1}\gameAction_{1}\ldots\gameState_{t})\right).\] Note that \(\expectation_{t \sim \genericDistribution}[t] = \expectedLength^{\fullcommunication}_{+} -1\)

We have \[{\Pr}^{\probabilityMeasure^{\intermittent}}(Success \; \& \; \neg L) = \sum_{t=0}^{\infty} \quad \sum_{\substack{\gameState_{0}\gameAction_{0}\gameState_{1}\gameAction_{1}\ldots \gameState_{t} \in (\gameStateSpace \times \gameActionSpace)^{t} \times \gameStateSpace \\\gameState_{0},\ldots, \gameState_{t-1} \not \in \targetSet \cup \deadSetPrime\\ \gameState_{t} \in \targetSet }} \probabilityMeasure^{\fullcommunication}(\gameState_{0}\gameAction_{0}\gameState_{1}\gameAction_{1}\ldots\gameState_{t}) (1-\probabilityFailureOneStep)^{t+1} =  \gameValue^{\fullcommunication} \expectation_{t \sim Q} \left[ \genericFunctionAlt(t) \right].\] Since \(\genericFunctionAlt(t)\) is a convex function of \(t\), we get  \[{\Pr}^{\probabilityMeasure^{\intermittent}}(Success \; \& \; \neg L)  =  \gameValue^{\fullcommunication} \expectation_{t \sim Q} \left[ \genericFunctionAlt(t) \right] \geq \gameValue^{\fullcommunication}  \genericFunctionAlt(\expectation_{t \sim Q} \left[ t \right]) = \gameValue^{\fullcommunication} (1-\probabilityFailureOneStep)^{\expectation_{t \sim Q} \left[ t \right] + 1}  =  \gameValue^{\fullcommunication}(1-p)^{\expectedLength^{\fullcommunication}_{+}} .\]  by Jensen's inequality~\cite{boyd2004convex}.

Finally, using \(\gameValue^{\intermittent} \geq {\Pr}^{\probabilityMeasure^{\intermittent}}(Success \; \& \; \neg L)\) and \(\expectedLength^{\fullcommunication} \geq \expectedLength^{\fullcommunication}_{+} \gameValue^{\fullcommunication}\), we get \[\gameValue^{\intermittent}  \geq \gameValue^{\fullcommunication}(1-\probabilityFailureOneStep)^{\frac{\expectedLength^{\fullcommunication}}{\gameValue^{\fullcommunication}}}.\]

We now show \(\gameValue^{\intermittent} \geq \gameValue^{\fullcommunication} - \sqrt{1-\exp(-\probabilityFailureOneStep\totalCorrelation_{\jointPolicy})}\). Let \(R\) be the set of paths that reach \(\targetSet\). A path \(\gameState_{0}\gameAction_{0}\gameState_{1}\gameAction_{1}\ldots \in R\) if and only if there exists \(t \geq 0\) such that  \(\gameState_{t} \in R\). Also let \(R'\) be an arbitrary set of paths. 
\begin{subequations}
\begin{align}
    \gameValue^{\fullcommunication} - \gameValue^{\intermittent} &=  \sum_{\gameState_{0}\gameAction_{0}\gameState_{1}\gameAction_{1}\ldots \in R} \probabilityMeasure^{\fullcommunication}(\gameState_{0}\gameAction_{0}\gameState_{1}\gameAction_{1}\ldots ) - \probabilityMeasure^{\intermittent}(\gameState_{0}\gameAction_{0}\gameState_{1}\gameAction_{1}\ldots )
    \\
    &\leq \left| \sum_{\gameState_{0}\gameAction_{0}\gameState_{1}\gameAction_{1}\ldots \in R} \probabilityMeasure^{\fullcommunication}(\gameState_{0}\gameAction_{0}\gameState_{1}\gameAction_{1}\ldots ) - \probabilityMeasure^{\intermittent}(\gameState_{0}\gameAction_{0}\gameState_{1}\gameAction_{1}\ldots )\right |
    \\
    &\leq \sup_{R'} \left |\sum_{\gameState_{0}\gameAction_{0}\gameState_{1}\gameAction_{1}\ldots \in R'} \probabilityMeasure^{\fullcommunication}(\gameState_{0}\gameAction_{0}\gameState_{1}\gameAction_{1}\ldots ) - \probabilityMeasure^{\intermittent}(\gameState_{0}\gameAction_{0}\gameState_{1}\gameAction_{1}\ldots )\right |
        \\
    &\leq \sqrt{1-\exp(-\kl(\gamePathDist^{\fullcommunication} || \gamePathDist^{\intermittent}))} \label{bretagnollehuber3}
        \\
    &\leq  \sqrt{1-\exp(-\probabilityFailureOneStep\totalCorrelation_{\jointPolicy})}\label{klislowerthanqC}
\end{align}
\end{subequations}
where \eqref{bretagnollehuber3} is due to Bretagnolle-Huber inequality~\cite{bretagnolle1979estimation}, and \eqref{klislowerthanqC} is due to Lemma \ref{lemma:intermittent}. Rearranging the terms of \eqref{klislowerthanqC} yields to the desired result.
\end{proof}

\newpage
\begin{proof}[Proof of Proposition \ref{proposition:stationaryissufficient}]
We first show that for every policy \(\jointPolicy'\) there exists a stationary policy \(\jointPolicy^{st}\) such that the value of \eqref{optproblemgeneric} in the main body of the paper for \(\jointPolicy^{st}\) is lower than equal to the value for \(\jointPolicy'\). 

Since every \(\gameState \in \gameStateSpace \) has a finite occupancy measure and \( \gameProcessEndState \) is absorbing, there exists a stationary policy \(\jointPolicy^{st}\) such that the occupancy measures of \(\jointPolicy'\) and \(\jointPolicy^{st}\) are equal for all \(\gameState \in \gameStateSpace, \gameAction \in \gameActionSpace \cup \lbrace \epsilonTransition \rbrace\)~\cite{altman1999constrained}. 

We note that \[\gameValue^{\fullcommunication} = \sum_{\substack{\gameState \in \gameStateSpace \setminus (\targetSet \cup \deadSetPrime) \\ \gameAction \in \gameActionSpace \\ \gameStateAlt \in \targetSet}} \occupancyVar_{\gameState, \gameAction} \gameTransition(\gameState, \gameAction, \gameStateAlt)\] and \[\expectedLength^{\fullcommunication} = \sum_{\substack{\gameState \in \gameStateSpace  \\ \gameAction \in \gameActionSpace}} \occupancyVar_{\gameState, \gameAction}.\] 
Hence the value of \(\gameValue^{\fullcommunication}\) is the same for \(\jointPolicy'\) and \(\jointPolicy^{st}\).
Similarly, the value of \(\expectedLength^{\fullcommunication}\) is the same for \(\jointPolicy'\) and \(\jointPolicy^{st}\).

The entropy \(\entropy(\mdpMixedStateActionProcess^{i})\) of the stationary state-action process\cite{savas2019entropy} \(\mdpMixedStateActionProcess^{i}\) is 
\[
      \sum_{\substack{\mdpState^{i} \in \mdpStateSpace^{i}  \\ \mdpAction^{i} \in \mdpActionSpace^{i} \cup \lbrace \epsilonTransition^{i} \rbrace }}  \occupancyVar_{\mdpState^{i}, \mdpAction^{i}}  \log\left(\frac{\underset{\mdpActionAlt^{i} \in \mdpActionSpace^{i}}{\sum} \occupancyVar_{\mdpState^{i}, \mdpActionAlt^{i}}}{\occupancyVar_{\mdpState^{i}, \mdpAction^{i}}}\right)  + \sum_{\substack{\mdpState^{i} \in \mdpStateSpace^{i} \\ \mdpAction^{i} \in \mdpActionSpace^{i} \cup \lbrace \epsilonTransition^{i} \rbrace }} \occupancyVar_{\mdpState^{i}, \mdpAction^{i}} \sum_{\mdpStateAlt^i \in \mdpStateSpace^{i} \cup \lbrace\mdpProcessEndState^{i}\rbrace} \mdpTransition^{i}(\mdpState^{i}, \mdpAction^{i}, \mdpStateAlt^{i}) \log\left( \frac{1}{\mdpTransition^{i}(\mdpState^{i}, \mdpAction^{i}, \mdpStateAlt^{i})} \right).,
\] which is the same for for both \(\jointPolicy'\) and \(\jointPolicy^{st}\). 

Given a set of policies with the same occupancy measures, the stationary policy achieves the highest entropy~\cite{savas2019entropy}. Consequently, the value of \(\entropy(\gameStateActionProcess) \) for \(\jointPolicy^{st}\) is greater than or equal to the value for\(\jointPolicy'\). 

Since \(\jointPolicy^{st}\) achieves a higher value \(\entropy(\gameStateActionProcess) \) and the other terms have equal values for both \(\jointPolicy'\) and \(\jointPolicy^{st}\), the value of \eqref{optproblemgeneric} in the main body of the paper for \(\jointPolicy^{st}\) is lower than equal to the value for \(\jointPolicy'\).

Given that the stationary policies suffice, \eqref{optproblemgeneric} in the main body of the paper can be rewritten in terms of the occupancy measures:
\begin{subequations}
\begin{align}
    \max_{\occupancyVar} \quad & \gameValue^{\fullcommunication} - \expectedLengthCoef \expectedLength^{\fullcommunication} - \totalCorrelationCoef \left( \sum_{i=1}^{\numAgents} \entropy(\mdpMixedStateActionProcess^{i}) - \entropy(\gameStateActionProcess) \right) \label{eq:policy_synthesis_objective}  \\
    \textrm{s.t.} \quad & \gameValue^{\fullcommunication} = \sum_{\substack{\gameState \in \gameStateSpace \setminus (\targetSet \cup \deadSetPrime)\\ \gameAction \in \gameActionSpace \\ \gameStateAlt \in \targetSet}} \occupancyVar_{\gameState, \gameAction} \gameTransition(\gameState, \gameAction, \gameStateAlt)
    \\
    &\expectedLength^{\fullcommunication} = \sum_{\substack{\gameState \in \gameStateSpace \\ \gameAction \in \gameActionSpace \cup \lbrace \epsilonTransition \rbrace}} \occupancyVar_{\gameState, \gameAction}
    \\& \entropy(\gameStateActionProcess) =      \sum_{\substack{\gameState \in \gameStateSpace  \\ \gameAction \in \gameActionSpace}} \occupancyVar_{\gameState, \gameAction} \log\left(\frac{\underset{\gameActionAlt \in \gameActionSpace}{\sum} \occupancyVar_{\gameState, \gameActionAlt}}{\occupancyVar_{\gameState, \gameAction}}\right) + \sum_{\substack{\gameState \in \gameStateSpace  \\ \gameAction \in \gameActionSpace}} \occupancyVar_{\gameState, \gameAction}\sum_{\gameStateAlt \in \gameStateSpace}  \gameTransition(\gameState, \gameAction, \gameStateAlt) \log\left(\frac{1}{\gameTransition(\gameState, \gameAction, \gameStateAlt)}  \right) 
    \\
    &\entropy(\mdpMixedStateActionProcess^{i}) = \sum_{\substack{\mdpState^{i} \in \mdpStateSpace^{i}  \\ \mdpAction^{i} \in \mdpActionSpace^{i} \cup \lbrace \epsilonTransition^{i} \rbrace }}  \occupancyVar_{\mdpState^{i}, \mdpAction^{i}}  \log\left(\frac{\underset{\mdpActionAlt^{i} \in \mdpActionSpace^{i}}{\sum} \occupancyVar_{\mdpState^{i}, \mdpActionAlt^{i}}}{\occupancyVar_{\mdpState^{i}, \mdpAction^{i}}}\right)  + \sum_{\substack{\mdpState^{i} \in \mdpStateSpace^{i} \\ \mdpAction^{i} \in \mdpActionSpace^{i} \cup \lbrace \epsilonTransition^{i} \rbrace }} \occupancyVar_{\mdpState^{i}, \mdpAction^{i}} \sum_{\mdpStateAlt^i \in \mdpStateSpace^{i} \cup \lbrace\mdpProcessEndState^{i}\rbrace} \mdpTransition^{i}(\mdpState^{i}, \mdpAction^{i}, \mdpStateAlt^{i}) \log\left( \frac{1}{\mdpTransition^{i}(\mdpState^{i}, \mdpAction^{i}, \mdpStateAlt^{i})} \right).
    \\
    &\sum_{\gameAction \in \gameActionSpace \cup \lbrace \epsilonTransition \rbrace} \occupancyVar_{\gameState, \gameAction} = \sum_{\substack{\gameStateAlt \in \gameStateSpace \\ \gameActionAlt \in \gameActionSpace \cup \lbrace \epsilonTransition \rbrace}} \occupancyVar_{\gameStateAlt, \gameActionAlt}\gameTransition(\gameStateAlt, \gameActionAlt, \gameState) + \mathbbm{1}_{\{\gameInitialState = \gameState\}}, \; \forall \gameState \in \gameStateSpace \\
    & \occupancyVar_{\gameState, \gameAction} \geq 0, \; \forall \gameState \in \gameStateSpace, \gameAction \in \gameActionSpace \cup \lbrace \epsilonTransition \rbrace
    \\
    & \occupancyVar_{\gameProcessEndState, \gameAction} = 0, \; \forall  \gameAction \in \gameActionSpace.
\end{align}
\label{optproblemnumeric}
\end{subequations}

Since the occupancy measures are bounded and closed, i.e.,\(\sum_{\gameAction \in \gameActionSpace} \occupancyVar(\gameState,\gameAction) \leq \constantNumber \) for all \(\gameState \in \gameStateSpace\), the feasible space is compact. Since the feasible space is compact and the objective function is continuous, there exists a solution to \eqref{optproblemnumeric}. Hence there exists a stationary policy that is a solution to \eqref{optproblemgeneric} in the main body of the paper.
\end{proof}
\newpage
\section{Details on the Minimum-Dependency Policy Synthesis Algorithm}

In this section, we give the full form of the optimization problem for policy synthesis and explain the preprocessing of \(\game\) in order to ensure that \(\totalCorrelationUpperBound_{\jointPolicy}\) is well-defined. 

For synthesis purposes, we consider the state-action processes until reaching an state in \(\targetSet \cup \deadSetPrime\) and assume that the states in \(\targetSet \cup \deadSetPrime\) are absorbing. In the actual game, the states in \(\targetSet \cup \deadSetPrime\) are not necessarily absorbing and the state-action processes do not terminate at a finite time. Despite this difference, considering the total correlation of the joint state-action process until reaching an state in \(\targetSet \cup \deadSetPrime\) is sufficient to ensure the bounds given in Section 7 of the paper. 

Let \(\mdpStateActionProcessAlt^{i}\) be the infinite length state-action process of  \(\mdpState^i_0 \mdpAction^i_0 \mdpState^i_1 \mdpAction^i_1\ldots\) of Agent \(i\) and \(\gameStateActionProcessAlt\) be the actual joint state-action process  \(\gameState_0 \gameAction_0 \gameState_1 \gameAction_1 \ldots\) of the team. Let \(\mdpStateActionProcess^{i}\) be the modified state-action process of  \(\mdpState^i_0 \mdpAction^i_0 \mdpState^i_1 \mdpAction^i_1 \ldots \mdpState^i_{\tau^{i}-1} \mdpAction^i_{\tau^{i}-1} \mdpState^{i}_{\tau^{i}} \epsilonTransition^{i} \mdpProcessEndState^{i}\epsilonTransition^{i} \mdpProcessEndState^{i} \ldots \) of Agent \(i\) where \(\tau^{i}\) is the first time-step that Agent \(i\) thinks that the team reached a state in \( \targetSet \cup \deadSetPrime \). Formally, \( (\hat{\mdpState}^{1}_{\tau^{i},i}, \ldots, \hat{\mdpState}^{i-1}_{\tau^{i},i}, \mdpState^{i}_{\tau^{i}}, \hat{\mdpState}^{i+1}_{\tau^{i},i}, \ldots, \hat{\mdpState}^{N}_{\tau^{i},i}) \in \targetSet \cup \deadSetPrime\) and  \( (\hat{\mdpState}^{1}_{t,i}, \ldots, \hat{\mdpState}^{i-1}_{t,i}, \mdpState^{i}_{t}, \hat{\mdpState}^{i+1}_{t,i}, \ldots, \hat{\mdpState}^{N}_{t,i}) \not \in \targetSet \cup \deadSetPrime \) for all \(t < \tau\) in Algorithms \ref{alg:imaginaryplay} and \ref{alg:intermittent}. Let \(\gameStateActionProcess\) be the joint modified state-action process of the team, i.e., \( \mdpStateActionProcess^{1}, \ldots, \mdpStateActionProcess^{\numAgents}\). 

Let \(\gameValue^{loss}\) be the reach-avoid probability under communication loss (the reach-avoid probability under Algorithm \ref{alg:imaginaryplay} or \ref{alg:intermittent}). Define a reward function such that the reward is \(1\) if \(\gameState_{t} \not \in \deadSetPrime\) for all \(t < \tau = \min_{i} \tau^{i}\) and \(\gameState_{\tau} \not \in \deadSetPrime\) for \(\tau = \min_{i} \tau^{i}\), and \(0\) otherwise. Let \(\gameValue^{loss}_{modified}\) denote the expected reward with this reward function under communication loss (the reach-avoid probability under Algorithm \ref{alg:imaginaryplay} or \ref{alg:intermittent}). Also, let \(\gameValue^{full}_{modified}\) denote the expected reward with this reward function under joint policy with full communication. 

We note that \(\gameStateActionProcess\) is equal to \(\gameStateActionProcessAlt\) for \(t = 0, \ldots, \tau\). Consequently, if \(\gameStateActionProcess\) collects a reward of \(1\) in the new reward function, then \(\gameStateActionProcessAlt\) satisfies the reach-avoid specification. Hence, we have \( \gameValue^{loss}_{modified} \leq \gameValue^{loss}\). Also note that under full communication \(\gameStateActionProcess\) is equal to \(\gameStateActionProcessAlt\) for \(t = 0, \ldots, \tau\) where \(\tau = \tau^{1}= \ldots = \tau^{\numAgents}\), i.e., the modified state-action processes of the agents are always aligned with each other under full communication. Consequently, \(\gameStateActionProcess\) collects a reward of \(1\) in the new reward function if and only if \(\gameStateActionProcessAlt\) satisfies the reach-avoid specification. Hence, we have \( \gameValue^{\fullcommunication}_{modified} = \gameValue^{\fullcommunication}\). 

By using the modified state-action processes, the bounds given in Section 7 of the paper can be derived in terms of \(\gameValue^{loss}_{modified}\), \(\gameValue^{full}_{modified}\) and the total correlation of the joint state-action process until reaching state  \(\gameProcessEndState = (\mdpProcessEndState^{1}, \ldots, \mdpProcessEndState^{\numAgents})\). Since \( \gameValue^{loss}_{modified} \leq \gameValue^{loss}\) and \( \gameValue^{\fullcommunication}_{modified} = \gameValue^{\fullcommunication}\), we can recover the bounds derived in in Section 7 using the total correlation of the joint state-action process until reaching state \(\gameProcessEndState = (\mdpProcessEndState^{1}, \ldots, \mdpProcessEndState^{\numAgents})\).

Given that the stationary policies suffice, \eqref{optproblemgeneric} from the main body of the paper can be rewritten in terms of the occupancy measures using the closed-form expression for the entropy of a stationary state-action process given in \cite{savas2019entropy}:
\begin{subequations}
\begin{align}
    \max_{\occupancyVar} \quad & \gameValue^{\fullcommunication} - \expectedLengthCoef \expectedLength^{\fullcommunication} - \totalCorrelationCoef \left( \sum_{i=1}^{\numAgents} \entropy(\mdpMixedStateActionProcess^{i}) - \entropy(\gameStateActionProcess) \right) \label{eq:policy_synthesis_objective}  \\
    \textrm{s.t.} \quad & \gameValue^{\fullcommunication} = \sum_{\substack{\gameState \in \gameStateSpace \setminus (\targetSet \cup \deadSetPrime)\\ \gameAction \in \gameActionSpace \\ \gameStateAlt \in \targetSet}} \occupancyVar_{\gameState, \gameAction} \gameTransition(\gameState, \gameAction, \gameStateAlt)
    \\
    &\expectedLength^{\fullcommunication} = \sum_{\substack{\gameState \in \gameStateSpace \\ \gameAction \in \gameActionSpace \cup \lbrace \epsilonTransition \rbrace}} \occupancyVar_{\gameState, \gameAction}
    \\& \entropy(\gameStateActionProcess) =      \sum_{\substack{\gameState \in \gameStateSpace  \\ \gameAction \in \gameActionSpace}} \occupancyVar_{\gameState, \gameAction} \log\left(\frac{\underset{\gameActionAlt \in \gameActionSpace}{\sum} \occupancyVar_{\gameState, \gameActionAlt}}{\occupancyVar_{\gameState, \gameAction}}\right) + \sum_{\substack{\gameState \in \gameStateSpace  \\ \gameAction \in \gameActionSpace}} \occupancyVar_{\gameState, \gameAction}\sum_{\gameStateAlt \in \gameStateSpace}  \gameTransition(\gameState, \gameAction, \gameStateAlt) \log\left(\frac{1}{\gameTransition(\gameState, \gameAction, \gameStateAlt)}  \right) 
    \\
    &\entropy(\mdpMixedStateActionProcess^{i}) = \sum_{\substack{\mdpState^{i} \in \mdpStateSpace^{i}  \\ \mdpAction^{i} \in \mdpActionSpace^{i} \cup \lbrace \epsilonTransition^{i} \rbrace }}  \occupancyVar_{\mdpState^{i}, \mdpAction^{i}}  \log\left(\frac{\underset{\mdpActionAlt^{i} \in \mdpActionSpace^{i}}{\sum} \occupancyVar_{\mdpState^{i}, \mdpActionAlt^{i}}}{\occupancyVar_{\mdpState^{i}, \mdpAction^{i}}}\right)  + \sum_{\substack{\mdpState^{i} \in \mdpStateSpace^{i} \\ \mdpAction^{i} \in \mdpActionSpace^{i} \cup \lbrace \epsilonTransition^{i} \rbrace }} \occupancyVar_{\mdpState^{i}, \mdpAction^{i}} \sum_{\mdpStateAlt^i \in \mdpStateSpace^{i} \cup \lbrace\mdpProcessEndState^{i}\rbrace} \mdpTransition^{i}(\mdpState^{i}, \mdpAction^{i}, \mdpStateAlt^{i}) \log\left( \frac{1}{\mdpTransition^{i}(\mdpState^{i}, \mdpAction^{i}, \mdpStateAlt^{i})} \right).
    \\
    &\sum_{\gameAction \in \gameActionSpace \cup \lbrace \epsilonTransition \rbrace} \occupancyVar_{\gameState, \gameAction} = \sum_{\substack{\gameStateAlt \in \gameStateSpace \\ \gameActionAlt \in \gameActionSpace \cup \lbrace \epsilonTransition \rbrace}} \occupancyVar_{\gameStateAlt, \gameActionAlt}\gameTransition(\gameStateAlt, \gameActionAlt, \gameState) + \mathbbm{1}_{\{\gameInitialState = \gameState\}}, \; \forall \gameState \in \gameStateSpace \\
    & \occupancyVar_{\gameState, \gameAction} \geq 0, \; \forall \gameState \in \gameStateSpace, \gameAction \in \gameActionSpace \cup \lbrace \epsilonTransition \rbrace
    \\
    & \occupancyVar_{\gameProcessEndState, \gameAction} = 0, \; \forall  \gameAction \in \gameActionSpace.
    \label{eq:policy_synthesis_absorb_state_zero}
\end{align}
\label{optproblemnumeric2}
\end{subequations}

\newpage
\section{Supplementary Experimental Details}
\label{sec:app_experiment_details}

\paragraph{Solving for the baseline joint policy.}
The baseline joint policy maximizes the reach-avoid probability \(\gameValue^{\fullcommunication}\). We solve the optimization problem given in \eqref{optproblembaseline} to synthesize the baseline joint policy. The optimization problem is a linear program with occupancy measures as the variables. 
\begin{subequations}
\begin{align}
    \max_{\occupancyVar} \quad & \gameValue^{\fullcommunication}  \\
    \textrm{s.t.} \quad & \gameValue^{\fullcommunication} = \sum_{\substack{\gameState \in \gameStateSpace \setminus (\targetSet \cup \deadSetPrime)\\ \gameAction \in \gameActionSpace \\ \gameState \in \targetSet}} \occupancyVar_{\gameState, \gameAction} \gameTransition(\gameState, \gameAction, \gameStateAlt)
    \\
    &\sum_{\gameAction \in \gameActionSpace} \occupancyVar_{\gameState, \gameAction} = \sum_{\substack{\gameStateAlt \in \gameStateSpace \\ \gameActionAlt \in \gameActionSpace}} \occupancyVar_{\gameStateAlt, \gameActionAlt}\gameTransition(\gameStateAlt, \gameActionAlt, \gameState) + \mathbbm{1}_{\{\gameInitialState = \gameState\}}, \; \forall \gameState \in \gameStateSpace \setminus (\targetSet \cup \deadSetPrime)\\
    & \occupancyVar_{\gameState, \gameAction} \geq 0, \; \forall \gameState \in \gameStateSpace \setminus (\targetSet \cup \deadSetPrime), \gameAction \in \gameActionSpace
    \\
    & \occupancyVar_{\gameState, \gameAction} = 0, \; \forall \gameState \in (\targetSet \cup \deadSetPrime), \gameAction \in \gameActionSpace.
\end{align}
\label{optproblembaseline}
\end{subequations}

\paragraph{Linearizing the policy synthesis objective function.}

To solve for a local optimum of the policy synthesis optimization problem \eqref{eq:policy_synthesis_objective}-\eqref{eq:policy_synthesis_absorb_state_zero}, we apply the convex-concave procedure \cite{yuille2002concave}. 
This algorithm requires that we linearize the convex term (i.e. \(- \sum_{i=1}^{\numAgents} \entropy(\mdpMixedStateActionProcess^{i})\)) in the objective function and iteratively solve the resulting exponential cone program.
At each iteration of the algorithm, we linearize the convex term about the solution from the exponential cone program of the previous iteration.

\paragraph{Software implementation.}
All project code is implemented in Python and is available at \href{https://github.com/cyrusneary/multi-agent-comms}{github.com/cyrusneary/multi-agent-comms}. We model the optimization problems using CVXPY~\cite{diamond2016cvxpy} and use MOSEK~\cite{aps2020mosek} to solve these optimization problems. 

\paragraph{Discretizing the multiagent navigation task.}
In order to implement the multiagent task environment illustrated in Figure 2 of the main paper, we discretize the space into a \(5 \times 5\) grid of states.
This discretization is visualized in Figure \ref{fig:gridworld}.

\begin{figure}[h!]
    \centering
    \definecolor{grass}{HTML}{A5A871}
\definecolor{water}{HTML}{586494}
\definecolor{mountain}{HTML}{8c8188}
\definecolor{target}{HTML}{4472C4}

\definecolor{r1Path}{HTML}{002060}
\definecolor{r2Path}{HTML}{ff0000}

\newcommand{\pathwidth}{0.8mm}

\begin{tikzpicture}[scale=1.0]

\filldraw[fill=grass, draw=black] (0,0) rectangle (5,5);

\filldraw[fill=blue!20!white, draw=black] (1,0) rectangle (2,1);
\filldraw[fill=blue!20!white, draw=black] (3,0) rectangle (4,1);

\filldraw[fill=mountain, draw=black] (2,0) rectangle (3,1);
\filldraw[fill=mountain, draw=black] (2,2) rectangle (3,3);
\filldraw[fill=mountain, draw=black] (2,4) rectangle (3,5);

\filldraw[fill=water, draw=black] (0,4) rectangle (1,5);
\filldraw[fill=water, draw=black] (1,4) rectangle (2,5);
\filldraw[fill=water, draw=black] (3,4) rectangle (4,5);

\node[anchor=center] at (0.5, 0.5) {\huge \(\robot_2\)};
\node[anchor=center] at (4.5, 0.5) {\huge \(\robot_1\)};
\node[anchor=center] at (1.5, 0.5) {\huge \(T_1\)};
\node[anchor=center] at (3.5, 0.5) {\huge \(T_2\)};

\draw[draw=black] (0,0) rectangle (5, 5);
\draw[draw=black] (0, 0) grid (5, 5); 


\draw[->, color=r1Path, line width=\pathwidth] (0.5, 0.8) -- (0.5, 1.3);
\draw[->, color=r1Path, line width=\pathwidth] (0.5, 1.3) -- (3.5, 1.3);
\draw[->, color=r1Path, line width=\pathwidth] (3.5, 1.3) -- (3.5, 0.8);

\draw[->, color=r2Path, line width=\pathwidth] (4.3, 0.8) -- (4.3, 1.7);
\draw[->, color=r2Path, line width=\pathwidth] (4.3, 1.7) -- (1.7, 1.7);
\draw[->, color=r2Path, line width=\pathwidth] (1.7, 1.7) -- (1.7, 0.8);

\draw[->, dashed, color=r2Path, line width=\pathwidth] (4.7, 0.8) -- (4.7, 3.5);
\draw[->, dashed, color=r2Path, line width=\pathwidth] (4.7, 3.5) -- (1.3, 3.5);
\draw[->, dashed, color=r2Path, line width=\pathwidth] (1.3, 3.5) -- (1.3, 0.8);


\end{tikzpicture}
    \caption{
    The discretized two-agent navigation task environment.
    The robots begin from the states marked \(R_1\) and \(R_2\).
    They must navigate to states \(T_1\) and \(T_2\) respectively.
    The blue and grey states mark the dangerous water and the impassable mountains, as described in Section 4 of the paper.
    }
    \label{fig:gridworld}
\end{figure}

\paragraph{Supplementary visualizations of \(\policy_{MD}\) and \(\policy_{base}\).}
Figure \ref{fig:occupancy_heatmap_total_corr} shows heatmaps of the occupancy measures of the individual agents under the synthesized joint policy \(\policy_{MD}\). Similarly, Figure \ref{fig:occupancy_heatmap_reachability} shows heatmaps of the occupancy measures of the individual agents under the baseline policy \(\policy_{base}\).
Specifically, each heatmap visualizes the values of the variables \(\occupancyVar_{s^i}\) for some agent \(i\) under one of these joint policies.
These occupancy measures for the individual agents are defined as \(\occupancyVar_{\mdpState^i} =\sum_{\mdpAction^i \in \mdpActionSpace}  \occupancyVar_{\mdpState^i, \mdpAction^{i}}= \sum_{\mdpAction^i \in \mdpActionSpace} \sum_{\gameState^{-i} \in \gameStateSpace^{-i}} \sum_{\gameAction^{-i}\in\gameActionSpace^{-i}} \occupancyVar_{\mdpState^i, \gameState^{-i}, \mdpAction^{i}, \gameAction^{-i}}\).

Intuitively, we may think of the value of \(\occupancyVar_{\mdpState^i}\) as being a measure of the frequency at which agent \(i\) visits local state \(\mdpState^i\) if the joint policy is repeatedly followed from the initial state.
We observe from the figures that the minimum-dependency policy \(\policy_{MD}\) results in robot \(R_1\) navigating through the top valley, while robot \(R_2\) navigates through the bottom valley.
Conversely, under the baseline joint policy (which does not take potential losses in communication into account), both robots navigate through the bottom valley to reach their targets.
This provides empirical confirmation of the qualitative behavior that one might expect from a joint policy that is robust to losses in communication, as discussed in Section 4 and in Section 9 of the paper.

\begin{figure}[h!]
    \centering
     \begin{subfigure}[b]{0.49\textwidth}
         \centering
         \includegraphics[width=\textwidth]{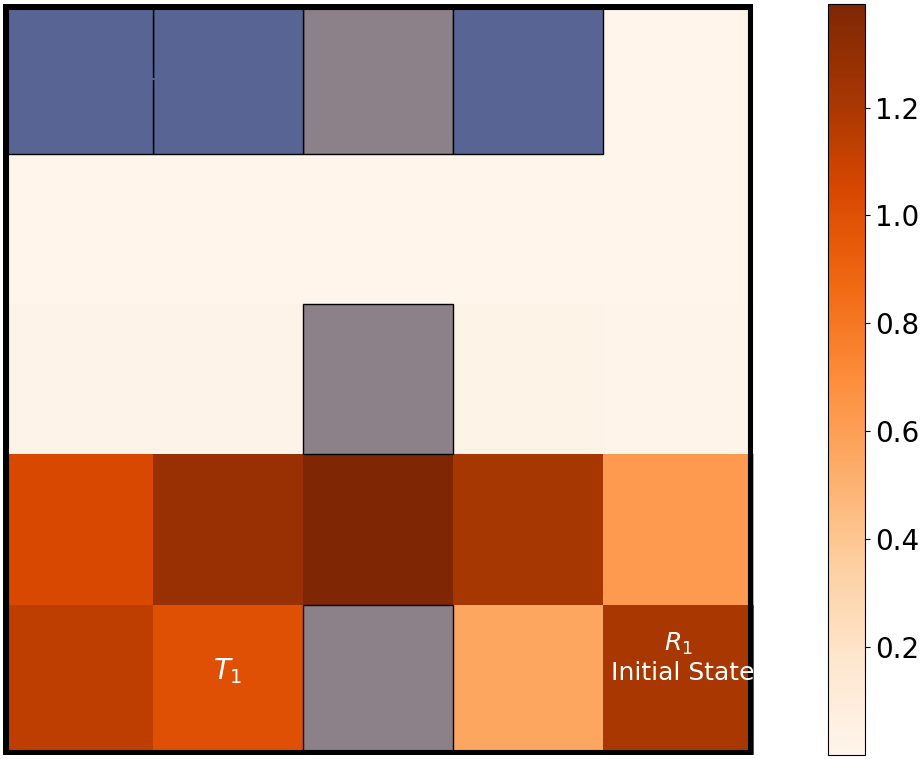}
         \caption{Occupancy measures \(\occupancyVar_{\mdpState^1}\) of Robot \(R_1\)}
         \label{fig:r1_total_corr_heatmap}
     \end{subfigure}
     \hfill
     \begin{subfigure}[b]{0.49\textwidth}
         \centering
         \includegraphics[width=\textwidth]{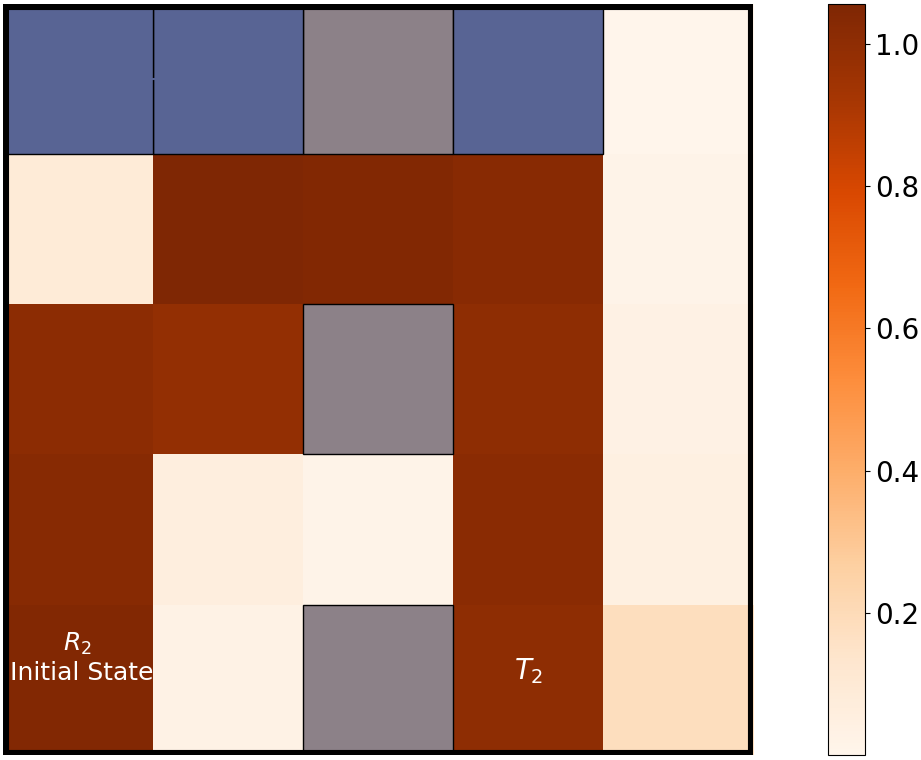}
         \caption{Occupancy measures \(\occupancyVar_{\mdpState^2}\) of Robot \(R_2\)}
         \label{fig:r2_total_corr_heatmap}
     \end{subfigure}
    \caption{Occupancy measures for the individual agents under joint policy \(\policy_{MD}\) (the proposed policy).}
    \label{fig:occupancy_heatmap_total_corr}
\end{figure}

\begin{figure}[h!]
    \centering
     \begin{subfigure}[b]{0.49\textwidth}
         \centering
         \includegraphics[width=\textwidth]{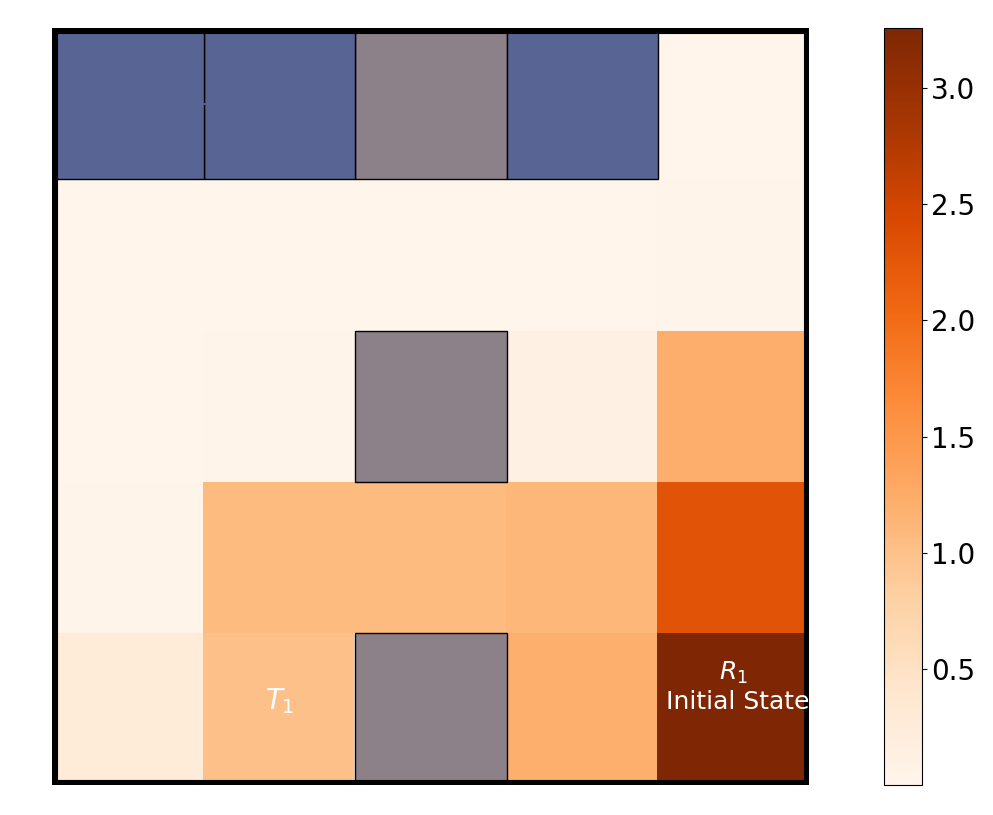}
         \caption{Occupancy measures \(\occupancyVar_{\mdpState^1}\) of Robot \(R_1\)}
         \label{fig:r1_reachability_heatmap}
     \end{subfigure}
     \hfill
     \begin{subfigure}[b]{0.49\textwidth}
         \centering
         \includegraphics[width=\textwidth]{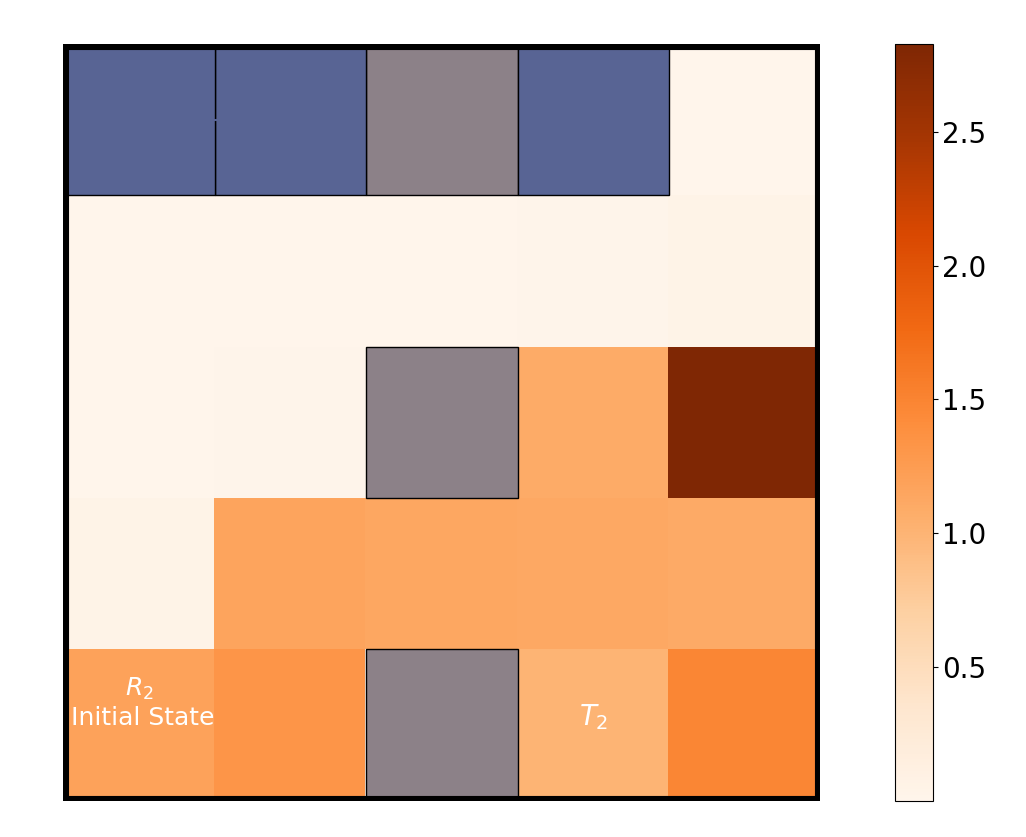}
         \caption{Occupancy measures \(\occupancyVar_{\mdpState^2}\) of Robot \(R_2\)}
         \label{fig:r2_reachability_heatmap}
     \end{subfigure}
    \caption{Occupancy measures for the individual agents under joint policy \(\policy_{base}\).}
    \label{fig:occupancy_heatmap_reachability}
\end{figure}
\clearpage
\section{Additional Three-Agent Experiment}
\label{sec:app_three_agent_example}

In this appendix, we present an three-agent experiment, which demonstrates the ability of the proposed approach to generalize to multiagent systems including more than two agents.
Figure \ref{fig:three_agent_gridworld} illustrates the three-agent task. 
Robots \(\robot_1\), \(\robot_2\), and \(\robot_3\) start in the locations marked in the figure.
Each robot must navigate to its respective target location \(T_1\), \(T_2\), or \(T_3\), which are located in the opposite corner of the environment.
Meanwhile the robots must avoid collisions with each other. The actions of the agents and the slip probabilities associated with these actions are the same as the numerical example given the main body.

\begin{figure}[h!]
    \centering
    \newcommand{\pathwidth}{0.8mm}

\begin{tikzpicture}[scale=1.3]

\filldraw[fill=white, draw=black] (0,0) rectangle (3,3);

\node[anchor=center] at (0.5, 2.5) {\huge \(\robot_2, T_{1}\)};
\node[anchor=center] at (0.5, 0.5) {\huge \(T_{3}\)};
\node[anchor=center] at (2.5, 2.5) {\huge \(\robot_3\)};
\node[anchor=center] at (2.5, 0.5) {\huge \(\robot_{1}, T_{2}\)};

\draw[draw=black] (0, 0) rectangle (3, 3);
\draw[draw=black] (0, 0) grid (3, 3); 

\end{tikzpicture}
    \caption{
    The discretized three-agent navigation task environment.
    The robots begin from the states marked \(R_1\),\(R_2\), and \(R_3\).
    They must navigate to states \(T_1\), \(T_2\), and  \(T_3\), respectively, without crashing into each other.
    }
    \label{fig:three_agent_gridworld}
\end{figure}
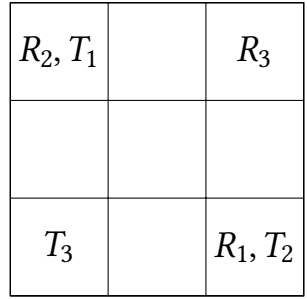

Figure \ref{fig:three_agent_intermittent_play_success_prob} compares the total correlation and the success probabilities of the synthesised minimum-dependency policy \(\policy_{MD}\) and the baseline policy \(\policy_{base}\). The baseline joint policy maximizes the reach-avoid probability \(\gameValue^{\fullcommunication}\). The details on the synthesis of this policy are given in Appendix \ref{sec:app_experiment_details}.
We note that in this example, even when there is no communication available between the agents, \(\policy_{MD}\) has a success probability of \(70\) percent, while the success probability of \(\policy_{base}\) drops to \(17\) percent.

Figure \ref{fig:occupancy_heatmap_total_corr_three_agents} and Figure \ref{fig:occupancy_heatmap_reachability_three_agents} illustrate the occupancy measures of the individual agents under \(\policy_{MD}\) and \(\policy_{base}\), as discussed in Appendix \ref{sec:app_experiment_details}.
We observe that \(\policy_{base}\) induces behavior in which the agents also move around the edges of the environment, but they do so in a way that depends on their teammates.
While \(\policy_{base}\) results in a very high success probability when the communication is available, this type of behavior becomes far too unreliable when communication is lost.
On the other hand, \(\policy_{MD}\) induces behavior in which the agents always move in a counter-clockwise pattern around the edge of the environment.
This pre-agreed behavior is much more likely to succeed when communication is lost.

\begin{figure}
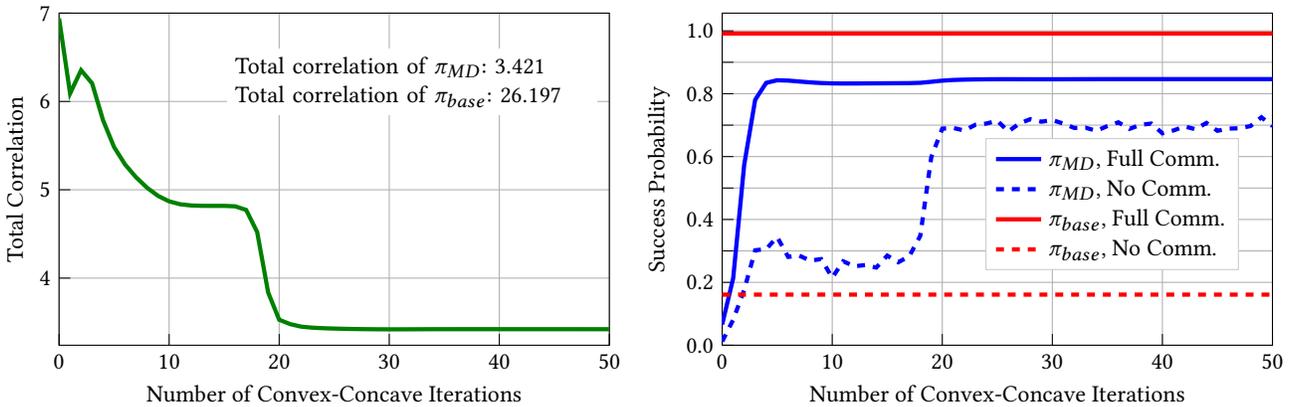

    \centering
\begin{tikzpicture}

\begin{groupplot}[group style={group name = plots, group size=2 by 1, vertical sep=0.0cm, horizontal sep=1.5cm}]

\nextgroupplot[
width=0.5\columnwidth, 
height=6cm,
legend cell align={left},
legend style={
  fill opacity=0.8,
  draw opacity=1,
  text opacity=1,
  at={(4.5cm, 0.4cm)},
  anchor=south west,
  draw=white!80!black
},
tick align=inside,
tick pos=left,
x grid style={white!69.0196078431373!black},
xlabel={Number of Convex-Concave Iterations},
xmajorgrids,
xmin=0.0, xmax=50.0,
xtick style={color=black},
y grid style={white!69.0196078431373!black},
ylabel={Total Correlation},
ymajorgrids,
ymin=3.24, ymax=7.0,
ytick style={color=black},
]

\input{tikz/total_corr_vs_iters_three_agent_aux_action}

\node[fill=white] at (4.7cm, 3.5cm) [text width=4.7cm] {Total correlation of \(\policy_{MD}\): \(3.421\) \\ Total correlation of \(\policy_{base}\): \(26.197\)};

\nextgroupplot[
width=0.5*\columnwidth, 
height=6cm,
legend cell align={left},
legend columns = 1,
legend style={
  fill opacity=1,
  draw opacity=1,
  text opacity=1,
  at={(3.5cm, 1.0cm)},
  anchor=south west,
  draw=white!80!black
},
tick align=inside,
tick pos=left,
x grid style={white!69.0196078431373!black},
xlabel={Number of Convex-Concave Iterations},
xmajorgrids,
xmin=0.0, xmax=50.0,
xtick style={color=black},
y grid style={white!69.0196078431373!black},
ylabel={Success Probability},
ymajorgrids,
ymin=0.0, ymax=1.0562822829432,
ytick style={color=black},
ytick={0.0, 0.2, 0.3, 0.4, 0.5, 0.6, 0.7, 0.8, 0.9, 1.0},
yticklabels={0.0, 0.2, ,0.4, ,0.6, , 0.8, , 1.0},
]

\input{tikz/success_prob_vs_iters_three_agent_aux_action}

\end{groupplot}

\end{tikzpicture}
    \caption{Total correlation and success probability values of the minimum-dependency policy \(\policy_{MD}\) during policy synthesis on the three-agent navigation experiment. (Left) Total correlation value of the policy as a function of the number of elapsed iterations of the convex-concave optimization procedure. (Right) Probability of task success. 
    For comparison, we plot the success probability resulting from both imaginary play execution (no communication) and centralized execution (full communication).
    To estimate the probability of task success, we perform rollouts of the joint policy and compute the empirical rate at which the team accomplishes its objective.}
    \label{fig:three_agent_navigation_results}
\end{figure}

\begin{figure}
    \centering
\begin{tikzpicture}

\begin{axis}[
tick pos=left,
width=0.6\columnwidth, 
height=6cm,
tick align=inside,
tick pos=left,
x grid style={white!69.0196078431373!black},
legend cell align={left},
legend columns = 1,
legend style={
  fill opacity=1,
  draw opacity=1,
  text opacity=1,
  at={(0.2cm, 0.2cm)},
  anchor=south west,
  draw=white!80!black
},
xlabel={q, Probability of Comm. Loss in a Given Timestep},
xmajorgrids,
xmin=-0.02, xmax=1.02,
xtick style={color=black},
y grid style={white!69.0196078431373!black},
ylabel={Success Probability},
ymajorgrids,
ymin=0, ymax=1.0323,
ytick style={color=black}
]
\addplot [ultra thick, blue, mark=*]
table {%
0 0.8461
0.1 0.8423
0.2 0.8332
0.3 0.8319
0.4 0.8172
0.5 0.8085
0.6 0.8072
0.7 0.7895
0.8 0.7703
0.9 0.7422
1 0.6938
};\addlegendentry{\(\policy_{MD}\)}
\addplot [ultra thick, red, mark=*]
table {%
0 0.9911
0.1 0.9792
0.2 0.9518
0.3 0.907
0.4 0.8539
0.5 0.7722
0.6 0.681
0.7 0.5753
0.8 0.4481
0.9 0.3189
1 0.1698
}; \addlegendentry{\(\policy_{base}\)}

\end{axis}

\end{tikzpicture}
    \caption{Success probability of intermittent communication for different values of \(\probabilityFailureOneStep\) on the three-agent navigation experiment. 
    When \(\probabilityFailureOneStep = 0\) communication is available at every timestep, and when \(\probabilityFailureOneStep = 1\) communication is never available.}
    \label{fig:three_agent_intermittent_play_success_prob}
\end{figure}
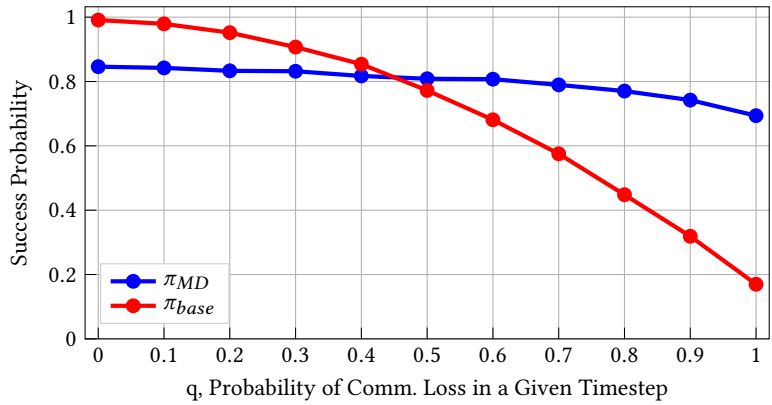

\begin{figure}[h!]
    \centering
     \begin{subfigure}[b]{0.32\textwidth}
         \centering
         \includegraphics[width=\textwidth]{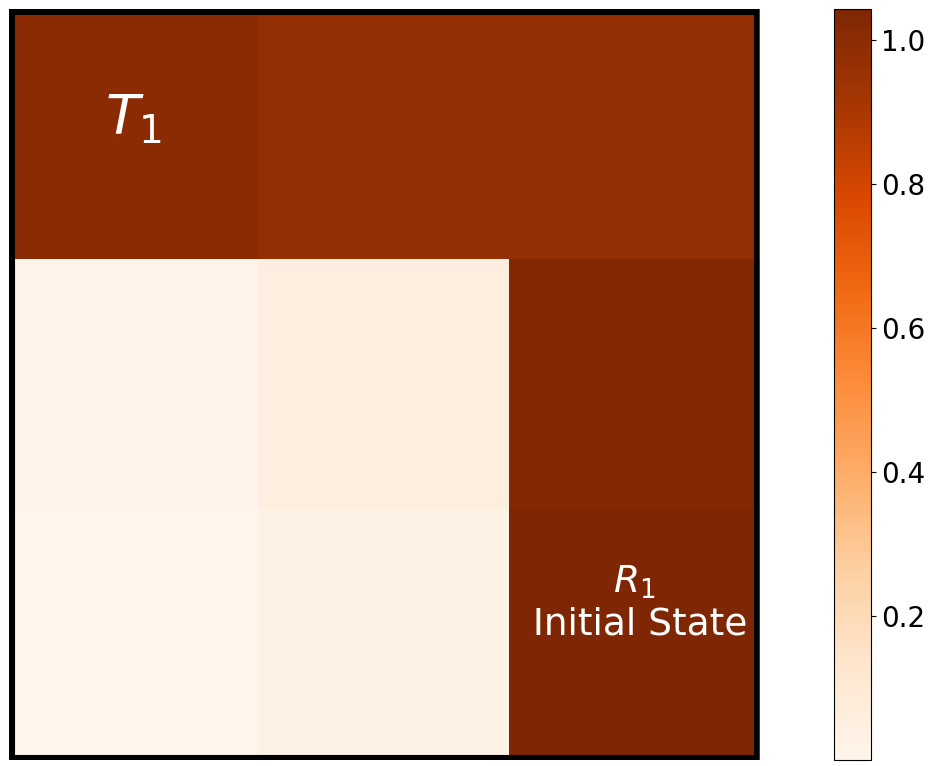}
         \caption{Occupancy measures \(\occupancyVar_{\mdpState^1}\) of Robot \(R_1\)}
         \label{fig:r1_total_corr_heatmap_three_agents}
     \end{subfigure}
     \hfill
     \begin{subfigure}[b]{0.32\textwidth}
         \centering
         \includegraphics[width=\textwidth]{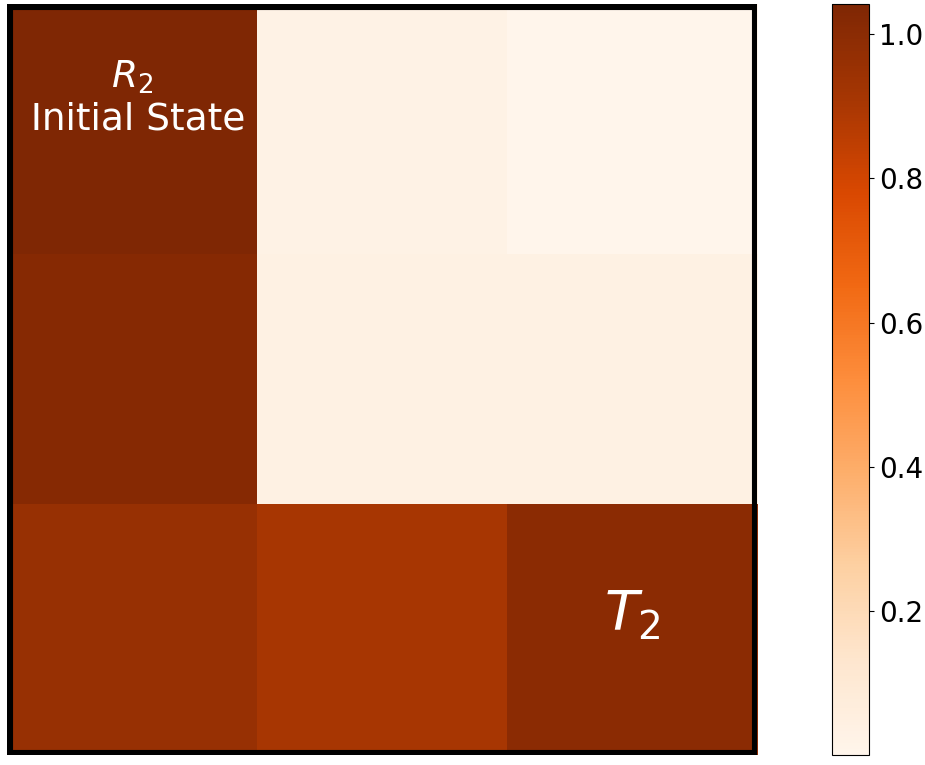}
         \caption{Occupancy measures \(\occupancyVar_{\mdpState^2}\) of Robot \(R_2\)}
         \label{fig:r2_total_corr_heatmap_three_agents}
     \end{subfigure}
     \hfill     
     \begin{subfigure}[b]{0.32\textwidth}
         \centering
         \includegraphics[width=\textwidth]{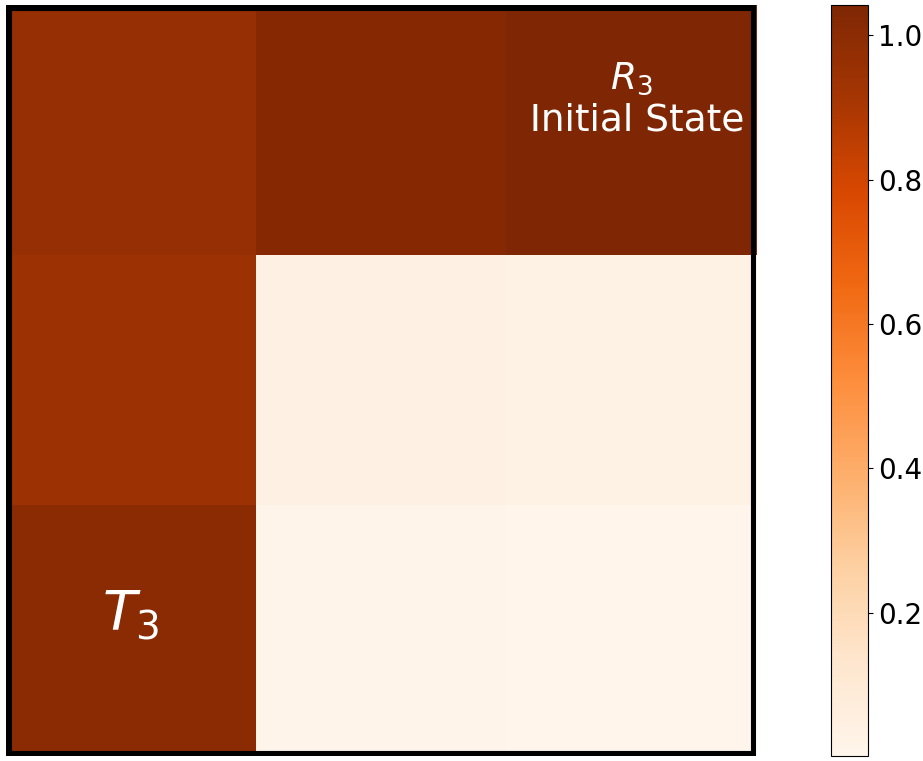}
         \caption{Occupancy measures \(\occupancyVar_{\mdpState^3}\) of Robot \(R_3\)}
         \label{fig:r1_total_corr_heatmap_three_agents}
     \end{subfigure}
    \caption{Occupancy measures for the individual agents under joint policy \(\policy_{MD}\) (the proposed policy).}
    \label{fig:occupancy_heatmap_total_corr_three_agents}
\end{figure}

\begin{figure}[h!]
    \centering
     \begin{subfigure}[b]{0.32\textwidth}
         \centering
         \includegraphics[width=\textwidth]{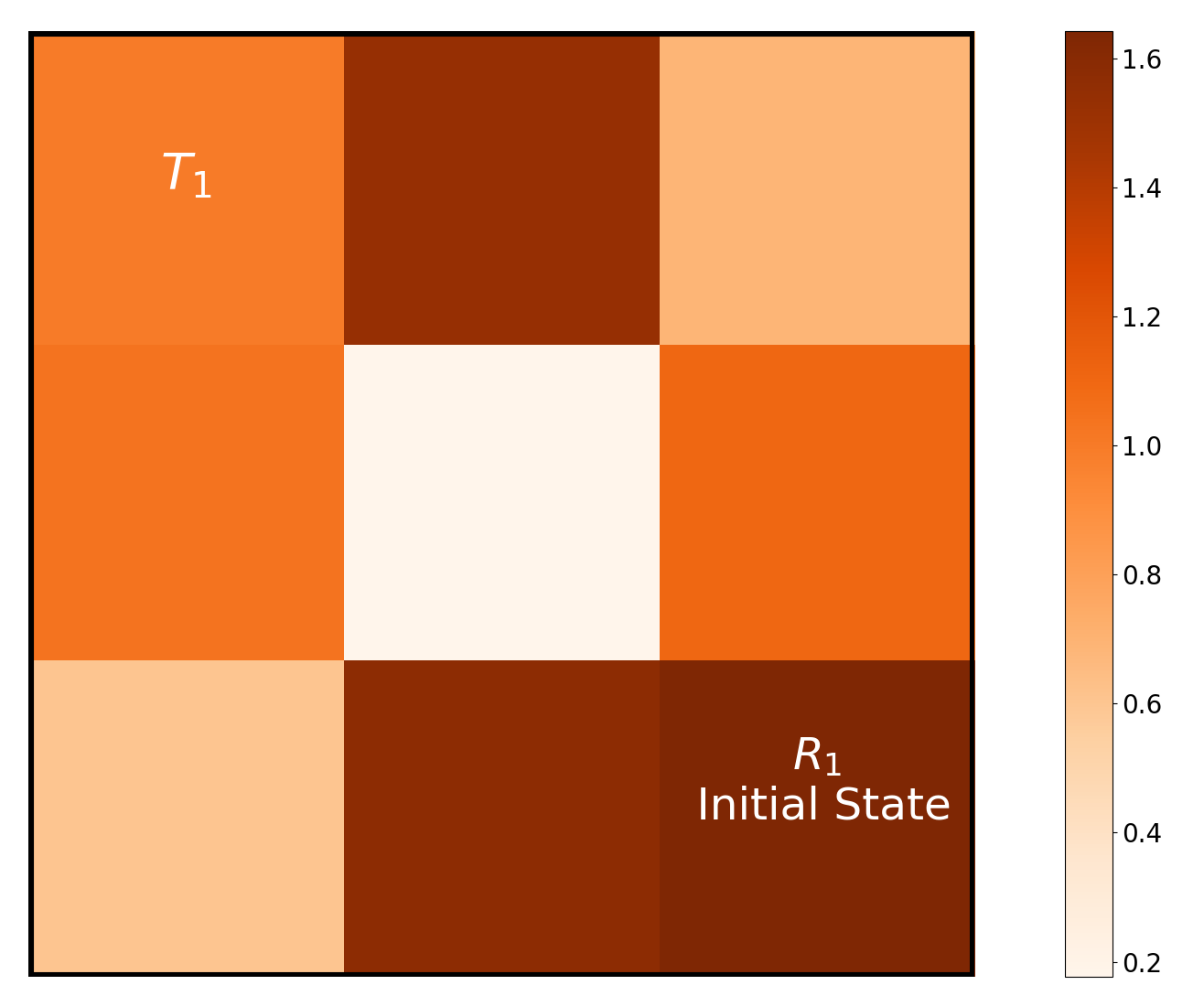}
         \caption{Occupancy measures \(\occupancyVar_{\mdpState^1}\) of Robot \(R_1\)}
         \label{fig:r1_reachability_heatmap_three_agents}
     \end{subfigure}
     \hfill
     \begin{subfigure}[b]{0.32\textwidth}
         \centering
         \includegraphics[width=\textwidth]{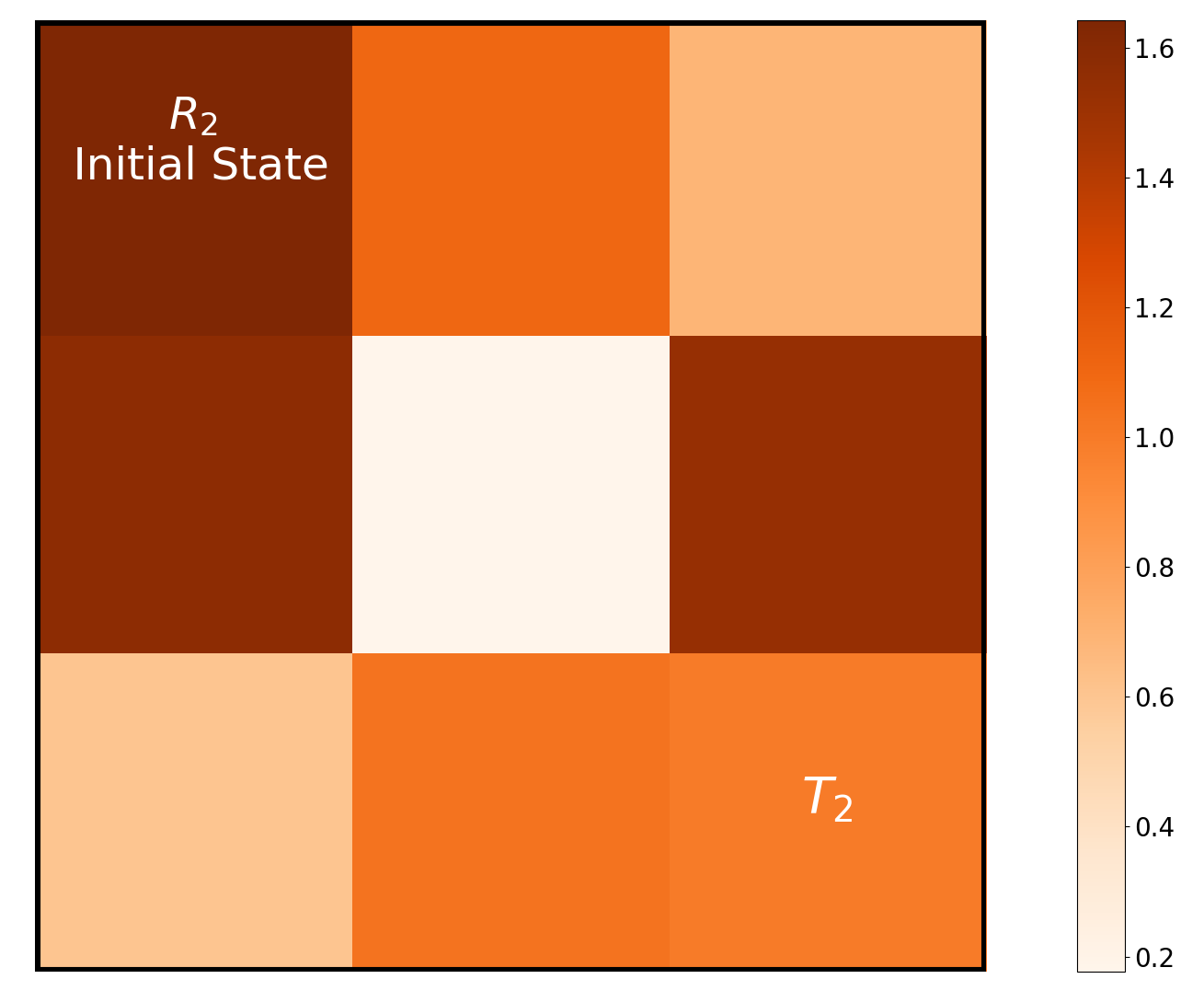}
         \caption{Occupancy measures \(\occupancyVar_{\mdpState^2}\) of Robot \(R_2\)}
         \label{fig:r2_reachability_heatmap_three_agents}
     \end{subfigure}
     \hfill
          \begin{subfigure}[b]{0.32\textwidth}
         \centering
         \includegraphics[width=\textwidth]{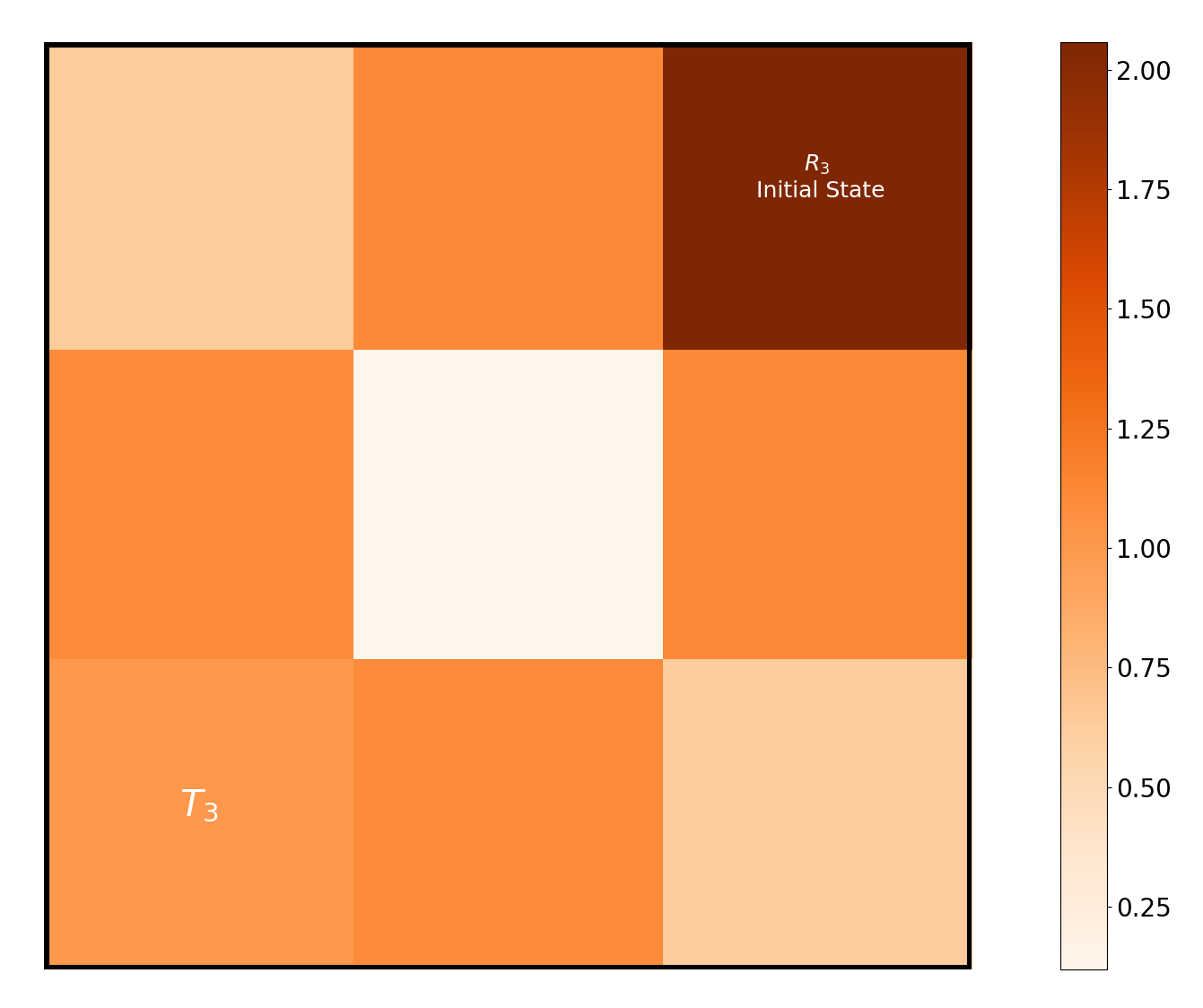}
         \caption{Occupancy measures \(\occupancyVar_{\mdpState^3}\) of Robot \(R_3\)}
         \label{fig:r1_reachability_heatmap_three_agents}
     \end{subfigure}
    \caption{Occupancy measures for the individual agents under joint policy \(\policy_{base}\).}
    \label{fig:occupancy_heatmap_reachability_three_agents}
\end{figure}

\end{document}